\documentclass{CSML}
\pdfoutput=1

\usepackage{lastpage}
\newcommand{\dOi}{14(1:17)2018}
\lmcsheading{}{1--\pageref{LastPage}}{}{}%
{Dec.~30, 2015}{Feb.~27, 2018}{}

\usepackage{mathrsfs}
\usepackage[USenglish]{babel}
\usepackage{theoremref} 
\usepackage{extarrows}
\usepackage{color}
\usepackage{xspace}
\usetikzlibrary{arrows,shapes,automata,backgrounds,decorations}
\usepackage{graphicx}
\usepackage{url}
\usepackage{xifthen}
\usepackage{hyperref}
\hypersetup{hidelinks}
\usepackage{wasysym,amsmath,mathtools,amssymb}

\usepackage{DynamicCausality}

\title{Dynamic Causality in Event Structures}
\author[Y. Arbach]{Youssef Arbach}
\address{\vskip-6pt}
\author[D. Karcher]{David S. Karcher}
\address{\vskip-6pt}
\author[K. Peters]{Kirstin Peters}
\address{\vskip-6pt}
\author[U. Nestmann]{Uwe Nestmann}
\address{Technische Universit\"at Berlin, Germany}
\email{youssef.arbach@tu-berlin.de}
\email{david.s.karcher@tu-berlin.de}
\email{kirstin.peters@tu-berlin.de}
\email{uwe.nestmann@tu-berlin.de}

\thanks{Supported by the DFG Research Training Group SOAMED}	


\begin{document}

\begin{abstract}
	Event Structures (ESs) address the representation of direct relationships between individual events, usually capturing the notions of causality and conflict. Up to now, such relationships have been static, \ie they cannot change during a system run. Thus, the common ESs only model a static view on systems.
	We make causality dynamic by allowing causal dependencies between some events to be changed by occurrences of other events. We first model and study the case in which events may entail the \emph{removal} of causal dependencies, then we consider the \emph{addition} of causal dependencies, and finally we combine both approaches in the so-called \emph{Dynamic Causality ESs}. For all three newly defined types of ESs, we study their expressive power in comparison to the well-known \emph{Prime} ESs, \emph{Dual} ESs, \emph{Extended Bundle} ESs, and ESs for \emph{Resolvable Conflicts}. Interestingly, Dynamic Causality ESs subsume Extended Bundle ESs and Dual ESs but are incomparable with ESs for Resolvable Conflicts.
\end{abstract}

\maketitle

\section{Introduction}
\textit{Motivation.}
Modern process-aware systems emphasize the need for flexibility into their design to adapt to changes in their environment \cite{Reichert:EnhanceFlexibility}. One form of flexibility is the ability to change the work-flow during the run-time of the system deviating from the default path, due to changes in
regulations or to exceptions.
Such changes could be ad hoc or captured at the build time of the system
\cite{Reichert:ForwardBackwardJumps}. For instance---as taken from
\cite{Reichert:EnhanceFlexibility}---\emph{during the treatment process, and for a particular patient, a planned computer tomography must not be performed due to the fact that she has a cardiac pacemaker. Instead, an X-ray activity shall be performed.}
In this paper, we provide a formal model that can be used for such scenarios, showing what is the \emph{regular} execution path and what is the \emph{exceptional} one \cite{Reichert:ForwardBackwardJumps}. In the Conclusions section, we highlight the advantages of our model over other static-causality models \wrt such scenarios.

\vspace*{0.5em}
\noindent
\textit{Concurrency Model.} 
In many formalisms for concurrent systems, as \eg in Petri nets or different process calculi, causality is a derived concept (\cite{Baldan2004129,VaraccaExtrusion}). As a consequence, adding dynamicity to the causality relation usually requires the duplication of modeled actions, \ie, transitions in Petri nets or parts of processes in process calculi. Event Structures (ESs) model causality more directly. They usually address statically defined relationships that constrain the possible occurrences of events, typically represented as \emph{causality} (for precedence) and \emph{conflict} (for choice). An event is a single occurrence of an action; it cannot be repeated. ESs were first used to give semantics to Petri nets \cite{Winskel:Thesis}, then to process calculi \cite{flowES,Langerak:Thesis}, and recently to model quantum strategies and games \cite{QuantumES}. The semantics of an ES itself is usually provided by the sets of traces compatible with the constraints, or by configuration-based sets of events, possibly in their partially-ordered variant (\emph{posets}).

\vspace*{0.5em}
\noindent
\textit{Overview.}
We study the idea---motivated by application scenarios---of events changing the causal dependencies of other events. 
In order to deal with dynamicity in causality, usually 
duplications of events are used (see \eg\cite{VaraccaExtrusion},
where copies of the same event have the same label, but different
dependencies). In this paper, we want to express dynamic changes of causality more directly without duplications. We allow dependencies to change  during a system run, by modifying the causality itself. In this way, we avoid duplications of events, and keep the model simple and more intuitive.
We separate the idea of dropping (shrinking) causality from adding (growing) causality and study each one separately first, and then combine them into so-called Dynamic Causality ESs (DCESs). 

\vspace*{0.5em}
\noindent
\textit{Example.}
\begin{figure}
	\centering
		\begin{tikzpicture}[node distance=3cm,>=stealth']
			\node (plow) [circle,fill,minimum size=4pt,inner sep=0pt, label=left:plow]{};
			\node (water) [below of=plow,circle,fill,minimum size=4pt,inner sep=0pt, label=left:water]{};
			\node (plant) [right of=plow,circle,fill,minimum size=4pt,inner sep=0pt, label=above:plant]{};
			\node (harvest) [right of=plant,circle,fill,minimum size=4pt,inner sep=0pt, label=above:harvest]{};
			\draw [->,thick] (plow) to node {} (plant);
			\draw [->,thick] (water) to node {} (plant);
			\draw [->,thick] (plant) to node {} (harvest);
			\node (rain) [right of=water,circle,fill,minimum size=4pt,inner sep=0pt, label=right:rain]{};
			\draw[dropping]	(1.5,-1.5) -- (rain);
			\node (pestin) [right of=harvest,circle,fill,minimum size=4pt,inner sep=0pt, label=right:pest infestation]{};
			\node (pestcon) [below of=pestin,circle,fill,minimum size=4pt,inner sep=0pt, label=right:pest control]{};
			
			\draw [->,dotted,thick] (pestcon) to node {} (harvest);
			\draw[adding]	(pestin) -- (7.5,-1.5);
		\end{tikzpicture}
 	\caption{An example Dynamic Causality ES (DCES). Events are represented as dots, dependencies as solid arrows. The fact that an event can insert or delete a dependency is represented by an arrow between the event and the dependency, with a filled arrowhead for insertion and an empty arrowhead for deletion. Initially absent dependencies, which may be added, are dotted.}
	\label{fig:DCESex}
\end{figure}
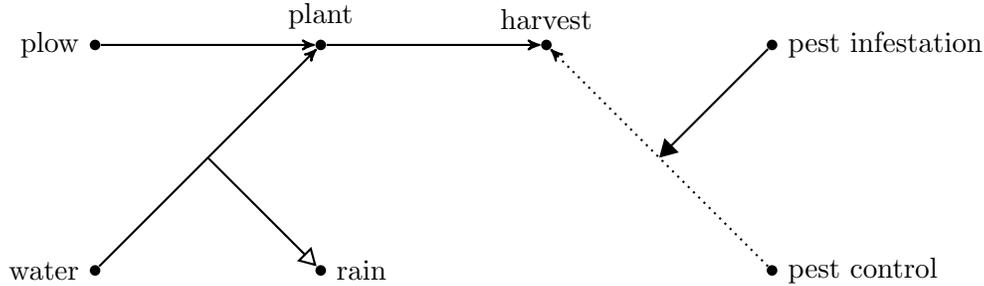
Figure \ref{fig:DCESex} presents an example DCES: In the regular work-flow, after plowing and watering, some crop can be planted and finally harvested. Exceptional behavior changes those dependencies: First, rain deletes the necessity of watering for planting. Second, a pest infestation inserts a new precondition|pest control|for harvesting. Note that pest control could also be done prophylactically, but it becomes mandatory after a pest infestation.

\vspace*{0.5em}
\noindent
\textit{Related Work.}
Kuske and Morin in \cite{localTraces} worked on local independence, using local traces.
There, actions can be independent from each other after a given history.
By contrast, in our work we provide a mechanism for independence of events through the growing and shrinking causality, while this related work abstracts from the way actions become independent.
In \cite{ResolvableConflict}, van Glabbeek and Plotkin introduced Event Structures for Resolvable Conflicts (RCESs), where conflicts can be resolved or created by the occurrence of other events. This dynamicity of conflicts is complementary to our approach.
As visualized in Figure~\ref{fig:landscape}, DCESs and RCESs are incomparable but---similarly to RCESs---DCESs are more expressive than many other types of ESs.
In \cite{DCRgraphs}, Hildebrandt \etal present a generalized version of ESs, DCR-Graphs, for similar purposes as discussed here. They can, \eg, express that a response is needed after an action. Causal dependencies, however, are static in their approach, \ie, the representation of dynamic changes---as present in systems with a regular execution path and some exceptional execution paths---requires the duplication of actions.

\vspace*{0.5em}
\noindent
\textit{Structure.}
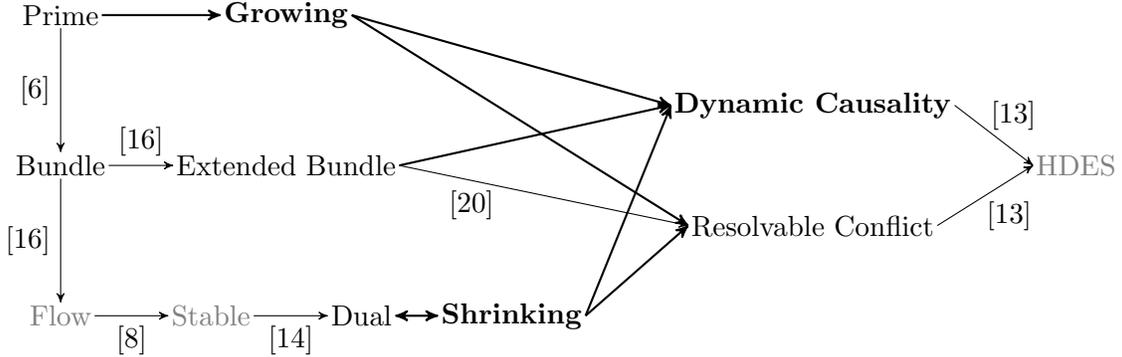
\begin{figure}
	\centering
	\begin{tikzpicture}[node distance=1.5cm,>=stealth']
		\node[inner sep=1pt] (CES) {Prime};
		\node[inner sep=1pt] (BES) [node distance=2cm,below of=CES] {Bundle};
		\node[inner sep=1pt] (EBES) [node distance=3cm, right of=BES] {Extended Bundle};
		\node[inner sep=1pt] (FLOW) [node distance=2cm, below of=BES] {\color{gray}Flow};
		\node[inner sep=1pt] (STABLE) [node distance=2cm, right of=FLOW] {\color{gray}Stable};
		\node[inner sep=1pt] (DES) [node distance=2cm, right of=STABLE] {Dual}; 
		\node[inner sep=1pt] (SES) [node distance=2cm, right of=DES] {\textbf{Shrinking}};
		\node[inner sep=1pt] (GES)[node distance=3cm, right of=CES] {\textbf{Growing}};
		\node (aux1) [node distance=7cm, right of=EBES] {};
		\node[inner sep=1pt] (RESCON) [node distance=0.8cm, below of=aux1] {Resolvable Conflict};
		\node[inner sep=1pt] (DCES) [node distance=0.8cm, above of=aux1] {\textbf{Dynamic Causality}};
		\node[inner sep=1pt] (HDES) [node distance=3.5cm, right of =aux1] {\color{gray}{HDES}};
		\draw [->, thick] (GES.east) to (DCES.west);	
		\draw [->, thick] (SES.east) to (DCES.west);	
		\draw [->] (CES) -- node [left] {\cite{FLOW}} ++ (BES);
		\draw [->] (BES) -- node [above] {\cite{Langerak:Thesis}} ++ (EBES);
		
		\draw [->] (BES) -- node [left] {\cite{Langerak:Thesis}} ++ (FLOW);
		\draw [->] (FLOW) -- node [below] {\cite{flowES}} ++ (STABLE);
		\draw [->] (STABLE) -- node [below] {\cite{Katoen:Thesis}} ++ (DES);
		\draw [->, thick] (SES.east) to (RESCON.west);
		\draw [->, thick] (EBES.east) to (DCES.west);
		\draw [->] (EBES.east) -- node [near start, below] {\cite{ResolvableConflict}} (RESCON.west);
		\draw [->, thick] (CES.east) to (GES.west);
		\draw [->, thick] (GES.east) to (RESCON.west);
		\draw [->] (DCES.east)  -- node [near end, above =4pt] {\cite{HDES}}  (HDES.west);
		\draw [->] (RESCON.east)  -- node [near end, below =4pt] {\cite{HDES}}  (HDES.west);
		\draw[<->, thick] (DES) to (SES);
	\end{tikzpicture}
	\vspace*{-1em}
	\caption{Landscape of Event Structures (newly defined ESs are bold)}
	\label{fig:landscape}
\end{figure}
In \S\ref{sec:techPre}, we define a relaxed version of \emph{Prime ESs} \cite{Winskel:IntroToES}, denoted as relaxed Prime ESs (rPESs), and recap the definitions for \emph{Bundle ESs} (BESs) \cite{Langerak:Thesis}, \emph{Extended Bundle ESs} (EBESs) \cite{Langerak:Thesis}, \emph{Dual ES} (DESs) \cite{Langerak97causalambiguity}, and \emph{ESs for Resolvable Conflicts} (RCESs) \cite{ResolvableConflict}.

In \S\ref{sec:Semantics}, we investigate the various kinds of semantic models that are used by the aforementioned ESs in order to compare them with respect to their expressive power.

In \S\ref{sec:SES}, we define \emph{Shrinking cau\-sa\-li\-ty Event Structures} (SESs); we show that SESs are strictly less expressive than RCESs, have the same expressive power as DESs, and thus are strictly more expressive than rPESs, PESs, BESs, \emph{Flow Event Structures} (FESs) \cite{flowES} and \emph{Stable Event Structures} (StESs) \cite{Winskel:IntroToES}, and are incomparable to EBESs.

In \S\ref{sec:GES}, we define \emph{Growing causality Event Structures} (GESs); we show that GESs are strictly less expressive than RCESs and strictly more expressive than rPESs and PESs.

In \S\ref{sec:DCES}, we combine both aforementioned concepts within the \emph{Dynamic Causality Event Structures} (DCESs) and show that they are strictly more expressive than EBESs, which are incomparable to SESs and GESs. 
Although RCESs are shown to be more expressive than GESs and SESs, they are incomparable with DCESs. To complete the picture, note that the \emph{set-based Higher order Dynamic causality Event Structures} (HDES) \cite{HDES} are then strictly more expressive than both RCESs and DCESs (\cf \S\ref{sec:HDES}).

We defer some of the more technical lemmata and some technical definitions to the Appendix.
The relations among the various classes of ESs are summarized in Figure~\ref{fig:landscape}, where an arrow from one class to another means that the first is less expressive than the second. Here, bold arrows indicate newly derived results, whereas thin arrows indicate results taken from literature and are augmented with the respective reference. The bold structures are newly defined in this paper and the gray ones are added to complete the picture but are discussed here only briefly.

In \S\ref{sec:conclusion}, we summarize the contributions and show the limitations of other static-causality models \wrt our example, and conclude with future work.

This paper is an extended version of \cite{dynamicCausality15}, with a special focus on the different type of used equivalences and an alternative, equivalent, and more intuitive transition definition for SESs and GESs. For the DCESs, we omit the old transition definition and only state the new one, which simplifies most proofs. We also omit the conflict relation from the definitions of the GESs and DCESs, as it can there be expressed as a derived concept.

\section{Technical Preliminaries}\label{sec:techPre}

We investigate the idea of dynamically evolving dependencies between events. Therefore, we want to allow for the occurrence of events to create new causal dependencies between events or to remove such dependencies. We base our extension on (a relaxed variant of) Prime ESs, because it represents a very simple causality model.
In the following, we shortly revisit the main definitions of the types of ESs from the literature that we compare with. ESs are sometimes augmented with labels for events, \eg, to relate several events and an action in another formalism (\cite{Langerak:Thesis, flowES}). Since our results are not influenced by the presence of labels, we restrict our attention to unlabeled ESs.

\subsection{Relaxed Prime Event Structures}

A Prime Event Structure (PES) \cite{Winskel:IntroToES} consists of a set of events and two relations describing conflicts and causal dependencies. To cover the intuition that events causally depending on an infinite number of other events can never occur, Winskel \cite{Winskel:IntroToES} requires PESs to satisfy the \emph{axiom of finite causes}.
Additionally, the enabling relation is assumed to be a partial order, \ie, it is reflexive, transitive, and anti-symmetric. Furthermore, the concept of \emph{conflict heredity} is required: an event~$a$ that is in conflict with another event~$b$ is also in conflict with all causal successors of~$b$.

If we allow causal dependencies to be added or dropped dynamically, it is hard to maintain the conflict heredity as well as transitivity and reflexivity of enabling.
Thus, we define \emph{relaxed Prime Event Structures} (rPES) where we omit the axiom of conflict heredity (as \eg in \cite{boudolCastellani88}) and do not require that enabling is a partial order.

% \begin{definition} [\cite{Winskel:IntroToES}]	\label{def:PES}
%   A \emph{Prime Event Structure (PES)} is a triple $ \pi = \left( E, \confOp, \enabOp \right) $, where $ E $ is a set of \emph{events}, $ \confOp \subseteq E^2 $ is an irreflexive symmetric relation (the \emph{conflict} relation), and $ \enabOpT \subseteq E^2 $ is the \emph{enabling} relation.
% \end{definition}
\begin{defi} 	\label{def:rPES}\label{def:PES}
  A \emph{relaxed Prime Event Structure (rPES)} is a triple $ \pi = \left( E, \confOp, \enabOp \right) $, where
  \begin{itemize}
  \item $ E $ is a set of so-called \emph{events},
  \item $ \confOp \subseteq E^2 $ is an irreflexive symmetric relation (the \emph{conflict} relation), and
  \item $ \enabOpT \subseteq E^2 $ is the \emph{enabling} relation.
  \end{itemize}
  
\end{defi}

\noindent
Note that, in comparison to PES, we also omit the finite causes property; its intention will instead be provided through the constraint on finite configurations (see, \eg, Definition~\ref{def:SEStrans}).
Note that rPESs have the same expressive power as PESs in~\cite{Winskel:IntroToES} \wrt finite configurations, since the axiom of finite causes is trivially satisfied for finite configurations and on the other hand the concept of conflict heredity does not influence the expressive power, but only ensures that syntactic and semantic conflicts coincide.

The computation state of a process that is modeled as a rPES is represented by the set of events that have occurred.
%Given a rPES $ \pi = \left( E, \confOp, \enabOp \right) $, we call such sets $ C \subseteq E $ that respect $ \confOp $ and $ \enabOp $ configurations of $ \pi $. 
Naturally, such sets are required to be consistent with the causality and conflict relations of the given rPES.

%\begin{defi} [\cite{Winskel:IntroToES}]
	%\label{def:pesConfigs}
	%Let $ \pi = \left( E, \confOp, \enabOp \right) $ be a PES.\\
	%A set of events $ C \subseteq E $ is a \emph{configuration} of $ \pi $ if it is \emph{conflict-free}, \ie $ \forall e, e' \in C \logdot \neg \left( \conf{e}{e'} \right) $, \emph{downward-closed}, \ie $ \forall e, e' \in E \logdot \enab{e}{e'} \land e' \in C \implies e \in C $, and the transitive closure of the enabling relation is \emph{acyclic}, \ie $ \enabOp^{\ast} \cap \; C^2 $ is free of cycles.
	%We denote the set of configurations of $ \pi $ by $ \configurations{\pi} $.
%\end{defi}
\begin{defi} 
  \label{def:rpesConfigs}\label{def:pesConfigs}
  Let $ \pi = \left( E, \confOp, \enabOp \right) $ be a rPES.\\
  A set of events $ C \subseteq E $ is a \emph{configuration} of $ \pi $ if it is
  \begin{itemize}
  \item \emph{conflict-free}, \ie,
    $ \forall e, e' \in C \logdot \neg \left( \conf{e}{e'} \right) $,
  \item \emph{downward-closed}, \ie,
    $ \forall e, e' \in E \logdot \enab{e}{e'} \land e' \in C
    \implies e \in C $, and 
  \item \emph{free of enabling cycles}, \ie, $\enabOp^\ast \cap \; C^2 $ is antisymmetric.
  \end{itemize}
  We denote the set of configurations of $ \pi $ by $ \configurations{\pi} $.\\
  For each $C\in\configurations{\pi}$, 
  we define the partial order 
  ${\leq_C}\mathrel{:=} {{(\enabOp^*)}\cap {C^2}}$.
  % $\leq_C$ as restriction of the enabling relation to $C^2$ by 
  % \{(a,b)\mid (a,b)\in{{\enabOp^*}\cap {C^2}}\}$.
  
\end{defi}

An event $ e $ is called \emph{impossible} in a rPES $\pi$ if it does not occur in any configuration of $\pi$. Events can be impossible because of (1)~enabling cycles, (2)~an overlap between the enabling and the conflict relation, or (3)~impossible predecessors.

In the relaxed version of the PESs, and in contrast to the original PES,  the following two properties do not hold due to impossible events and the lack of conflict heredity.

\begin{defi}[Fullness and faithfulness, \cite{DBLP:conf/icalp/Winskel82,DBLP:conf/litp/Boudol90}] \label{def:ESFullFaithful}
An ES is called \emph{full} if every event is possible, \ie, occurs in a configuration. 
An ES is called \emph{faithful} if the syntactic conflict relation $\confOp$ on possible events coincides with the semantic conflict relation $\confOp_s$ given by:
\[e \mathrel{\confOp_s} e' \iff \forall \text{ finite } C\logdot \{e,e'\}\not\subseteq C\]
\end{defi}

\begin{lem}\label{lma:rPESnotFullFaithful}
There are rPES that are not full or faithful.
\end{lem}

\begin{proof}
For fullness, consider $\pi=(\{a,b,c\},\emptyset,\{(a,b),(b,c),(c,a)\})$, a rPES with three events, no conflicts, and a dependency cycle. In $\pi$, we only have the empty set as a configuration; thus, $\pi$ is not full.

Now for faithfulness, consider $\pi'=(\{a,b,c\},\{(a,c)\},\{(a,b)\})$, a rPES with three events, a conflict between $a$ and $c$, and where $b$ depends on $a$ (this is the rPES in Fig.~\ref{fig:GESconflict} on page~\pageref{fig:GESconflict}). There is no syntactic conflict between $b$ and $c$, but there is a semantic conflict because any configuration containing $b$ and $c$ should both contain $a$, to be downward-closed, and not contain $a$, to be conflict-free; thus, $\pi'$ is not faithfull.
\end{proof}

The set of causes of an event $e$, following the reflexive and transitive closure of the enabling relation backwards, is denoted by $\rtCauses{e}$. Moreover, if $ D $ is a configuration and $ d \in D $, then $ \rtCausesConf{d}{D} $ is the restriction of $ \rtCauses{d} $ to events in $ D $.

\begin{defi}[\cite{ILARIA}]
	\label{def:rtCauses}
	Let $\pi=(E,\confOp,\enabOp)$ be a rPES and $e\in E$ then
	\[ \rtCauses{e} := \Set{ e'\in E \mid e' \enabOp^{\ast} e }. \]
	Let $D\in\configurations{\pi}$ and $d\in D$ then 
	\[ \rtCausesConf{d}{D} := \Set{d'\in D\mid d' \leq_D d}. \]
\end{defi}

For rPESs, the following lemmas hold.

\begin{lem}[Primality, \cite{ILARIA}]
  \label{lem:primality}
  Let $\pi= (E,\confOp,\enabOp)$ be a rPES and $e$ be possible in $\pi$. Then:
  \begin{itemize}
  \item $\lceil{e}\rceil_C=\lceil{e}\rceil$ for all $C\in\configurations{\pi}$ such that $e$ in $C$,
  \item $\lceil{e}\rceil\in \configurations{\pi}$, and 
  \item $\lceil{e}\rceil$ is the minimal configuration in $\configurations{\pi}$ containing $e$.
  \end{itemize}
\end{lem}
\begin{proof}
  The property $\lceil{e}\rceil_C=\lceil{e}\rceil$ follows from the fact that 
  ${\leq_C}\mathrel{:=} {{(\enabOp^*)}\cap {C^2}}$,  
  %$\leq_C:=\{(a,b)\mid (a,b)\in\enabOp^*\cap C^2\}$, 
  while the property $\lceil{e}\rceil\in \configurations{\pi}$ follows from the fact that there is at least one configuration $C\in\configurations{\pi}$ such that $e\in C$, and that $\lceil{e}\rceil$ inherits the required  properties from $C$. Minimality is implied by the fact that $\lceil{e}\rceil_C=\lceil{e}\rceil$ for any $C\in\configurations{\pi}$, since any configuration $C'\subset \lceil{e}\rceil_C$ would miss some event $e'\in \lceil{e}\rceil$ and thus would fail to satisfy $\lceil{e}\rceil_{C'}=\lceil{e}\rceil$.
\end{proof}
\begin{lem}[Stability, \cite{ILARIA}]
	\label{lem:stability}
	Let $\pi= (E,\confOp,\enabOp)$ be a rPES and 
        let $C_1,C_2\in\configurations{\pi}$ such that $(C_1\cup C_2)\in\configurations{\pi}$. 
        Then: 
        \begin{itemize}
        \item $(C_1\cap C_2)\in\configurations{\pi}$, and
        \item for all $e\in (C_1\cap C_2)$, 
          we have $\lceil{e}\rceil_{C_1\cap C_2}=\lceil{e}\rceil_{C_1}=\lceil{e}\rceil_{C_2}$.
        \end{itemize}

\end{lem}
\begin{proof}
	Assume $C_1,C_2, (C_1\cup C_2)\in\configurations{\pi}$. We start by showing that also $(C_1\cap C_2)\in\configurations{\pi}$. The properties that $(C_1\cap C_2)$ is conflict-free and ${\enabOp}\cap{(C_1\cap C_2)^2}$ is acyclic  follow from the analogue properties of $(C_1\cup C_2)$, while the downward-closure of $(C_1\cap C_2)$ follows  from the downward-closure of both $C_1$ and $C_2$. Therefore $(C_1\cap C_2)\in\configurations{\pi}$. Now, we can use primality to conclude  that $\lceil{e}\rceil_{C_1\cap C_2}=\lceil{e}\rceil_{C_1}=\lceil{e}\rceil_{C_2}$ for any $e\in (C_1\cap C_2)$.
\end{proof}

\subsection{Bundle, Extended Bundle and Dual Event Structures}
\label{sec:BES}
\label{sec:EBES}
\label{sec:DES}

The rPESs are simple but also limited. They do not allow one to describe an optional or conditional enabling of events. Bundle event structures (BESs)---among others---were designed to overcome these limitations \cite{Langerak:Thesis}.
\emph{Bundles} are pairs $ \left( X, e \right) $, denoted as $\buEn{X}{e}$, where $X$ is a set of events and $e$ is the event pointed by that bundle.
A bundle is satisfied when one event of $ X $ occurs. An event is enabled when all bundles pointing to it are satisfied. This \emph{disjunctive causality} allows for optionality in enabling events.

\begin{defi}[\cite{Langerak:Thesis}]
	\label{def:BES} 
	A \emph{Bundle Event Structure (BES)} is a triple $ \beta = \left( E, \confOp, \buEnOp \right) $, where
        \begin{itemize}
        \item $ E $ is a set of \emph{events}, 
        \item $ \confOp \subseteq E^2 $ is an irreflexive symmetric relation (the \emph{conflict} relation), and 
        \item $ \buEnOpT \subseteq {\Powerset{E} \times E} $ is the
          \emph{enabling} relation, satisfying the following \emph{stability} condition: \\
          % for all $ X \subseteq E $ and $ e \in E $ the bundle 
          $ \buEn{X}{e} $ implies that for all
          $ e_1, e_2 \in X $ with $ e_1 \neq e_2 $ we have
          $ \conf{e_1}{e_2} $.
        \end{itemize}

\end{defi}

\noindent
Figure~\ref{fig:exampleBES}(a) shows an example of a BES. The solid arrows denote causality, \ie, reflect the enabling relation, the line between the arrows indicates a bundle, and the dashed line denotes a mutual conflict.

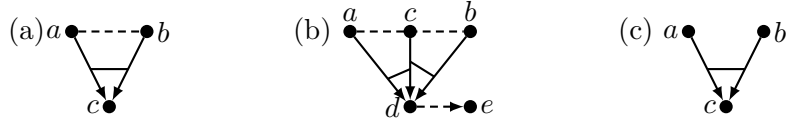
\begin{figure}[t]
	\centering
	\begin{tikzpicture}[bend angle=30]
		% figure a
		\event{b2}{0.3}{1.5}{left}{$ a $};
		\event{b3}{1.3}{1.5}{right}{$ b $};
		\event{b4}{0.8}{0.5}{left}{$ c $};
		\draw[enablingPES] (b2) edge (b4);
		\draw[enablingPES] (b3) edge (b4);
		\draw[thick] (0.55, 1) -- (1.05, 1);
		\draw[conflictPES] (b2) edge (b3);
		\node (a) at (-0.3, 1.5) {(a)};
		% figure b
		\event{a1}{4}{1.5}{above}{$ a $};
		\event{a2}{5.6}{1.5}{above}{$ b $};
		\event{a3}{4.8}{1.5}{above}{$ c $};
		\event{a4}{4.8}{0.5}{left}{$ d $};
		\event{a5}{5.6}{0.5}{right}{$ e $};
		\draw[enablingPES] (a1) edge (a4);
		\draw[enablingPES] (a2) edge (a4);
		\draw[enablingPES] (a3) edge (a4);
		\draw[thick]	(4.5, 0.87) -- (4.8, 1);
		\draw[thick]	(4.8, 1.1) -- (5.12, 0.9);
		\draw[conflictPES] (a1) edge (a2);
		\draw[conflictEBES] (a4) edge (a5);
		\node (c) at (3.5, 1.5) {(b)};
		% figure c
		\event{d1}{8.5}{1.5}{left}{$ a $};
		\event{d2}{9.5}{1.5}{right}{$ b $};
		\event{d3}{9}{0.5}{left}{$ c $};
		\draw[enablingPES] (d1) edge (d3);
		\draw[enablingPES] (d2) edge (d3);
		\draw[thick] (8.75, 1) -- (9.25, 1);
		\node (d) at (7.8, 1.5) {(c)};
	\end{tikzpicture}
	\vspace*{-1em}
	\caption{A Bundle ES, an Extended Bundle ES, and a Dual ES.}
	\label{fig:exampleBES}
\end{figure}

A configuration of a BES is again a conflict-free set of events that is
downward-closed. Therefore, the stability condition avoids causal
ambiguity \cite{Langerak97causalambiguity}. To exclude sets of events that result from enabling cycles, we use (event) traces (called \emph{proving sequences} in \cite{Langerak:Thesis}).
For a sequence $ t = e_1 \cdots e_n $ of events, let $ \overline{t} = \Set{
  e_1, \ldots, e_n } $ and $ t_i = e_1 \cdots e_i $ for
$1\leq{i}\leq{n}$. Let $ \epsilon $ denote the empty sequence.
Let $\B{e}\deff\Set{ X \subseteq E \mid \buEn{X}{e}}$.

\begin{defi}[\cite{Langerak:Thesis}]
  \label{def:BEStrace}
  Let $ \beta = \left( E, \confOp, \buEnOp \right) $ be a BES. \\
  A \emph{trace} is a sequence of distinct events $ t = e_1 \cdots e_n $ with $ \overline{t} \subseteq E $ that respects
  \begin{itemize}
  \item conflicts, \ie, $ \forall 1 \leq i, j \leq n \logdot \neg \left( \conf{e_i}{e_j} \right) $, and 
  \item bundle satisfaction, 
    \ie,
    $ \forall 1 \leq i \leq n \logdot 
    \forall X\in\B{e_i} \logdot
    \overline{t_{i - 1}} \cap X \neq \emptyset $.
  \end{itemize}
  A set of events $ C \subseteq E $ is a \emph{configuration} of $ \beta $ if there is a trace $ t $ such that $ C = \overline{t} $.
\end{defi}

\noindent
This trace-based definition of a configuration will be the same for Extended Bundle and Dual ESs.
Let $ \traces{\beta} $ denote the set of traces and $ \configurations{\beta} $ the set of configurations of $ \beta $.

A partially ordered set, or \emph{poset}, is a pair $ \left( A, \leq \right) $, where $ A $ is a finite set and $ \leq $ is a \emph{partial order} over $ A $.
Posets are used as a semantic model for several kinds of ESs and also other models of concurrency \cite{posetsForConfigurations}. 
For example, and in contrast to mere configurations, if $A$ is a set of \emph{events}, then the poset $ \left( A, \leq \right) $ does not only record which events have happened, but the order $\leq$ also captures their precedence relations.

A poset represents a set of system runs, differing for permutation of independent events. To describe the semantics of the entire ES, families of posets \cite{posetsForConfigurations} with a prefix relation are used. According to Rensink~\cite{posetsForConfigurations}, families of posets provide a convenient underlying structure for models of concurrency, and are at least as expressive as families of configurations.

To obtain the posets of a BES, we augment each of its configurations with a partial order.
Let $ \beta = \left( E, \confOp, \buEnOp \right) $ be a BES, $ C \in \configurations{\beta} $, and $ e,e' \in C $. Then $ e \prec_C e' $ holds if $ \exists X \subseteq E \logdot e \in X \land \buEn{X}{e'} $.
Let $ \leq_C $ be the reflexive and transitive closure of $ \prec_C $.
It is proved in \cite{Langerak:Thesis} that $ \leq_{C} $ is a partial order over $ C $. Let $\posets{\beta}$ denote the set of posets of $ \beta $.
Each linearization (obeying the defined precedence relations) of a given poset of a BES (or EBES) yields an event \emph{trace} of that structure (\cite{Langerak:Thesis}).

\vspace{0.5em}
\noindent
The first extension of BESs that we consider are \emph{Extended Bundle Event Structures} (EBESs) from \cite{Langerak:Thesis}. The conflict relation $\#$ is replaced by  a \emph{disabling} relation. An event $ e_1 $ disables another event $ e_2 $, meaning that the occurrence of $ e_1 $ precludes any subsequent occurrence of  $ e_2 $ afterwards.
The symmetric conflict $ \# $ can be modeled through mutual disabling.  
Therefore, EBESs are a generalization of BESs, and thus are more expressive \cite{Langerak:Thesis}.

\begin{defi}[\cite{Langerak:Thesis}]
	\label{def:EBES} 
	An \emph{Extended Bundle Event Structure (EBES)} is a triple $ \xi = \left( E, \disaOp, \buEnOp \right) $, where
        \begin{itemize}
        \item $ E $ is a set of \emph{events}, 
        \item $ \disaOpT \subseteq E^2 $ is
          the irreflexive \emph{disabling} relation, and
        \item $ \buEnOpT \subseteq \Powerset{E} \times E $ is the
          \emph{enabling} relation satisfying the following \emph{stability} condition: \\
          % for all $ X \subseteq E $ and $ e \in E $ the bundle 
          $ \buEn{X}{e} $ implies that for all
          $ e_1, e_2 \in X $ with $ e_1 \neq e_2 $ we have
          $ \disa{e_1}{e_2} $.
        \end{itemize}

\end{defi}

\noindent
Stability again ensures that two distinct events within a bundle set are in mutual disabling.
Figure~\ref{fig:exampleBES}(b) shows an EBES with the two bundles $ \buEn{\Set{ a, c }}{d} $ and $ \buEn{\Set{ b, c }}{d} $.
The dashed lines denote again mutual disabling as required by stability. A disabling $ \disa{d}{e} $, to be read ``$ e $ disables $ d $'', is represented by a dashed arrow.

\begin{defi} [\cite{Langerak:Thesis}]
  \label{def:EBESconf}
  Let $ \xi = \left( E, \disaOp, \buEnOp \right) $ be an EBES.\\
  A \emph{trace} is a sequence of distinct events $ t = e_1 \cdots
  e_n $ with $ \overline{t} \subseteq E $ that respects
  \begin{itemize}
  \item disabling, \ie,
    $ \forall 1 \leq i, j \leq n \logdot \disa{e_i}{e_j} \implies i <
    j $, and 
  \item bundle satisfaction, 
    \ie,
    $ \forall 1 \leq i \leq n \logdot 
    \forall X\in\B{e_i} \logdot
    \overline{t_{i - 1}} \cap X \neq \emptyset $.
        \end{itemize}

\end{defi}

We adapt the definitions of configurations and traces of BESs accordingly.
For $ C \in \configurations{\xi} $ and $ e, e' \in C $, let $e\prec_C e'$ 
if $ \disa{e}{e'} $ or
%$ \exists X \subseteq E \logdot e \in X \land \buEn{X}{e'} $ 
$ \exists X\in\B{e'} \logdot e \in X $.
Again $ \leq_C $ denotes the reflexive and transitive closure of $ \prec_C $, and $\posets{\xi}$ denotes the set of posets of $\xi$.

\vspace{0.5em}
\noindent
\emph{Dual Event Structures} (DESs) are obtained
by dropping the stability condition of BESs. This leads to causal ambiguity, \ie,
given a trace and one of its events, it is not always possible to determine
which other events caused this event.
The definition of DESs differs between \cite{Katoen:Thesis} (based on EBESs) and \cite{Langerak97causalambiguity} (based on BESs).
Here, we rely on~\cite{Langerak97causalambiguity}.

\begin{defi} [\cite{Langerak97causalambiguity}]
	\label{def:DES}
	A \emph{Dual Event Structure (DES)} is a triple $ \delta = \left( E, \confOp, \buEnOp \right) $, where
        \begin{itemize}
        \item $ E $ is a set of \emph{events}, 
        \item $ \confOp \subseteq E^2 $ is an irreflexive symmetric relation (the \emph{conflict} relation), and 
        \item $ \buEnOpT \subseteq \Powerset{E} \times E $ is the \emph{enabling} relation.
        \end{itemize}
\end{defi}

\noindent
Figure~\ref{fig:exampleBES}(c) shows a DES with one bundle, namely $ \buEn{\Set{ a, b}}{c} $. It is a relaxed version of the BES in Figure~\ref{fig:exampleBES}(a) since the stability condition (which enforced a conflict between $a$ and $b$) is dropped, and there is no conflict.

The definitions of traces and configurations are relaxed accordingly.
\begin{defi}[\cite{Langerak97causalambiguity}]
  \label{def:DESconf}
  Let $ \delta = \left( E, \confOp, \buEnOp \right) $ be a DES. \\
  A \emph{trace} is a sequence of distinct events $ t = e_1 \cdots e_n $ 
  with $ \overline{t} \subseteq E $ that respects
  \begin{itemize}
  \item conflicts, \ie,
    $ \forall 1 \leq i, j \leq n \logdot \neg \left( \conf{e_i}{e_j}
    \right) $, and
  \item bundle satisfaction, 
    \ie,
    $ \forall 1 \leq i \leq n \logdot 
    \forall X\in\B{e_i} \logdot
    \overline{t_{i - 1}} \cap X \neq \emptyset $.
  \end{itemize}
\end{defi}

Because of the causal ambiguity, the definition of $ \leq_C $ is difficult and the behavior of a DES \wrt a configuration cannot be described by a single poset anymore.
In \cite{Langerak97causalambiguity}, the authors illustrate various possible interpretations of causality. The authors defined five different intensional posets: liberal, bundle satisfaction, minimal, early and late posets. 
They show the equivalence of the behavioral semantics, and that the early causality and trace equivalence coincide. Thus, we concentrate on early causality in the following. The remaining intensional partial order semantics are discussed in Appendix \ref{app:partialOrderSemantics}. 

To capture causal ambiguity, we have to consider all traces of a configuration to obtain its posets.
In essence, the approach defines causes of events~$e$ in a given trace~$t$ as not necessarily uniquely defined sets~$U$ of events. 
Posets are then derived by $c \prec e$ for $c\in{U}$ for all events in the trace and their possible causes as reflexive-transitive closure of $\prec$.

\begin{defi}[\cite{Langerak97causalambiguity}]
  \label{def:DESposets}
  Let $ \delta = \left( E, \confOp, \buEnOp \right) $ be a DES, 
  let $ t = e_1 \cdots e_n $ be one of its traces.
%, let $ 1 \leq i \leq n $, and 
%  let $ \buEn{X_1}{e_i}, \ldots, \buEn{X_m}{e_i} $ be all bundles pointing to $ e_i $.
  
  For two sets $ U_1,U_2 $ of events of $ t $, we refer to $ U_1 $ as \emph{earlier} than $ U_2 $, if the largest index in $ U_1 \setminus U_2 $ is smaller than the largest index in $ U_2 \setminus U_1 $.
  
  A set $U\subseteq{E}$ is a \emph{cause} of $ e_i $ in $ t $ (\ie, for $ 1 \leq i \leq n $), if 
  \begin{itemize}
  \item $ \forall e \in U \logdot \exists\, 1 \leq j < i \logdot e = e_j $, 
  \item $ \forall X\in\B{e_i} \logdot X \cap U \neq \emptyset $, and
  % \item $ \forall 1 \leq k \leq m \logdot X_k \cap U \neq \emptyset $, and
  \item  $ U $ is the earliest set satisfying the previous two conditions.
  \end{itemize}
  Let $ \posetsEar{t} $ be the set of posets obtained this way for $ t $.
\end{defi}

Note that, 
% for two sets of events of $ t $, $ U_1 $ is earlier than $ U_2 $ if the largest index in $ U_1 \setminus U_2 $ is smaller than the largest index in $ U_2 \setminus U_1 $ \cite{Langerak97causalambiguity}.
% Also note that 
for BESs, EBESs, and DESs, families of posets are the most discriminating semantics studied in the literature. So, in these cases, we consider two ESs as behaviorally equivalent if they have the same set of posets.

\begin{defi}
	\label{def:posetsEq}
	Let $ \mu_1,\mu_2 $ be either of BESs, EBESs, or DESs. They are called \emph{poset equivalent}, written $ \posetsEq{\mu_1}{\mu_2} $ if $ \mu_1 $ and $ \mu_2 $ have the same set of posets.
\end{defi}

\subsection{Event Structures for Resolvable Conflicts}

Event Structures for Resolvable Conflicts (RCES) were introduced in \cite{ResolvableConflict} to generalize former types of ESs and to give semantics to general Petri Nets. They allow us to model the case where $ a $ and $ c $ cannot occur together until $ b $ takes place, \ie, initially $ a $ and $ c $ are in conflict, and they stay in conflict until the occurrence of $ b $ resolves it.
An RCES consists of a set of events and an \emph{enabling} relation between sets of events. The latter serves to provide witnesses for transitions between configurations $X$ and $Y$: for each subset $Z$ of  configuration $Y$, there must be a witnessing subset $W$ of the preceding configuration $X$, where $W$ enables $Z$. The enabling relation also allows us to implicitly model conflicts between events, as we will see in the example below.

\begin{defi} [\cite{ResolvableConflict}]
	\label{def:RCES}
	An \emph{Event Structure for Resolvable Conflicts (RCES)} is a pair $ \rho = \left( E, \enrcD \right) $, where $ E $ is a set of \emph{events} and $ {\enrcOp} \subseteq \Powerset{E}^2 $ is the \emph{enabling relation}.
\end{defi}

In \cite{ResolvableConflict}, several versions of configurations are defined. Here, we consider only configurations that are both reachable and finite.

\begin{defi} [\cite{ResolvableConflict}]
  \label{def:ResConfTrans}
  \label{def:ResConfConf}
  Let $ \rho = \left( E, \enrcOp \right) $ be an RCES and $ X, Y \subseteq E $.
  Then:
  \begin{align*}
    \transRC[\rho]{X}{Y} \iff( X \subset Y \land \forall Z \subseteq Y \logdot \exists W \subseteq 	X \logdot W \vdash Z )
  \end{align*}
  If no confusion is possible, we sometimes omit the subscript~$\rho$.
  
  The set of \emph{(reachable) configurations} of $ \rho $ is defined as 
  \begin{displaymath}
    \configurations{\rho} = 
    \Set{ X \subseteq E \mid
      \emptyset \mathrel{\transRCOp[\rho]^*} X 
      \land X \text{ is finite} } 
\end{displaymath}
  where $ \transRCOp[\rho]^* $ is the reflexive and transitive closure of $ \transRCOp[\rho] $. 
\end{defi}

\noindent
Note the difference between $ \subset $ and $ \subseteq $ in the above definition.

As an example, consider the RCES $ \rho = \left( E, \enrcOp \right) $, where $ E = \Set{ a, b, c } $, with $ \enrc{\Set{ b }}{\Set{ a, c }} $, and with $ \enrc{\emptyset}{X} $ iff $ X \subseteq E $ and $ X \neq \Set{ a, c } $.
It models the initial conflict between $ a $ and $ c $ that can be resolved by $ b $.
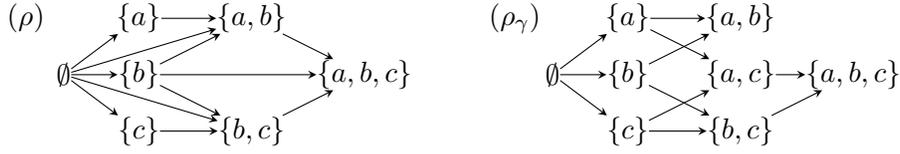
\begin{figure}[t]
	\centering
	\begin{tikzpicture}
		\node[config] (c1) at (0, 0.75) {$ \emptyset $};
		\node(l1) at (-0.5,1.5) {$(\rho)$};
		\node[config] (c2) at (1, 1.5) {$ \Set{ a } $};
		\node[config] (c4) at (1, 0.75) {$ \Set{ b } $};
		\node[config] (c3) at (1, 0) {$ \Set{ c } $};
		\node[config] (c5) at (2.5, 1.5) {$ \Set{ a, b } $};
		\node[config] (c6) at (2.5, 0) {$ \Set{ b, c } $};
		\node[config] (c7) at (4, 0.75) {$ \Set{ a, b, c } $};
		\draw[-stealth] (c1) -- (c2);
		\draw[-stealth] (c1) -- (c3);
		\draw[-stealth] (c1) -- (c4);
		\draw[-stealth] (c1) -- (c5);
		\draw[-stealth, bend right=45] (c1) -- (c6);
		\draw[-stealth] (c2) -- (c5);
		\draw[-stealth] (c3) -- (c6);
		\draw[-stealth] (c4) -- (c5);
		\draw[-stealth] (c4) -- (c6);
		\draw[-stealth] (c4) -- (c7);
		\draw[-stealth] (c5) -- (c7);
		\draw[-stealth] (c6) -- (c7);
		
		\node[config] (empt) at (6.5, 0.75) {$ \emptyset $};
		\node (l2) at (6,1.5) {$(\rho_\gamma)$};
		\node[config] (b) at (7.5, 0.75) {$ \Set{b} $};
		\node[config] (a) at (7.5, 1.5) {$ \Set{a} $};
		\node[config] (c) at (7.5, 0) {$ \Set{c} $};
		\node[config] (ab) at (9, 1.5) {$ \Set{a,b} $};
		\node[config] (ac) at (9, 0.75) {$ \Set{a,c} $};
		\node[config] (bc) at (9, 0) {$ \Set{b,c} $};
		\node[config] (abc) at (10.5, 0.75) {$ \Set{a, b, c} $};
		\draw[-stealth] (empt) -- (a);
		\draw[-stealth] (empt) -- (b);
		\draw[-stealth] (empt) -- (c);
		\draw[-stealth] (a) -- (ab);
		\draw[-stealth] (a) -- (ac);
		\draw[-stealth] (b) -- (ab);
		\draw[-stealth] (b) -- (bc);
		\draw[-stealth] (c) -- (ac);
		\draw[-stealth] (c) -- (bc);
		\draw[-stealth] (ac) -- (abc);
		\draw[-stealth] (bc) -- (abc);
 	\end{tikzpicture}
 	\vspace*{-1em}
 	\caption{Transition graphs of RCESs with resolvable conflict ($\rho$) and disabling ($\rho_\gamma$).}
 	\label{fig:transgraphs}
\end{figure}
In Figure~\ref{fig:transgraphs}($ \rho $), the respective transition graph is shown, \ie, the nodes are all reachable configurations of $ \rho $ and the directed edges represent $ \transRCOp[\rho] $. Note that, because of $\{a,c\}\subset\{a,b,c\}$ and $\emptyset \not\vdash \Set{ a, c }$, there is no transition from $\emptyset$ to $\{a,b,c\}$.

We consider two RCESs as equivalent if they have the same transition graphs. Note that, since we consider only \emph{reachable} configurations, the tran\-si\-tion equi\-va\-len\-ce defined below is denoted as \emph{reachable} transition equivalence in \cite{ResolvableConflict}.

\begin{defi} [\cite{ResolvableConflict}]\label{def:transeqrc}
  % Let 
  % $ \rho_1 = \left( E_1, \enrcOp_1 \right) $ and
  % $ \rho_2 = \left( E_2, \enrcOp_2 \right) $ be RCESs.
  % , and $\transRCOp[1]$ and $\transRCOp[2]$ their associated transition relations.
  Two RCESs 
  $ \rho_1 = \left( E_1, \enrcOp_1 \right) $ and
  $ \rho_2 = \left( E_2, \enrcOp_2 \right) $ 
  are called \emph{(reachable) transition equivalent}, 
  written $ \transEq{\rho_1}{\rho_2}$, 
  if $ E_1 = E_2 $ 
  and ${\transRCOp[\rho_1] \cap \left( \configurations{\rho_1} \right)^2} 
  = {\transRCOp[\rho_2] \cap \left( \configurations{\rho_2} \right)^2} $.
\end{defi}

For RCESs, transition equivalence is the most discriminating sensible semantics studied in the literature, so we consider two RCESs \emph{behaviorally equivalent} if they have the same reachable transition graphs. 
In analogy with poset equivalence, we may lift this notion to compare arbitrary types of ESs that are equipped with a transition relation.

\begin{defi}
	\label{def:transEq}
	Let $ \mu_1,\mu_2 $ be of either type of ESs with transition relation. They are called \emph{transition equivalent}, written $ \transEq{\mu_1}{\mu_2} $, if they have the same reachable transition~graphs.
\end{defi}

\section{Semantics of Event Structures}
\label{sec:Semantics}

We observe that, due to the different expressive power of the aforementioned types of event structures, different kinds of semantic models are used to compare elements of the same type. Relaxed Prime ESs are compared \wrt\ finite configurations. BESs, EBESs, and DESs are compared \wrt\ families of posets. RCESs are compared \wrt\ their transition graphs. In order to build a hierarchy later on, we relate these semantical models. Fortunately, the three semantical models only grow with the expressive power of the respective ESs. In other words, these three semantical models form a hierarchy on their own with respect to their power to discriminate event structures.

\subsection{Posets}

Now, we naturally extend the definition of posets to rPESs. Let $ \pi = \left( E, \confOp, \enabOp \right) $ be a rPES, $ C \in \configurations{\pi} $, and $ e, e' \in C $. Then $ e \prec_C e' $ if $ \enab{e}{e'} $. Again $ \leq_C $ is the reflexive and transitive closure of $ \prec_C $. Then let $ \posets{\pi} $ denote the sets of posets of $ \pi $. By Definition~\ref{def:posetsEq}, $ \posetsEq{\pi_1}{\pi_2} $ tells us that $ \posets{\pi_1} = \posets{\pi_2} $. Note that for rPESs, there is no need to distinguish between different definitions of posets as done in \cite{Langerak97causalambiguity} for DESs. Accordingly, posets in rPESs are a direct consequence of the enabling relation. Since the effect of the enabling relation is already captured by the set of configurations, the consideration of posets does not add discriminating power in this case.

\begin{lem}
	\label{lem:confToPosets}
	Let $ \pi_1 = \left( E_1, \confOp_1, \enabOp_1 \right) $ and $ \pi_2 = \left( E_2, \confOp_2, \enabOp_2 \right) $ be rPESs.\\
	Then $ \configurations{\pi_1} = \configurations{\pi_2} \Longleftrightarrow \posetsEq{\pi_1}{\pi_2} $.
\end{lem}

\begin{proof}[Proof \cite{ILARIA}]
	Let $ \leq_C^i $ be the reflexive and transitive closure of $ \prec_C $ \wrt\ $ \pi_i $. Thus, $ \leq_C^i $ is the reflexive and transitive closure of $ \enabOp_i $ restricted on events in $ C $.
	
	The implication $ \posetsEq{\pi_1}{\pi_2} \Longrightarrow \configurations{\pi_1} = \configurations{\pi_2} $ is straightforward, since, by Definition, $ P \in \posets{\pi_i} $ if and only if $ P = \left( C, \leq_C^i \right) $ for some $ C \in \configurations{\pi_i} $.
	
	We prove the reverse implication $ \configurations{\pi_1} = \configurations{\pi_2} \Longrightarrow \posetsEq{\pi_1}{\pi_2} $. Assume $ \configurations{\pi_1} = \configurations{\pi_2} $. Suppose that in some $ C \in \configurations{\pi_i} $, the two partial order relations $ \leq_C^i $ are different. Without loss of generality, we may assume that this is because for some $ e, e' \in C $, we have $ e \leq_C^1 e' $ and $ e \not\leq_C^2 e' $.
	Let $ \rtCausesConf{e'}{C}^i = \Set{ e'' \in C \mid e'' \leq_C^i e' } $ (see Definition~\ref{def:rtCauses}). Then $ e \in \rtCausesConf{e'}{C}^1 $ and $ e \notin \rtCausesConf{e'}{C}^2 $.
	Let now $ \rtCauses{e'}^1 = \Set{ e'' \in E_1 \mid e'' \enabOp_1^{*} e' } $. By Lemma~\ref{lem:primality} (primality), $ \rtCausesConf{e'}{C}^1 = \rtCauses{e'}^1 $ and $ \rtCauses{e'}^1 = \rtCausesConf{e'}{C'}^1 $ for any configuration $ C' \in \configurations{\pi_1} $ that contains $ e' $. This implies $ \rtCausesConf{e'}{C}^2 \notin \configurations{\pi_1} $, which contradicts the hypothesis $ \configurations{\pi_1} = \configurations{\pi_2} $.
\end{proof}

%\begin{proof}
%	The definition of posets is based on configurations. Hence whenever $ \posetsEq{\pi_1}{\pi_2} $ then also $ \configurations{\pi_1} = \configurations{\pi_2} $.
%	
%	Assume $ \configurations{\pi_1} = \configurations{\pi_2} $. Thus the events of both $ \pi_1 $ and $ \pi_2 $ contain all events that occur in configurations---and possibly additional impossible events that are not necessarily contained in both event structures. Since impossible events have no influence on the sets of posets, we ignore them in the following. Let us also assume that $ \pi_1 $ and $ \pi_2 $ do not have the same enabling relation. Without loss of generality assume $ \enab{e}{e'} $ holds in $ \pi_1 $ but not in $ \pi_2 $ for some $ C \in \configurations{\pi_1} $ such that $ e, e' \in C $. Let $ C' \in \configurations{\pi_1} $ be the smallest configuration containing $ e $ and $ e' $, \ie there is no other event except $ e' $ that is contained in $ C' $ and that causally depends on $ e $. Then there exists a configuration $ C'' = C' \setminus \Set{ e } $ such that $ C'' \in \configurations{\pi_2} $, because $ \enab{e}{e'} $ does not hold in $ \pi_2 $. But, because $ \enab{e}{e'} $ holds in $ \pi_1 $, $ C'' \notin \configurations{\pi_1} $. That contradicts the assumption that $ \configurations{\pi_1} = \configurations{\pi_2} $. Thus, because of $ \configurations{\pi_1} = \configurations{\pi_2} $, the enabling relation of $ \pi_1 $ and $ \pi_2 $ is the same with respect to all not impossible events. Then, by the definition of posets, also $ \posetsEq{\pi_1}{\pi_2} $.
%\end{proof}

Since the definitions of posets coincide with respect to early, liberal, bundle satisfaction, minimal, and late causality, this result holds with respect to posets defined by either of these notions.
Also note that, as shown in \cite{Langerak97causalambiguity}, posets are more discriminating for DESs with respect to liberal, bundle satisfaction, minimal, and late causality; but coincide with sets of configurations for BESs and DESs, where posets are defined based on early causality.

\subsection{Transition Graphs}

For a transition-based ES with a few additional properties, there is a natural embedding into RCESs.

\begin{defi}
  \label{def:translationIntoRCES}
  Let $ \mu $ be an ES with event set $E$ and transition relation $ \transOp $ %defined on configurations 
  that is
  \begin{itemize}
  \item \emph{strict}, \ie, $ \trans{X}{Y} $ implies $ X \subset Y $, and
  \item \emph{dense}, \ie, $ X' \subseteq X \subset Y \subseteq Y' $ implies that $ \trans{X'}{Y'} \implies \trans{X}{Y} $
  \end{itemize}
  for all configurations $ X, Y, X', Y' $ of $ \mu $.\\
  Then, $ \rces{\mu} := \left( E, \Set{ ({X},{Z}) \mid \exists Y \subseteq E \logdot \trans{X}{Y} \land Z \subseteq Y } \right) $.
\end{defi}

% By Definition~\ref{def:transEq}, $ \transEq{\mu_1}{\mu_2} $ holds for two ESs that satisfy these Conditions, if they have the same transition graphs.
% Note that the SESs and GESs as defined in the following sections satisfy the Conditions (1) and (2). 

\noindent
By Definition~\ref{def:RCES}, the resulting structure $ \rces{\mu} $ is a RCES.

We show that $ \rces{\mu} $ is transition equivalent to $ \mu $.

\begin{lem}
	\label{lma:TransInRCES}
	Let $ \mu $ be an ES that is strict and dense according to Definition~\ref{def:translationIntoRCES}.\\ 
	Then, 
        $ {\transRCOp[\rces{\mu}]}={\transRCOp[\mu]}$.
\end{lem}

\begin{proof}  
  Assume $ \trans[\mu]{X}{Y} $. Then, by Definition~\ref{def:translationIntoRCES}, $ X \subset Y $ and $ \enrc{X}{Z} $ for all $ Z \subseteq Y $. Then, by Definition~\ref{def:ResConfTrans}, also $ \transRC[\rces{\mu}]{X}{Y} $.
  
  Assume $ \transRC[\rces{\mu}]{X}{Y} $. Then, by Definition~\ref{def:ResConfTrans}, $ X \subset Y $ and there is $ X' \subseteq X $ such that $ \enrc{X'}{Y} $.
  By the construction of $ \rces{\cdot} $ in Definition~\ref{def:translationIntoRCES}, there is $ Y' $ such that $ \trans[\mu]{X'}{Y'} $ and $ Y \subseteq Y' $. Thus, $ X' \subseteq X \subset Y \subseteq Y' $ and $ \trans[\mu]{X'}{Y'} $. Then, as $\transOp[\mu]$ is dense, also $ \trans[\mu]{X}{Y} $.
\end{proof}

Next, we show that there is a natural way for|PESs, BESs, EBESs, and DESs|to derive a transition relation from the subset relation between posets.
% Using Lemma~\ref{lma:TransInRCES}, we can embed rPESs, BESs, EBESs, and DESs into RCESs in order to compare them \wrt\ transition graphs. 
% \textcolor{red}{Yes, we can. But we do not \emph{need} the embedding into RCES, if we define $\transOp$ directly.}
% To this aim, we derive a suitable (\ie, strict and dense) transition relation from the subset relation between posets.

\begin{defi}
	\label{def:translationIntoRCES2}
	Let $ \mu $ be a rPES, BES, EBES, or DES with event set $E$.\\
        We define the transition relation $ {\mapsto} \subseteq \Powerset{E}^2 $ as follows:\\
	$ \trans{X}{Y} $ if there are two posets $ \left( X, \leq_1 \right) $ and $ \left( Y, \leq_2 \right) $ of $ \mu $ such that $ X \subset Y $ and ${\leq_1}\subseteq{\leq_2}$.
        % $ \Set{ \left( a, b \right) \mid a, b \in X \wedge a \leq_1 b } \subseteq \Set{ \left( a, b \right) \mid a, b \in Y \wedge a \leq_2 b } $.
\end{defi}

By construction, the above-defined transition relation is strict and dense as required in Definition~\ref{def:translationIntoRCES}.
Accordingly, we associate each rPES, BES, EBES and DES with a corresponding transition equivalent RCES using the Definitions~\ref{def:translationIntoRCES} and \ref{def:translationIntoRCES2}.

\begin{lem}
  \label{lem:posetsToTransitions}
  Let $ \mu_1 $ and $ \mu_2 $ be either of type rPES, BES, EBES, or DES, respectively.\\
  Then, $ \posetsEq{\mu_1}{\mu_2} \Longleftrightarrow \transEq{\mu_1}{\mu_2} $.
\end{lem}

\begin{proof}
  Assume $ \posetsEq{\mu_1}{\mu_2} $. Then, $ \mu_1 $ and $ \mu_2 $ have the same posets and thus the same subset relations between these posets. Then, by Definition~\ref{def:translationIntoRCES2}, $ \mu_1 $ and $ \mu_2 $ are associated with the same transition relation.
  % ; both satisfying the Conditions~(1) and (2) of Definition~\ref{def:translationIntoRCES}.
  % Then, by Lemma~\ref{lma:TransInRCES}
  Immediately, we get $ \transEq{\mu_1}{\mu_2} $.
  
  Assume $ \transEq{\mu_1}{\mu_2} $. Then, $ \mu_1 $ and $ \mu_2 $ have the same transitions between sets of configurations. By Definition~\ref{def:ResConfConf}, $ \mu_1 $ and $ \mu_2 $ then have the same configurations and traces. As shown in \cite{Langerak97causalambiguity}, this implies that $ \posetsEq{\mu_1}{\mu_2} $ for all ESs without causal ambiguity (\ie, for rPESs, BESs, EBESs, DESs, and for DESs with posets based on early causality).
\end{proof}

If posets are based on liberal, bundle satisfaction, minimal, or late causality, then the implication $ \posetsEq{\mu_1}{\mu_2} \implies \transEq{\mu_1}{\mu_2} $ still holds, because its proof does not depend on the kind of posets that are used.

\section{Shrinking Causality}
\label{sec:SES}

Now, we add a new relation that represents the removal of causal dependencies as a ternary relation between events $ \shrinkingCausality \subseteq E^3 $. For instance, $ \left( d, c, t \right) \in \shrinkingCausality $, also written as \drops{d}{c}{t}, models that the cause $ c $ is dropped from the set of causal predecessors of the target $ t $ by the occurrence of the so-called \emph{dropper} event $ d $. By definition, we ensure $d\notin\{c,t\}$; this is different \wrt\ the version in \cite{dynamicCausality15}.

The dropping is visualized in Figure~\ref{fig:SESExample} by an arrow with an empty head from the dependency $ \enab{c}{t} $ to its dropper $ d $.
We add this relation to rPESs and denote the result as \emph{shrinking-causality} event structures.

\begin{defi}
	\label{def:SES}
	A \emph{Shrinking Causality Event Structure (SES)} is a pair $ \sigma = \left( \pi, \shrinkingCausality \right) $, where $ \pi = \left( E, \confOp, \enabOp \right) $ is a rPES and $ \shrinkingCausality \subseteq E^3 $ is the \emph{shrinking causality} relation such that, for all $ c,d,t \in E $, \drops{d}{c}{t} implies $ \enab{c}{t} $ and $d\notin\{c,t\}$.
\end{defi}

\noindent
Sometimes we flatten $ \left( \pi, \shrinkingCausality \right) $ into $ \left( E, \#, \rightarrow, \shrinkingCausality \right) $.
For \drops{m}{c}{t}, we call $ m $ the \emph{modifier}, $ t $ the \emph{target}, and $ c $ the \emph{cause}. We denote the set of all modifiers dropping $ c \rightarrow t $ by \droppers{c}{t}.
We refer to the set of \emph{dropped causes} of an event \wrt\ a specific history by the function $ \dcD : \Powerset{E} \times E \to \Powerset{E} $ defined as: $ \dc{H}{e} = \Set{ e' \mid \exists d \in H \logdot \drops{d}{e'}{e} } $. 
We refer to the \emph{initial causes} of an event by the function $ \icD :E \to \Powerset{E} $ such that: $ \ic{e} = \Set{ e' \mid \enab{e'}{e} } $.

The semantics of a SES can be defined based on posets similar to BESs, EBESs, and DESs, or based on a transition relation similar to RCESs. We consider both.

\begin{defi}
  \label{def:SEStraceDef}
  \label{def:SESposets}
  \label{def:SEStrans}
  \label{def:SESconf}
  Let $ \sigma = \left( E, \confOp, \enabOp, \shrinkingCausality \right) $ be a SES.
  
  \begin{enumerate}
  \item
    
    A \emph{trace} of $ \sigma $ is a sequence of distinct events $ t = e_1 \cdots e_n $ with $ \overline{t} \subseteq E $ such that 
    \begin{itemize}
    \item $\forall 1 \leq i, j \leq n \logdot \neg \left( \conf{e_i}{e_j} \right)$ and
    \item $\forall 1 \leq i \leq n \logdot \left( \ic{e_i} \setminus \dc{\overline{t_{i - 1}}}{e_i} \right) \subseteq \overline{t_{i - 1}}$.
    \end{itemize}
    Then, $ C \subseteq E $ is a \emph{trace-based configuration} of $ \sigma $, if there is a trace $ t $ with $ C = \overline{t} $.
    $ \configTraces{\sigma} $ denotes the set of trace-based configurations, and $ \traces{\sigma} $ the set of traces of~$ \sigma $.
    
  \item 

    Let $ t = e_1 \cdots e_n \in \traces{\sigma} $ and $ 1 \leq i \leq n $. 
    A set $ U $ is a cause of $ e_i $ in $ t $ if 
    \begin{itemize}
    \item $\forall e \in U \logdot \exists 1 \leq j < i \logdot e = e_j $,
    \item $\left( \ic{e_i} \setminus \dc{\overline{t_{i - 1}}}{e_i} \right) \subseteq U$,  
      and
    \item $ U $ is the earliest set satisfying the previous two conditions.
    \end{itemize}
    Let $ \posetsSES{t} $ be the set of posets obtained this way for $ t $.
    
  \item
    
    Let $ X, Y \subseteq E $. 
    Then, $ \transS{X}{Y} $ iff 
    \begin{itemize}
    \item $X \subset Y$,
    \item $\forall e, e' \in Y \logdot \neg ( \conf{e}{e'} )$, and
    \item $\forall e \in Y{\setminus}X \logdot \left( \ic{e} \setminus \dc{X}{e} \right) \subseteq X$.
    \end{itemize}
    
  \item
    
    The set of all configurations of $ \sigma $ is 
    $ \configurations{\sigma} := \Set{ X \subseteq E \mid \emptyset \! \transSOp^* \! X \land X \text{ is finite} } $, \\
    where $ \transSOp^* $ is the reflexive and transitive closure of $ \transSOp $.

  \item

    For $X\subseteq Y\subseteq E$, we define the set of dropped dependencies as
    \[\dropped{X,Y}:=\left\{(c,t)\mid\exists d\in Y\setminus X\logdot \drops{d}{c}{t}  \right\}.\]

  \end{enumerate}
\end{defi}

\begin{figure}[t]
	\centering
	\begin{tikzpicture}
		\event{a1}{-.5}{0.75}{left}{$ c $};
		\event{a2}{.5}{0.75}{right}{$ t $};
		\event{a3}{-.5}{0}{left}{$ d $};
%		\node (a) at (-1.2, 1) {$ (a) $};
		\draw[enablingPES] (a1) edge (a2);
		\draw[dropping]	(0, 0.75) -- (a3);
	\end{tikzpicture}
	\vspace*{-1em}
	\caption{Dropper visualization in SESs.}
	\label{fig:SESExample}
\end{figure}
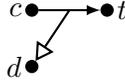

\noindent
We use the Definitions~\ref{def:posetsEq} and \ref{def:transEq} to obtain $ \posetsEq{}{} $ and $ \transEq{}{} $ for SESs.
The combination of initial and dropped causes ensures that for each $ e_i \in \overline{t} $, all its initial causes are either preceding $ e_i $ or dropped by other events preceding $ e_i $. Note that, as for DESs, we concentrate on early causality.
We consider the reachable and finite configurations \wrt $ \transSOp $ as well as configurations based on the traces. Note that both definitions coincide. (See Lemma \ref{lem:SESconf})

\subsection{An Alternative Transition Definition}

In the following, we provide an alternative transition definition for SESs. Thereafter, we prove the equivalence of both definitions.

\begin{defi}\label{def:SEStransalt}
  Let $ \sigma = \left( E, \confOp, \incaus, \shrinkingCausality \right) $ be a SES and
  $X,Y\subseteq E$.
  
  Then, $\transSa{(X,\caus{X})}{(Y,\caus{Y})}$ if 
  \begin{enumerate}
  \item $X\subset Y$, 
  \item $ \forall e, e' \in Y \logdot \neg ( \conf{e}{e'} ) $,
  \item
    \begin{enumerate}
    \item $ \forall e \in Y {\setminus} X \logdot \left\{e'\mid
        (e',e)\in{\caus{X}}\right\} \subseteq X $, and
    \item ${\caus{Y}} := {{\caus{X}}{\setminus}\dropped{X,Y}}$.
    \end{enumerate}
  \end{enumerate}

\end{defi}

\noindent
This definition stresses the \emph{state} of the causality relation due to dynamic shrinking: all of the current causality information is explicitly accessible, instead of hidden in the $ \dcD $ function. In order to prove the equivalence of both transition definitions, the following lemma provides an explicit, non-recursive definition of the current state of the causality relation.

\begin{lem}\label{lma:SEScaus}
	Let $ \sigma = \left( E, \confOp, \incaus, \shrinkingCausality \right) $ be a SES and $(X,\caus{X})$ be a state of $\sigma$.\\
	Then, ${\caus{X}} = {{\incaus}\setminus{\dropped{\emptyset,X}}}$.
\end{lem}

\begin{proof}
  Let $\emptyset=:C_0,C_1,\dots,C_{n-1},C_n:=X$ be such that
  $\transSa{(C_i,\caus{C_i})}{(C_{i+1},\caus{C_{i+1}})}$ for all $1\leq{i}<{n}$.
  % $\transSa{(C_0,\incaus)}{(C_1,\caus{C_1})} \dots \transSa{(C_i,\caus{C_i})}{(C_{i+1},\caus{C_{i+1}})} \dots\transSa{(C_{n-1},\caus{C_{n-1}})}{(X,\caus{X})}$, 
  By Definition~\ref{def:SEStransalt}, we get 
  $ {\caus{X}} = (\dots(\incaus \setminus \dropped{C_0,C_1}) \setminus \dots ) \setminus \dropped{C_{n-1},X} $, and 
  by basic set theory and  Definition of \ref{def:SEStrans}(5), 
  we have $ {\caus{X}} = {\incaus} \setminus \dropped{\emptyset,X} $. 
\end{proof}

Now we prove that both transition definitions of SESs coincide.

\begin{thm}
	Let $ \sigma = \left( E, \confOp, \incaus, \shrinkingCausality \right) $ be a SES and $ X, Y \subseteq E $. Then \[\transS{X}{Y}\iff\transSa{(X,\caus{X})}{(Y,\caus{Y})}.\]
\end{thm}

\begin{proof}
  Since the first conditions in the transition definitions of the Definitions~\ref{def:SEStrans} and \ref{def:SEStransalt} coincide, we show $\forall e\in Y\setminus X\logdot (\left( \ic{e} \setminus \dc{X}{e} \right)= \left\{e'\mid (e',e)\in{\caus{X}}\right\} )$. This follows from definitions of $\dcD$ and $\icD$ and Lemma~\ref{lma:SEScaus}:
  \begin{align*}
    \ic{e} \setminus \dc{X}{e}&\\
    =&\left\{e'\mid (e',e)\in{\incaus}\right\}\setminus\left\{c\mid\exists d\in X\logdot\drops{d}{c}{e}\right\}\\%
    =&\left\{e'\mid (e',e)\in{\incaus} \wedge\, e'\notin \left\{c\mid\exists d\in X\logdot\drops{d}{c}{e}\right\} \right\}\\
    =&\left\{e'\mid (e',e)\in{\incaus} \wedge\, (e',e)\notin \left\{(c,e)\mid\exists d\in X\logdot\drops{d}{c}{e}\right\} \right\}\\
    =&\left\{e'\mid (e',e)\in{\incaus}\setminus\left\{(c,t)\mid\exists d\in X\logdot \drops{d}{c}{t}\right\}\right\}\\
    =&\left\{e'\mid (e',e)\in{\incaus}\setminus\dropped{\emptyset,X}\right\} \\
    =&\left\{e'\mid (e',e)\in{\caus{X}}\right\}
  \end{align*}
  By Lemma~\ref{lma:SEScaus}, we have ${\caus{Y}} = {{\incaus}\setminus{\dropped{\emptyset,Y}}}$ and ${\caus{X}} = {{\incaus}\setminus{\dropped{\emptyset,X}}}$. We can split $\dropped{\emptyset,Y}$ into $\dropped{\emptyset,X}$ and $\dropped{X,Y}$ and get $ {\caus{Y}} = {\caus{X}} \setminus {\dropped{X,Y}} $. Thus, $ \transS{X}{Y} \iff \transSa{(X,\caus{X})}{(Y,\caus{Y})} $.
\end{proof}

\subsection{SESs versus DESs}

We show that SESs are as expressive as DESs
% (and thereby \textcolor{red}{*Does this make sense?*} show that $ \posetsEqOp $ and $ \transEqOp $ coincide on SESs).
by the definition of mutual encodings that result in structures with equivalent behaviors.

Consider the shrinking causality \drops{d}{c}{t}. It models the case that initially $ t $ causally depends on $ c $, which can be dropped by the occurrence of $ d $. Thus, for $ t $ to be enabled, either $ c $ occurs or $ d $ does. This is a disjunctive causality as modeled by DESs. In fact, \drops{d}{c}{t} corresponds to the bundle $ \buEn{\Set{ c, d }}{t} $. We prove that we can map each SES into a DES with the same behavior and vice versa.
To translate a SES into a DES, we create a bundle for each initial causal dependence and add all its droppers to the bundle set.

\begin{defi}
  \label{def:SESintoDES}
  Let $ \sigma = \left( E, \confOp, \enabOp, \shrinkingCausality \right) $ be a SES.
  \\
  Then, $ \des{\sigma} := \left( E, \confOp, \buEnOp \right) $, where $ \buEn{S}{y} $ iff 
  \begin{itemize}
  \item $ S \subseteq E $, 
  \item $ y \in E $, and
  \item $ \exists x \in E \logdot \enab{x}{y} \land S = (\Set{x} \cup \droppers{x}{y}) $.
  \end{itemize}
\end{defi}

We use Definition \ref{def:SESintoDES} to show that for each SES there is a DES with exactly the same traces and configurations.

\begin{lem}
  \label{lem:SEStoDES}
  Let $\sigma$ be a SES.
  Then, $\des{\sigma}$ is a DES. 
\end{lem}

\begin{proof}
  Let $ \sigma = \left( E, \confOp, \enabOp, \shrinkingCausality \right) $ be a SES. By Definitions~\ref{def:SES} and \ref{def:PES}, $ \confOp \subseteq E^2 $ is irreflexive and symmetric. Hence, by Definitions~\ref{def:DES} and \ref{def:SESintoDES}, $ \delta = \des{\sigma} $ is a DES.
\end{proof}

\begin{lem}
  \label{lem:SEStoDES-Semantik}
  Let $\sigma$ be a SES.
  Then,
  % For each SES $ \sigma $ there is a DES $ \delta $, namely $ \delta = \des{\sigma} $, such that 
  $ \traces{\sigma} = \traces{\des{\sigma}} $ and
  $ \configurations{\sigma} = \configurations{\des{\sigma}} $.
\end{lem}

\begin{proof}
  Let $ \sigma = \left( E, \confOp, \enabOp, \shrinkingCausality \right) $ be a SES. 
  Let $ t = e_1 \cdots e_n $.

  By Definition~\ref{def:SEStraceDef}, $ t \in \traces{\sigma} $ iff $ \overline{t} \subseteq E $ and $ \neg \left( \conf{e_i}{e_j} \right) $ and also it holds $ ( \ic{e_i} \setminus \dc{\overline{t_{i - 1}}}{e_i} ) \subseteq \overline{t_{i - 1}} $ for all $ 1 \leq i, j \leq n $.
	Since $ \dc{H}{e} = \{ e' \mid \exists d \in H \logdot \drops{d}{e'}{e} \} $ and $ \ic{e} = \Set{ e' \mid \enab{e'}{e} } $, we have $ \left( \ic{e_i} \setminus \dc{\overline{t_{i - 1}}}{e_i} \right) \subseteq \overline{t_{i - 1}} $ iff $ \forall e' \in E \logdot \enab{e'}{e_i} \implies e' \in \overline{t_{i - 1}} \lor \exists d \in \overline{t_{i - 1}} \logdot \drops{d}{e'}{e_i} $ for all $ 1 \leq i \leq n $.
	By Definition~\ref{def:SESintoDES}, then $ t \in \traces{\sigma} $ iff $ \overline{t} \subseteq E $, $ \neg \left( \conf{e_i}{e_j} \right) $, and $ \buEn{X}{e_i} \implies \overline{t_{i - 1}} \cap X \neq \emptyset $ for all $ 1 \leq i, j \leq n $ and all $ X \subseteq E $.
	Hence, by the definition of traces in \S\ref{sec:DES}, $ t \in
	\traces{\sigma} $ iff $ t \in \traces{\delta} $, \ie $ \traces{\sigma} = \traces{\delta} $.
	
	Because of $ \configTraces{\sigma} = \configurations{\sigma} $ (\cf Lemma~\ref{lem:SESconf}), \S\ref{sec:DES}, and Definition~\ref{def:SEStraceDef}, then also $ \configurations{\delta} = \configTraces{\sigma} = \configurations{\sigma} $.
\end{proof}

The encoding also preserve posets.

\begin{thm}
	\label{thm:SESintoDES}
        Let $ \sigma $ be a SES.
        Then, $ \posetsEq {\sigma }{\des{\sigma}} $.
\end{thm}

\begin{proof}
	Let $ \sigma = \left( E, \confOp, \enabOp, \shrinkingCausality \right) $ be a SES.
	By Lemmas~\ref{lem:SEStoDES} and~\ref{lem:SEStoDES-Semantik}, $ \delta = \des{\sigma} = \left( E, \confOp, \buEnOp \right) $ is a DES such that $ \traces{\sigma} = \traces{\delta} $ and $ \configurations{\sigma} = \configurations{\delta} $.
	Let $ t = e_1 \cdots e_n \in \traces{\sigma} $, $ 1 \leq i \leq n $, and the bundles $ \buEn{X_1}{e_i}, \ldots, \buEn{X_m}{e_i} $ all bundles pointing to $ e_i $. For $ U $ to be a cause for $ e_i $, Definition~\ref{def:SESposets} requires $ ( \ic{e_i} \setminus \dc{U}{e_i} ) \subseteq U $. Since $ \dc{H}{e} = \Set{ e' \mid \exists d \in H \logdot \drops{d}{e'}{e} } $ and $ \ic{e} = \{e'\mid\enab{e'}{e}\} $, this condition holds iff the condition $ \enab{e'}{e_i} \implies e' \in U \lor \exists d \in U \logdot \drops{d}{e'}{e_i} $ holds for all $ e' \in E $.
	By Definition~\ref{def:SESintoDES}, then $ ( \forall 1 \leq k \leq n \logdot X_k \cap U \neq \emptyset ) \iff ( ( \ic{e_i} \setminus \dc{U}{e_i} ) \subseteq U ) $.
	So, by Definitions~\ref{def:DESposets} and \ref{def:SESposets}, $ \posetsEq {\sigma }{\delta}$.
\end{proof}

In the opposite direction, we map each DES into a set of similar SESs such that each SES in this set has the same behavior as the DES.
Intuitively, we have to choose an initial dependency for each bundle and to translate the remainder of the bundle set into droppers for that dependency. Unfortunately, the bundles that point to the same event are not necessarily disjoint. Consider for example $ \buEn{\Set{ a, b }}{e} $ and $ \buEn{\Set{ b, c }}{e} $. If we choose $ \enab{b}{e} $ as initial dependency for both bundles to be dropped as \drops{a}{b}{e} and \drops{c}{b}{e}, then $ \Set{ a, e } $ is a configuration of the resulting SES but not of the original DES. So, we have to ensure that we choose distinct events as initial causes for all bundles pointing to the same event.
Therefore, for each bundle $ \buEn{X_i}{e} $, we choose a fresh event $ x_i $ as initial cause $ \enab{x_i}{e} $, make it impossible by a self-loop $ \enab{x_i}{x_i} $, and add all events $ d $ of the bundle $ X_i $ as droppers \drops{d}{x_i}{e}.
Note that, in order to translate a DES into a SES, we have to introduce additional events, \ie, it is not always possible to translate a DES into a SES without additional impossible events (\cf Lemma~\ref{lem:DESninSES}).

\begin{defi}
  \label{def:DESintoSES}
  Let $ \delta = \left( E, \confOp, \buEnOp \right) $ be a DES, 
  $ \Set{ X_i }_{i \in I} $ an enumeration of its bundles, and 
  $ \Set{ x_i }_{i \in I} $ a set of fresh events, \ie $ \Set{ x_i }_{i \in I} \cap E = \emptyset $.
  Then, $ \ses{\delta} := \left( E', \confOp, \enabOp, \shrinkingCausality \right) $ with 
  \begin{itemize}
  \item $ E' := E \cup \Set{ x_i }_{i \in I} $,
  \item $ \enabOpT := \Set{ \enab{x_i}{e} \mid \buEn{X_i}{e} } \cup \Set{ \enab{x_i}{x_i} \mid i \in I } $, and
  \item $ \shrinkingCausality := \Set{ \drops{d}{x_i}{e} \mid d \in X_i \land \buEn{X_i}{e} } $.
  \end{itemize}

\end{defi}

\noindent
Because the $ x_i $ are fresh, there are no droppers for the self-loops $ \enab{x_i}{x_i} $ in $ \ses{\delta} $. So the encoding ensures that all events in $ \Set{ x_i }_{i \in I} $ remain impossible forever in the resulting SES. In fact, we show that the DES and its encoding have the very same traces and configurations.

\begin{lem}
	\label{lem:DEStoSES}
	For each DES $ \delta $ there is a SES $ \sigma $, namely $ \sigma = \ses{\delta} $, such that $ \traces{\delta} = \traces{\sigma} $ and $ \configurations{\delta} = \configurations{\sigma} $.
\end{lem}

\begin{proof}
	Let $ \delta = \left( E, \confOp, \buEnOp \right) $ be a DES. By Definition~\ref{def:DES}, $ \confOp \subseteq E^2 $ is irreflexive and symmetric. Hence, by the Definitions~\ref{def:SES}, \ref{def:PES}, and \ref{def:DESintoSES}, $ \sigma = \ses{\delta} = \left( E', \confOp, \enabOp, \shrinkingCausality \right) $ is a SES.	
	
	Let $ t = e_1 \cdots e_n $.	Then, by Definition~\ref{def:SEStraceDef}, $ t \in \traces{\sigma} $ iff $ \overline{t} \subseteq E $ and $ \neg ( \conf{e_i}{e_j} ) $ and also $ ( \ic{e_i} \setminus \dc{\overline{t_{i - 1}}}{e_i} ) \subseteq \overline{t_{i - 1}} $ for all $ 1 \leq i, j \leq n $. Note that we have $ \overline{t} \subseteq E $ instead of $ \overline{t} \subseteq E' $, because all events in $ t $ have to be distinct and for all events in $ E' \setminus E $ there is an initial self-loop but no dropper.
	Since $ \dc{H}{e} = \Set{ e' \mid \exists d \in H \logdot \drops{d}{e'}{e} } $ and $ \ic{e} = \Set{ e' \mid \enab{e'}{e} } $, $ \left( \ic{e_i} \setminus \dc{\overline{t_{i - 1}}}{e_i} \right)$ $ \subseteq \overline{t_{i - 1}} $ iff $ \forall e' \in E \logdot \enab{e'}{e_i} \implies e' \in \overline{t_{i - 1}} \lor \exists d \in \overline{t_{i - 1}} \logdot \drops{d}{e'}{e_i} $ for all $ 1 \leq i \leq n $.
	By Definition~\ref{def:DESintoSES}, then $ t \in \traces{\sigma} $ iff $ \overline{t} \subseteq E $, $ \neg \left( \conf{e_i}{e_j} \right) $, and $ \buEn{X}{e_i} \implies \overline{t_{i - 1}} \cap X \neq \emptyset $ for all $ 1 \leq i, j \leq n $ and all $ X \subseteq E $.
	Hence, by the Definition of traces in \S\ref{sec:DES}, $ t \in
	\traces{\sigma} $ iff $ t \in \traces{\delta} $, \ie $ \traces{\sigma} = \traces{\delta} $.
	
	Because of $ \configTraces{\sigma} = \configurations{\sigma} $ (\cf Lemma~\ref{lem:SESconf}), the Definition of
	configurations in \S\ref{sec:DES}, and Definition~\ref{def:SEStraceDef}, then
	also $ \configurations{\delta} = \configTraces{\sigma} = \configurations{\sigma} $.
\end{proof}

Moreover the DES and its encoding have exactly the same posets.

\begin{thm}
	\label{thm:DESintoSES}
	For each DES $ \delta $ there is a SES $ \sigma  = \ses{\delta} $ such that $  \posetsEq{\delta} {\sigma} $.
\end{thm}

\begin{proof}
	Let $ \delta = \left( E, \confOp, \buEnOp \right) $ be a DES. By Lemma~\ref{lem:DEStoSES}, $ \sigma = \ses{\delta} = \left( E, \confOp, \enabOp, \shrinkingCausality \right) $ is a SES such that $ \traces{\delta} = \traces{\sigma} $ and $ \configurations{\delta} = \configurations{\sigma} $.\\
	Let $ t = e_1 \cdots e_n \in \traces{\delta} $, $ 1 \leq i \leq n $, and the bundles $ \buEn{X_1}{e_i}, \ldots, \buEn{X_m}{e_i} $ all bundles pointing to $ e_i $.
	For $ U $ to be a cause for $ e_i $ Definition~\ref{def:SESposets} requires $ ( \ic{e_i} \setminus \dc{U}{e_i} ) \subseteq U $.
	Since $ \dc{H}{e} = \Set{ e' \mid \exists d \in H \logdot \drops{d}{e'}{e} } $ and $ \ic{e} = \{ e' \mid \enab{e'}{e} \} $, this condition holds iff the condition $ \enab{e'}{e_i} \implies e' \in U \lor \exists d \in U \logdot \drops{d}{e'}{e_i} $ holds for all $ e' \in E $.
	By Definition~\ref{def:DESintoSES}, then $ ( \forall 1 \leq k \leq n \logdot X_k \cap U \neq \emptyset) $ iff $ \left( \left( \ic{e_i} \setminus \dc{U}{e_i} \right) \subseteq U \right) $.
	So, by the Definitions~\ref{def:DESposets} and \ref{def:SESposets}, $  \posetsEq{\delta}{\sigma} $.
\end{proof}

Thus SESs and DESs have the same expressive power.

\begin{thm}
	\label{thm:SESvsDES}
	SESs are as expressive as DESs.
\end{thm}

\begin{proof}
	By the Theorems~\ref{thm:SESintoDES} and \ref{thm:DESintoSES}.
\end{proof}

Another ES with disjunctive causality are the Stable Event Structures (StESs) \cite{Winskel:IntroToES}.
Katoen~\cite{Katoen:Thesis} proves that StESs are strictly less expressive than DESs. Note that for this particular expressiveness result it does not matter whether the definition of DESs is based on BESs (as here and in \cite{Langerak97causalambiguity}) or EBESs (as in \cite{Katoen:Thesis}).
Then, \cite{flowES} introduces the Flow Event Structures (FESs) and proves that PESs are strictly less expressive than FESs, and that FESs are strictly less expressive than StESs.
Since we restrict our attention to finite configurations and because of Definition~\ref{def:pesConfigs}, PESs and rPESs have the same expressive power.
Finally, that BESs are strictly less expressive than FESs was shown in \cite{FLOW} and (more directly) in \cite{Langerak:Thesis}.
With the above theorem we conclude that SESs are strictly more expressive than StESs, FESs, BESs, PESs and rPESs.

\begin{cor}
	SESs are strictly more expressive than rPESs, BESs, FESs, and StESs.
\end{cor}

In \cite{Langerak97causalambiguity}, the authors prove that for DESs equivalence \wrt posets based on early causality coincides with trace equivalence. Since SESs are as expressive as DESs \wrt families of posets based on early causality, the same correspondence holds for SESs.

\begin{cor}
	\label{col:SESequiv}
	Let $ \sigma_1, \sigma_2 $ be two SESs. Then $ \posetsEq{\sigma_1 }{\sigma_2} $ iff $ \traces{\sigma_1} = \traces{\sigma_2} $.
\end{cor}

In Appendix~\ref{app:partialOrderSemantics}, we show that each SES and its encoding as well as each DES and its encoding have the same set of posets considering liberal, minimal, and late causality. 
Thus, the concepts of SESs and DESs are not only behaviorally equivalent but---except for the additional impossible events---also structurally closely related.

Note that $ \posetsEqOp $, $ \transEqOp $, and trace equivalence coincide on SES.

\begin{thm}
	\label{thm:SESequiv}
	Let $ \sigma, \sigma' $ be two SESs.
	Then $ \posetsEq {\sigma }{\sigma'} $ iff $  \transEq{\sigma}{\sigma'} $ iff $ \traces{\sigma} = \traces{\sigma'} $.
\end{thm}

\begin{proof}
	By Corollary~\ref{col:SESequiv}, $  \posetsEq {\sigma}{\sigma' }$ iff $ \traces{\sigma} = \traces{\sigma'} $.
	
	If $ \configurations{\sigma} \neq \configurations{\sigma'} $ then, because of $ \configTraces{\sigma} = \configurations{\sigma} $ (\cf Lemma~\ref{lem:SESconf}) and the Definitions~\ref{def:SESposets} and \ref{def:SEStrans}, $ \posetsNEq{\sigma } {\sigma'} $ and $\transNEq{\sigma}{\sigma'} $. Hence assume $ \configurations{\sigma} = \configurations{\sigma'} $.
	Note that, by Definition~\ref{def:SEStraceDef} and $ \configTraces{\sigma} = \configurations{\sigma} $ (\cf Lemma~\ref{lem:SESconf}), for all $ C \in \configurations{\sigma} $ there is a trace $ t \in \traces{\sigma} $ such that $ \overline{t} = C $. Moreover, for every trace $ t \in \traces{\sigma} $ except the empty trace there is a sub-trace $ t' \in \traces{\sigma} $ and a sequence of events $ e_1 \cdots e_m $ such that $ t = t' e_1 \cdots e_m $ and $ \forall e \in \Set{ e_1, \ldots, e_m } \logdot \left( \ic{e} \setminus \dc{\overline{t'}}{e} \right) \subseteq \overline{t'} $.
	Thus, by Lemma~\ref{lem:SEStransTraces}, $ \traces{\sigma} = \traces{\sigma'} $ iff $  \transEq {\sigma}{\sigma'}$.
\end{proof}

\subsection{SESs versus EBESs}

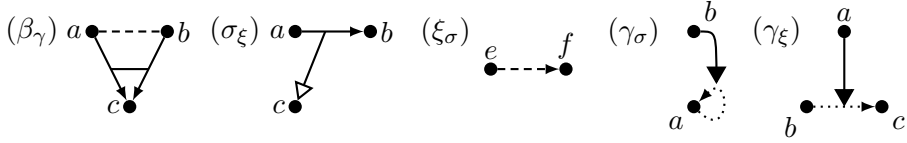
\begin{figure}[t]
	\centering
	\begin{tikzpicture}[bend angle=30]
		% figure a
		\event{b2}{0.3}{1.0}{left}{$ a $};
		\event{b3}{1.3}{1.0}{right}{$ b $};
		\event{b4}{0.8}{0}{left}{$ c $};
		\draw[enablingPES] (b2) edge (b4);
		\draw[enablingPES] (b3) edge (b4);
		\draw[thick] (0.55, 0.5) -- (1.05, 0.5);
		\draw[conflictPES] (b2) edge (b3);
		\node (g) at (-0.5, 1) {$(\beta_\gamma)$};
		% figure b
		\event{a1}{3}{1}{left}{$ a $};
		\event{a2}{4}{1}{right}{$ b $};
		\event{a3}{3}{0}{left}{$ c $};
		\node (a) at (2.2, 1) {$ (\sigma_\xi) $};
		\draw[enablingPES] (a1) edge (a2);
		\draw[dropping]	(3.4, 1) -- (a3);
		% figure c
		\event{b1}{5.6}{0.5}{above}{$ e $};
		\event{b2}{6.6}{0.5}{above}{$ f $};
		\node (b) at (5, 1) {$ (\xi_\sigma) $};
		\draw[conflictEBES] (b1) edge (b2);
		% figure d
		\event{ba}{8.3}{1}{above right}{$ b $};
		\event{bb}{8.3}{0}{below left}{$ a $};
		\node (gs) at (7.5, 1) {$(\gamma_\sigma)$};
		\draw[enablingAbsent] (bb) .. controls (8.8, -0.5) and (8.8, 0.5) .. (bb);
		\draw[adding] (ba) .. controls (8.6, 1) .. (8.6, 0.3);
		% figure e
		\event{ca}{9.8}{0}{below left}{$ b $};
		\event{cb}{10.8}{0}{below right}{$ c $};
		\event{cc}{10.3}{1}{above}{$ a $};
		\node (gx) at (9.4, 1) {$(\gamma_\xi)$};
		\draw[enablingAbsent] (ca) -- (cb);
		\draw[adding] (cc) -- (10.3, 0);
	\end{tikzpicture}
	\caption{Counterexamples.}
	\label{fig:counterExamples}
\end{figure}

SESs allow us to model disjunctive causality. As an example, consider the dropping of a causality as in $ \sigma_\xi $ of Figure~\ref{fig:counterExamples} (or Figure~\ref{fig:SESExample}). Such a disjunctive causality is not possible in EBESs.
On the other hand, the asymmetric conflict of an EBES cannot be modeled with a SES. As an example, consider $ \xi_\sigma $ of Figure~\ref{fig:counterExamples}, where $ f $ cannot precede $ e $.
Hence, EBESs and SESs are incomparable.

\begin{thm}
	\label{thm:SESvsEBES}
	SESs and EBESs are incomparable.
\end{thm}

\begin{proof}
	Let $ \sigma_{\xi} = \left( \Set{ a, b, c }, \emptyset, \Set{ \enab{a}{b} }, \Set{ \drops{c}{a}{b} } \right) $ be the SES that is depicted in Figure~\ref{fig:counterExamples}.
	Assume there is some EBES $ \xi = \left( E, \disaOp, \buEnOp \right) $ such that $ \traces{\sigma_\xi} = \traces{\xi} $.
	By Definition~\ref{def:SEStraceDef}, $ \traces{\sigma_\xi} = \Set{ \epsilon, a, c, ab, ac, ca, cb, abc, acb, cab, cba } $, \ie $ b $ cannot occur first.
	By Definition~\ref{def:EBESconf}, a disabling $ \disa{x}{y} $ implies that $ y $ can never precedes $ x $.
	Thus, we have $ \disaOp \cap \Set{ a, b, c }^2 = \emptyset $, because within $ \traces{\sigma_\xi} $ each pair of events of $ \Set{ a, b, c } $ occur in any order.
	Similarly, we have $ \buEnOp \cap \{ \buEn{X}{e} \mid e \in \Set{ a, b, c } \land X \cap \Set{ a, b, c } = \emptyset \} = \emptyset $, because $ \buEn{x}{y} $ implies that $ x $ always has to precede $ y $.
	Moreover, by Definition~\ref{def:EBESconf}, adding impossible events as causes or using them within the disabling relation does not influence the set of traces.
	Thus, there is no EBES $ \xi $ with the same traces as $ \sigma_{\xi} $. By Definition~\ref{def:EBESconf} and the definition of posets in EBESs, there is then no EBES $ \xi $ with the same configurations or posets as $ \sigma_{\xi} $.
	
	Let $ \xi_{\sigma} = \left( \Set{ e, f }, \Set{ \disa{e}{f} }, \emptyset \right) $ be the EBES that is depicted in Figure~\ref{fig:counterExamples}.
	Assume that there is some SES $ \sigma = \left( E, \confOp, \enabOp, \shrinkingCausality \right) $ such that $ \traces{\xi_{\sigma}} = \traces{\sigma} $.
	According to \S\ref{sec:EBES}, we have $ \traces{\xi_{\sigma}} = \Set{ \epsilon, e, f, ef } $.
	By Definition~\ref{def:SEStraceDef} and because of the traces $ e $ and $ f $, there are no initial causes for $ e $ and f, \ie $ \enabOp \cap \Set{ \enab{x}{y} \mid y \in \Set{ e, f } } = \emptyset $.
	Moreover, $ \confOp \cap \Set{ e, f }^2 = \emptyset $, because of the trace $ ef $ and because conflicts cannot be dropped.
	Thus $ fe \in \traces{\sigma} $ but $ fe \notin \traces{\xi_{\sigma}} $, \ie there is no SES $ \sigma $ with the same traces as $ \xi_{\sigma} $. 
        Then, by the Definitions~\ref{def:SEStraceDef} and \ref{def:SESposets}, there is no SES $ \sigma $ with the same configurations or families of posets as $ \xi_{\sigma} $.
\end{proof}

\subsection{SESs versus RCESs}

SESs are strictly less expressive than RCESs. Each SES can be translated into a transition-equivalent RCES.

\begin{lem}\label{lma:SESinRCES}
	For each SES $ \sigma $ there is a RCES $ \rho $ such that $  \transEq{\sigma}{\rho}$.
\end{lem}

\begin{proof}
	By the Definitions~\ref{def:SEStrans} and \ref{def:SESconf}, $ \transS{X}{Y} $ implies $ X \subseteq Y $  for all $ X, Y \in \configurations{\sigma} $.
	
	Assume that $ X \subseteq X' \subseteq Y' \subseteq Y $.
	Then, by Definition~\ref{def:SEStrans}, $ \transS{X}{Y} $ implies that $ \forall e, e' \in Y \logdot \neg \left( \conf{e}{e'} \right) $ and $ \forall e \in Y \setminus X \logdot \left( \ic{e} \setminus \dc{X}{e} \right) \subseteq X $.
	Then $ X \subseteq X' $ implies that $ \left( \ic{e} \setminus \dc{X'}{e} \right) \subseteq \left( \ic{e} \setminus \dc{X}{e} \right) $. Then, because of $ Y' \subseteq Y $, $ \forall e, e' \in Y' \logdot \neg \left( \conf{e}{e'} \right) $ and $ \forall e \in Y' \setminus X' \logdot \left( \ic{e} \setminus \dc{X'}{e} \right) \subseteq X' $.
	By Definition~\ref{def:SEStrans}, we then have $ \transS{X'}{Y'} $.
	
	Thus, $ \sigma $ satisfies the conditions of Definition~\ref{def:translationIntoRCES}. Then by Lemma~\ref{lma:TransInRCES}, $ \rho = \rces{\sigma} $ is a RCES such that $\transEq{\sigma}{\rho}$.
\end{proof}

On the other hand, there are RCESs that cannot be translated into a transition-equivalent SES.
As a counterexample, we use the RCES $ \rho_{\sigma}=(E,\enrcOp) $, where $ E =  \{ e , f \}$ and  $ \enrcOp =\{ (\emptyset,\{ e \}) , (\emptyset,\{ f \}) , (\{ f \},\{ e, f\}) \} $, which captures disabling in an EBES.

\begin{lem}
\label{lma:SESinRCESstrictly} 
	There is no transition-equivalent SES to the RCES $ \rho_{\sigma}=(E,\enrcOp) $, where $ E =  \{ e , f \}$ and  $ \enrcOp =\{ (\emptyset,\{ e \}) , (\emptyset,\{ f \}) , (\{ f \},\{ e, f\}) \} $.
\end{lem}

\begin{proof}
	Assume a SES $ \sigma = \left( E, \confOp, \enabOp, \shrinkingCausality \right) $ such that $ \transEq{\sigma }{\rho_{\sigma}} $. Then $ \configurations{\sigma} = \configurations{\rho_{\sigma}} $.
	By Definition~\ref{def:SEStraceDef} and $ \configTraces{\sigma} = \configurations{\sigma} $ (\cf Lemma~\ref{lem:SESconf}) and because of the configuration $ \Set{ e, f } \in \configurations{\rho_{\sigma}} $, the events $ e $ and $ f $ cannot be in conflict with each other, \ie $ \confOp \cap \Set{ e, f }^2 = \emptyset $.
	Moreover, because of the configurations $ \Set{ e }, \Set{ f } \in \configurations{\rho_{\sigma}} $, there are no initial causes for $ e $ and $ f $, \ie $ \enabOp \cap \Set{ \enab{x}{y} \mid y \in \Set{ e, f } } = \emptyset $.
	Note that the relation $ \shrinkingCausality $ cannot disable events.
	Thus, we have $ \forall a, b \in \Set{ e, f } \logdot \neg \left( \conf{a}{b} \right) $ and $ \left( \ic{e} \setminus \dc{\Set{ f }}{e} \right) = \emptyset \subseteq \Set{ f } $.
	But then, by Definition~\ref{def:SEStrans}, we have $ \transS{\Set{ f }}{\Set{ e, f }} $.
	Since $ \transRC{\Set{ f }}{\Set{ e, f }} $ does not hold, this violates our assumption, \ie there is no SES that is transition equivalent to $ \rho_{\sigma} $.
\end{proof}

Hence, SESs are strictly less expressive than RCESs.

\begin{thm}
	\label{thm:SESinRCESstrictly}
	SESs are strictly less expressive than RCESs.
\end{thm}

\begin{proof}
	Follows from the Lemmata~\ref{lma:SESinRCESstrictly} and \ref{lma:SESinRCES}.
\end{proof}

\section{Growing Causality}
\label{sec:GES}

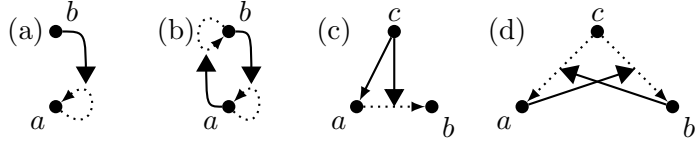
\begin{figure}[t]
	\centering
	\begin{tikzpicture}
		% figure b
		\event{ba}{1.8}{1}{above right}{$ b $};
		\event{bb}{1.8}{0}{below left}{$ a $};
		\node (b) at (1.4, 1) {(a)};
		\draw[enablingAbsent] (bb) .. controls (2.4, -0.5) and (2.4, 0.5) .. (bb);
		\draw[adding] (ba) .. controls (2.2, 1) .. (2.2, 0.3);
		% figure c
		\event{aa}{4.1}{0}{below left}{$ a $};
		\event{ab}{4.1}{1}{above right}{$ b $};
		\node (a) at (3.4, 1) {(b)};
		\draw[enablingAbsent] (aa) .. controls (4.6, -0.5) and (4.6, 0.5) .. (aa);
		\draw[adding] (ab) .. controls (4.4, 1) .. (4.4, 0.3);
		\draw[enablingAbsent] (ab) .. controls (3.6, 1.5) and (3.6, 0.5) .. (ab);
		\draw[adding] (aa) .. controls (3.8, 0) .. (3.8, 0.7);
		% figure d
		\event{da}{5.8}{0}{below left}{$ a $};
		\event{db}{6.8}{0}{below right}{$ b $};
		\event{dc}{6.3}{1}{above}{$ c $};
		\node (d) at (5.5, 1) {(c)};
		\draw[enablingAbsent] (da) -- (db);
		\draw[enablingPES] (dc) -- (da);
		\draw[adding] (dc) -- (6.3, 0);
		% figure e
		\event{ea}{8}{0}{below left}{$ a $};
		\event{eb}{10}{0}{below right}{$ b $};
		\event{ec}{9}{1}{above}{$ c $};
		\node (e) at (7.8, 1) {(d)};
		\draw[enablingAbsent] (ec) -- (eb);
		\draw[adding] (ea) -- (9.5, 0.5);
		\draw[enablingAbsent] (ec) -- (ea);
		\draw[adding] (eb) -- (8.5, 0.5);
	\end{tikzpicture}
	\vspace*{-1em}
	\caption{GESs modeling disabling, conflict, temporary disabling, and resolvable conflicts.}
	\label{fig:GESExamples}
\end{figure}

Like with SESs, we express our extension for growing causality on rPESs by a new relation: we use $ \growingCausality \subseteq E^3 $, where $ \left( a, c, t \right) \in \growingCausality $, also denoted as \addcause{a}{c}{t}, models the fact that the occurrence of the so-called \emph{adder} $ a $ adds $ c $ as a cause for the target $ t $. Thus, $ a $ is a condition for the causal dependency $ \enab{c}{t} $. By definition, we will ensure that $a\notin\{c,t\}$, which is a difference to the version in \cite{dynamicCausality15}.
The relation $\growingCausality$ is visualized in the example of Figure~\ref{fig:GESExamples}(c) by an arrow with a filled head from the modifier $ c $ to the added dependency $ \enab{a}{b} $; to  denote that this dependency does not initially exist, it is depicted as a dotted arrow. (In this example, there is an additional causality $ \enab{c}{a} $.)

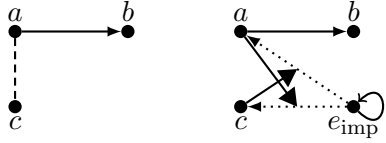
\begin{figure}[t]
	\centering
	\begin{tikzpicture}[]
		% figure a
		\event{a}{0}{1}{above}{$ a $};
		\event{b}{1.5}{1}{above}{$ b $};
		\event{c}{0}{0}{below}{$ c $};
		\draw[conflictPES] (a) edge (c);
		\draw[enablingPES] (a) edge (b);
		% figure b
		\event{aa}{3}{1}{above}{$ a $};
		\event{bb}{4.5}{1}{above}{$ b $};
		\event{cc}{3}{0}{below}{$ c $};
		\event{i}{4.5}{0}{below}{$ e_{\text{imp}} $};
 		\draw[enablingPES] (aa) edge (bb);
		\draw[enablingAbsent] (i) -- (cc);
		\draw[enablingAbsent] (i) -- (aa);
		\draw[enablingPES,->] (i) .. controls (5, -0.5) and (5, 0.5) .. (i);
		\draw[adding] (aa) -- (3.75,0);
		\draw[adding] (cc) -- (3.75,0.5);
	\end{tikzpicture}
	\caption{A rPES with binary conflict (and no inherited conflict) and a simulating GES.}
	\label{fig:GESconflict}
\end{figure}

In comparison to a rPES, we can model a binary conflict between $ a $ and $ c $ by a mutual disabling, and a disabling of an event by an addition of an impossible cause: assume a fresh impossible event $ e_{\text{imp}} $ with $ \enab{e_{\text{imp}}}{e_{\text{imp}}} $ and the mutual disabling $ \adds{b}{e_{\text{imp}}}{a} $ and $ \adds{a}{e_{\text{imp}}}{b} $ (see Figure~\ref{fig:GESconflict}). Thus, we omit the conflict relation in the Definition of GESs.

\begin{defi}\label{def:GES}
	A \emph{Growing Causality Event Structure (GES)} is a triple $ \gamma = \left( E,\enabOp, \growingCausality \right)$,  where $ E$ is a set of events, ${\enabOp}\subseteq{E^2}$ the initial causality relation, and $ \growingCausality \subseteq E^3 $ is the \emph{growing causality} relation such that, for all $a,c,t \in E$, $\addcause{a}{c}{t}$ implies $\lnot (c \rightarrow t)$ and $ a\notin\{c,t\}$.
\end{defi}

We refer to the causes added to an event \wrt a specific history by the function $ \acD :\Powerset{E} \times E \to \Powerset{E} $, defined as $\ac{H}{e} = \Set{ e' \mid \exists a \in H \logdot \addcause{a}{e'}{e} }$, and to the initial causality by the function $ \icD $ as defined in \S\ref{sec:SES}.
Similar to the RCESs the behavior of a GES can be defined by a transition relation.

\begin{defi}
	\label{def:GEStraceDef}
	\label{def:GESNoConcurTransition}
	\label{def:GESConfigurations}
	\label{def:GEStrans}
	Let $ \gamma = \left( E, \enabOp, \growingCausality \right) $ be a GES.
	\begin{itemize}[noitemsep]
		\item A \emph{trace} of $ \gamma $ is a sequence of distinct events $ t	= e_1 \cdots e_n $ with $ \overline{t} \subseteq E $ such that \[\forall 1 \leq i \leq n \logdot  \left(\ic{e_i} \cup \ac{\overline{t_{i-1}}}{e_i} \right) \subseteq \overline{t_{i-1}}.\]Then $ C \subseteq E $ is a \emph{trace-based configuration} of $ \gamma $ if there is a trace $ t $ such that $ C = \overline{t} $.	The set of traces of $ \gamma $ is denoted by $ \traces{\gamma} $ and the set	of its trace-based configurations is denoted by $ \configTraces{\gamma} $.
		\item Let $ X, Y \subseteq E $.	Then $ \transG{X}{Y} $ iff 
			\begin{enumerate}
				\item$X \subset Y$
				\item$\forall e \in Y \setminus X \logdot	( \ic{e} \cup \ac{X}{e} ) \subseteq X $
				\item$\forall t, a \in Y\setminus X\logdot\forall c\in E\logdot	\addcause{a}{c}{t} \implies c \in X$.
		\end{enumerate}
		\item  The set of all configurations of $ \gamma $ is $ \configurations{\gamma} = \Set{ X \subseteq E \mid \emptyset \! \transGOp^* \! X \land X \text{ is finite } } $, where $ \transGOp^*$ is the reflexive and transitive closure of $ \transGOp $.
		\item For $X\subseteq Y\subseteq E$  we define the set of added dependencies as \[\added{X,Y}:=\left\{(c,t)\mid\exists a\in Y\setminus X\logdot \adds{a}{c}{t}  \right\}.\]
	\end{itemize}
\end{defi}

\noindent
The last condition in the transition definition prevents the concurrent occurrence of a target and its modifier since they are not independent. The exception is when the cause has already occurred; in that case, the modifier does not change the target's predecessors. 
Again, we consider the reachable and finite configurations, and show in the Appendix (\cf Lemma~\ref{lma:GESConfigEquivalence}) that the definitions of reachable and trace-based configurations coincide.
We use Definition~\ref{def:transEq} to obtain $ \transEq{}{} $ for GESs and consider two GESs as equally expressive if they are transition equivalent.

\subsection{An Alternative Transition Definition}

In the following, we provide an alternative transition definition for GESs, similar to the SESs in Definition~\ref{def:SEStransalt}. As before, this definition stresses the \emph{state} of the causal dependencies, but this time due to its dynamic growth. Thereafter, we prove the equivalence of both definitions.

\begin{defi}\label{def:GEStransalt}
  Let $ \gamma = \left( E, \incaus, \growingCausality \right) $ be a GES, 
  and $ X,Y \subseteq E $. \\
  Then $\transGa{(X,\caus{X})}{(Y,\caus{Y})}$, if 
  \begin{enumerate}
  \item $ X \subset Y $,
  \item 
    \begin{enumerate}
    \item
      $ \forall e \in Y {\setminus} X \logdot \left\{e'\mid
        (e',e)\in{\caus{X}}\right\} \subseteq X $,
    \item $ {\caus{Y}} = {\caus{X}} \cup {\added{X,Y}} $, and
    \end{enumerate}
  \item $ \forall t, a \in Y {\setminus} X \logdot \forall c \in E \logdot \addcause{a}{c}{t} \implies c \in X $.
  \end{enumerate}
\end{defi}

\noindent
Here each state explicitly shows all causality information (without the $ \acD $ function) and the growing of the causality relation is explicitly represented by the causal state $\caus{X}$. 
%The other conditions are just the same as above.

The following Lemma provides an explicit, non-recursive definition of causal states.

\begin{lem}\label{lma:GEScaus}
	Let $ \gamma = \left( E, \incaus, \growingCausality \right) $ be a GES and $ (X,\caus{X})$ be a state of $ \gamma $.\\
	Then $ {\caus{X}} = {\incaus \cup \added{\emptyset,X}} $.
\end{lem}

\begin{proof}
	Let $\emptyset=:C_0,C_1,\dots,C_{n-1},C_n:=X$ be a configuration sequence such that we get $ \transSa{(C_0, \incaus)}{(C_1, \caus{C1})} \dots \transSa{(C_i, \caus{C_i})}{(C_{i+1}, \caus{C_{i+1}})} \dots \transSa{(C_{n-1}, \caus{C_{n-1}})}{(X, \caus{X})}$.\\
	By Definition~\ref{def:GEStransalt}, we have $ \caus{X} = (\dots(\incaus\cup\added{C_0,C_1})\cup \dots )\cup\added{C_{n-1},X} $, and with basic set theory, and the Definition of $ \added{X, Y} $ in \ref{def:GEStrans}, we get $ {\caus{X}} = {\incaus \cup \added{\emptyset,X}} $. 
\end{proof}

Now we can prove that both transition definitions of GESs coincide.

\begin{thm}
	Let $ \gamma = \left( E, \incaus, \growingCausality \right) $be a GES and $ X, Y \subseteq E $. Then \[\transG{X}{Y}\iff\transGa{(X,\caus{X})}{(Y,\caus{Y})}.\]
\end{thm}

\begin{proof}
  Since all but condition \emph{(2) }in the transition definitions in \ref{def:GEStrans} and \ref{def:GEStransalt} coincide, we just show $ \forall e\in Y\setminus X\logdot (\left( \ic{e} \cup \ac{X}{e} \right) = \left\{e'\mid {(e',e)}\in{\caus{X}}\right\})$. This follows from definitions of $\acD$ and $\icD$ and Lemma \ref{lma:GEScaus}:
  \begin{align*}
		\ic{e} \cup \ac{X}{e}&\\
		=&\left\{e'\mid (e',e)\in{\incaus}\right\}\cup\left\{c\mid\exists a\in X\logdot\adds{a}{c}{e}\right\}\\%
		=&\left\{e'\mid (e',e)\in{\incaus} \vee\, e'\in \left\{c\mid\exists a\in X\logdot\adds{a}{c}{e}\right\} \right\}\\
		=&\left\{e'\mid (e',e)\in{\incaus} \vee\, (e',e)\in \left\{(c,e)\mid\exists a\in X\logdot\adds{a}{c}{e}\right\} \right\}\\
		=&\left\{e'\mid (e',e)\in{\incaus}\cup\left\{(c,t)\mid\exists a\in X\logdot \adds{a}{c}{t}\right\}\right\}\\=&\left\{e'\mid (e',e)\in{\incaus}\cup\added{\emptyset, X}\right\}\\
		=&\left\{e'\mid (e',e)\in{\caus{X}}\right\}
	\end{align*}
	As $ {\caus{Y}} = {\caus{X} \cup \added{X, Y}} $ by Lemma~\ref{lma:GEScaus}, 
        we get $ \transG{X}{Y} \iff \transGa{(X,\caus{X})}{(Y,\caus{Y})} $.
\end{proof}

\subsection{Expressive Power}

%\paragraph{Disabling in GES} 
Disabling as defined in EBESs or the asymmetric event structure of \cite{Baldan20011} can be modeled by $ \growingCausality $. For example, $ \disa{b}{a} $ can be modeled by \addcause{b}{a}{a} as depicted in Figure~\ref{fig:GESExamples}(a). Analogously, conflicts can be modeled by $ \growingCausality $ through mutual disabling, as depicted in Figure~\ref{fig:GESExamples}(b). However, if we later combine shrinking and growing causality ESs in \S\ref{sec:DCES}, another dropper $ d $ can resolve the disabling.

%\paragraph{Temporary Disabling} 
Let us reconsider the example GES in Figure~\ref{fig:GESExamples}(c).
Initially, $ b $ is enabled and $ a $ is disabled. 
The occurrence of $ c $, which (enables $ a $ and) disables $ b $ by adding $ a $ as a cause with $ \addcause{c}{a}{b} $. By this, $ b $ may be \emph{temporarily disabled} by an occurrence of $c$ until $ a $ occurs and \emph{re-enables} $b$.
In \emph{inhibitor event structures} \cite{Baldan2004129}, another kind of disabling can be expressed, where an event $ b $ can be disabled by another event $ c $ until an event out of a \emph{set} $ X $ (instead of only a single event $a$, as above) occurs. This behavior|which may be called \emph{disjunctive re-enabling}|cannot be modeled in GESs, but in DCESs (\cf \S\ref{sec:DCES}).

%\paragraph{Resolvable Conflicts} 
Also resolvable conflicts can be modeled by a GES. For example the GES in Figure~\ref{fig:GESExamples}(d) with \addcause{a}{c}{b} and \addcause{b}{c}{a} models a conflict between $ a $ and $ b $ that can be resolved by $ c $. Note that this example depends on the idea that a modifier and its target cannot occur concurrently (\cf Definition~\ref{def:GESNoConcurTransition}). Note also that resolvable conflicts are a reason why families of configurations are not sufficient to describe the semantics of GESs or RCESs.

\subsection{GESs versus PESs}

The GESs are strictly more expressive than rPES and PESs, since they are in essence a generalization of the rPES (when the conflict relation is encoded as mutual disabling, \ie by mutual adding of an impossible cause) and they can express a disabling of an event (as in Fig. \ref{fig:GESExamples} (b)).

\begin{lem}
	\label{lma:rPESinGES}
	The GESs are strictly more expressive than the rPESs and PESs.
\end{lem}
\begin{proof}
	Let $ \pi = \left( E, \confOp, \enabOp \right) $ be a rPES. Then, its embedding $\gamma_\pi$ into a GES is given by $ \gamma_\pi:=(E\cup\{e_{\text{imp}}\},\enabOp\cup\{(e_{\text{imp}},e_{\text{imp}})\},\{(a,e_{\text{imp}},b),(b,e_{\text{imp}},a)\mid (a,b)\in\confOp)\} $ for a fresh (and impossible) event $e_{\text{imp}}$.

	Let $C$ be a configuration in $\pi$. By Definition \ref{def:rpesConfigs}, we have that $C$ is downward-closed, \ie $ \forall e, e' \in E \logdot \enab{e}{e'} \land e' \in C \implies e \in C $, conflict-free, \ie $ \forall e, e' \in C \logdot \neg \left( \conf{e}{e'} \right) $, and ${\enabOp}\cap{C^2}$ is free of cycles. Let further be $C=\Set{e_1,\dots,e_n}$ such that also all $C_j:=\Set{e_i\mid i\leq j}$ are configurations of $\pi$ (this is always possible since $C$ is conflict-free and ${\enabOp}\cap{C^2}$ is free of cycles).
	Now, we have $\transG{C_i}{C_{i+1}}$ for all $1\leq i\leq n{-}1$, as all three conditions in Definition \ref{def:GEStrans} are met. 
        The first condition is satisfied, as $C_i\subset C_{i+1}$ by construction. 
        The second condition is satisfied, because the only added causality is the disabling of conflicting events and thus does not affect events in $C$, and then it coincides with the downward-closure. 
        The third condition just prevents, \wrt our encoding, the concurrent occurrence of conflicting events.
	Let on the other hand $t=e_1,\dots e_n$ be a trace in $\gamma_\pi$.  By Definition \ref{def:GEStrans}, we have
\[\forall 1 \leq i \leq n \logdot  \left(\ic{e_i} \cup \ac{\overline{t_{i-1}}}{e_i} \right) \subseteq \overline{t_{i-1}}.\]
	We show that all $C_j:=\Set{e_i\mid i\leq j}$ with $j\leq n$ are conflict-free and downward-closed in $\pi$ and ${\enabOp}\cap{C^2}$ is free of cycles: Assume that $C_j$ is the smallest not conflict-free set then it contains $e_i$ and $e_j$ with $e_i\confOp e_j$ for some $i<j$. But since they are in conflict we have $\addcause{e_i}{e_{\text{imp}}}{e_j}\in\growingCausality$, where $e_{\text{imp}}$ is impossible. This contradicts 
	\[ \left(\ic{e_j} \cup \ac{\overline{t_{j-1}}}{e_j} \right) \subseteq \overline{t_{j-1}}\]
	because $e_{\text{imp}}$ is initially impossible and thus $e_{\text{imp}}\not\in\overline{t_{i-1}}$. If some $C_j$ would not be downward-closed, the $\ic{e_j}\subseteq \overline{t_{j-1}}$ would not hold. Let ${\enabOp}\cap{C_j^2}$ be the minimal not cycle-free relation, then $e_j$ must have introduced a cycle and thus be a cause of some $e_i$ with $i<j$, but then $\ic{e_i}\subseteq \overline{t_{i-1}}$ would not hold. Therefore, we have that $\gamma_\pi$ is a proper encoding of $\pi$.

	Now consider the GES $\gamma=(\Set{a,b},\emptyset,\Set{(b,a,a)})$ (as in Figure~\ref{fig:GESExamples}(b)). By Definition \ref{def:GEStrans}, $b$ and $ab$  are possible traces in $\gamma$, but $ba$ is forbidden. If there was an equivalent rPES $\pi_\gamma$, the trace $b$ would imply that $b$ does not depend on $a$. The trace $ab$ would imply that $\neg(a\confOp b)$, and therefore $ba$ would also be possible in $\pi_\gamma$, contradicting the equivalence assumption.

 Combining the  two results, we get that GESs are strictly more expressive than rPESs. As we restrict our attention to finite configurations and with Definition~\ref{def:pesConfigs}, rPESs and PESs have the same expressive power. Hence, GESs are strictly more expressive than PESs.
\end{proof}

\subsection{GESs versus EBESs}

As shown in Figure~\ref{fig:GESExamples}(a), GESs can model disabling. Nevertheless, EBESs and GESs are incomparable. 

On the one hand, GESs cannot model the disjunction in the enabling relation that EBESs inherit from BESs. The BES $ \beta_\gamma $ of Figure~\ref{fig:counterExamples} models a simple case of disjunction in the enabling relation.

\begin{lem}
	\label{lma:EBESninGES}
	There is no configuration-equivalent GES to $ \beta_\gamma $ of Figure~\ref{fig:counterExamples}.
\end{lem}

\begin{proof}
	Assume a GES $ \gamma = \left( E,  \enabOp, \growingCausality \right) $ such that $ \configurations{\gamma} = \configurations{\beta_\gamma} $.	According to \S\ref{sec:BES}, $ \configurations{\beta_\gamma} = \Set{\emptyset, \Set{ a }, \Set{ b }, \Set{ a, c }, \Set{ b, c } } $. Because $ \Set{ c } \notin \configurations{\beta_\gamma} $, $ \Set{ a, c } \in \configurations{\beta_\gamma} $, and by Definition~\ref{def:GEStraceDef} and $ \configTraces{\gamma} = \configurations{\gamma} $ (\cf Lemma~\ref{lma:GESConfigEquivalence}), $ a $ has to be an initial cause of $ c $ in $ \gamma $, \ie $ \enab{a}{c} $.
	But then, by Definition~\ref{def:GEStraceDef} and $ \configTraces{\gamma} = \configurations{\gamma} $ (\cf Lemma~\ref{lma:GESConfigEquivalence}), $ \Set{ b, c } \notin \configurations{\gamma} $ although $ \Set{ b, c } \in \configurations{\beta_\gamma} $. This violates our assumption, \ie no GES can be configuration equivalent to $ \beta_\gamma $.
\end{proof}

On the other hand, EBESs cannot model conditional causal dependencies as visualized by the GES $ \gamma_\xi $ of Figure~\ref{fig:counterExamples}.

\begin{lem}\label{lma:GESninEBES}
	There is no trace-equivalent EBES to $\gamma_\xi$ of Figure~\ref{fig:counterExamples}.
\end{lem}

\begin{proof}
	Assume a EBES $ \xi = \left( E, \confOp, \buEnOp \right) $ such that $ \traces{\xi} = \traces{\gamma_\xi} $. By Definition~\ref{def:GEStraceDef}, we have $ a, c, ca, bac \in \traces{\gamma_\xi} $ and $ ac \notin \traces{\gamma_\xi} $. Because of $ a, c \in \traces{\gamma_\xi} $ and by Definition~\ref{def:EBESconf}, $ a $ and $ c $ have to be initially enabled in $ \xi $, \ie $ \buEnOp \cap \Set{ \buEn{X}{y} \mid y \in \Set{ a, c } } = \emptyset $.	Moreover, because of $ ca, bac \in \traces{\gamma_\xi} $, $ a $ cannot disable $ c $, \ie $ \neg \left( \disa{c}{a} \right) $. But then $ ac \in \traces{\xi} $. This violates our assumption, \ie there is no trace-equivalent EBES to $ \gamma_\xi $.
\end{proof}

Thus GESs are incomparable to BESs as well as EBESs.

\begin{thm}	\label{thm:GESvsEBES}
	GESs are incomparable to BESs and EBESs.
\end{thm}

\begin{proof}
	By Lemma~\ref{lma:EBESninGES}, there is no GES that is configuration equivalent to the BES $ \beta_\gamma $. Thus, no GES can have the same families of posets as the BES $ \beta_\gamma $, because two BESs with different configurations cannot have the same families of posets (\cf \S\ref{sec:BES}). Moreover, by the Definitions~\ref{def:BES} and \ref{def:EBES}, each BES is also an EBES. Thus, no GES can have the same families of posets as the EBES $ \beta_\gamma $.
	
	By Lemma~\ref{lma:GESninEBES}, there is no EBES and thus also no BES that is trace equivalent to the GES $ \gamma_\xi $. By Definition~\ref{def:GESNoConcurTransition}, two GESs with different traces cannot have the same transition graphs. Thus no EBES or BES can be transition equivalent to $ \gamma_\xi $.
\end{proof}

\subsection{GESs versus SESs}

GESs are also incomparable to SESs, because the addition of causes cannot be modeled by SESs. As a counterexample, we use the GES $ \gamma_\sigma $ of Figure~\ref{fig:counterExamples}.

\begin{lem}\label{lma:GESninSES}
	There is no trace-equivalent SES to $\gamma_\sigma$ of Figure~\ref{fig:counterExamples}.
\end{lem}

\begin{proof}
	Assume a SES $ \sigma = \left( E, \confOp, \enabOp, \shrinkingCausality \right) $ such that $ \traces{\sigma} = \traces{\gamma_\sigma} $. By Definition~\ref{def:GEStraceDef}, $ \traces{\gamma_\sigma} = \Set{ \epsilon, a, b, ab } $.
	Because of the trace $ ab \in \traces{\gamma_\sigma} $ and by Definition~\ref{def:SEStraceDef}, $ a $ and $ b $ cannot be in conflict, \ie $ \neg (\conf{a}{b} ) $ and $ \neg ( \conf{b}{a} ) $.	Moreover, because of the traces $ a, b \in \traces{\gamma_\sigma} $, there are no initial cases for $ a $ or $ b $, \ie $ \enabOp \cap \Set{ \enab{x}{y} \mid y \in \Set{ a, b } } = \emptyset $.	Thus, by Definition~\ref{def:SEStraceDef}, $ ba \in \traces{\sigma} $ but $ ba \notin \traces{\gamma_\sigma} $.	This violates our assumption, \ie no SES can be trace equivalent to $	\gamma_\sigma $.
\end{proof}

Then since BESs are incomparable to GESs, BESs are less expressive than DESs, and DESs are as expressive as SESs, we conclude that GESs and SESs are incomparable.

\begin{thm}\label{thm:GESvsSES}
	GESs and SESs are incomparable.
\end{thm}

\begin{proof}
	By Lemma~\ref{lma:GESninSES}, no SES is trace equivalent to the GES $ \gamma_\sigma $. By Definition~\ref{def:GESNoConcurTransition}, two GESs with different traces cannot have the same transition graphs. Thus, no SES is transition equivalent to the GES $ \gamma_\sigma $.
	
	By \cite{Langerak:Thesis}, BESs are less expressive than EBESs and by \cite{Langerak97causalambiguity}, BESs are less expressive than DESs.	By Theorem~\ref{thm:GESvsEBES}, BESs and GESs are incomparable and, by Theorem~\ref{thm:SESvsDES}, DESs are as expressive as SESs.
	Thus, GESs and SESs are incomparable.
\end{proof}

\subsection{GESs versus RCESs}

As illustrated in Figure~\ref{fig:GESExamples}(d), GESs can model resolvable conflicts. Nevertheless, they are strictly less expressive than RCESs. We show first that each GES can be translated by Definition~\ref{def:translationIntoRCES} into a transition-equivalent RCES.

\begin{lem}
	\label{lma:GESinRCES}
	For each GES $ \gamma $ there is an RCES $ \rho $, namely $ \rho = \rces{\gamma} $, such that $ \transEq{\gamma}{\rho} $.
\end{lem}

\begin{proof}
	Let $ \gamma = \left( E,\enabOp, \growingCausality \right) $.	 By Definition~\ref{def:GESNoConcurTransition}, $ \transG{X}{Y} $ implies $ X \subseteq Y $.
	
	Assume $ X \subseteq X' \subseteq Y' \subseteq Y $ and $ \transG{X}{Y} $.
	By Definition~\ref{def:GESNoConcurTransition}, then we have that $\forall e \in ( Y' \setminus X' ) \logdot ( \ic{e} \cup \ac{X}{e} ) \subseteq X $, and $ \forall t, m \in Y\setminus X\logdot\forall c\in E \logdot \addcause{m}{c}{t} \implies c\in X $. Moreover, because $ \forall t, m \in Y\setminus X\logdot\forall c\in E \logdot \addcause{m}{c}{t} \implies c\in X $, $ \ac{X}{e} = \ac{X'}{e} $ for all $ e \in Y' \setminus X' $.
	As a consequence, we have that $ \forall e \in \left( Y' \setminus X' \right) . \left( \ic{e} \cup \ac{X'}{e} \right) \subseteq X' $ and $ \forall t, m \in Y'\setminus X'\logdot\forall c\in E \logdot \addcause{m}{c}{t} \implies c\in X' $. Thus, by Definition~\ref{def:GESNoConcurTransition}, $ \transG{X'}{Y'} $.
	
	By Lemma~\ref{lma:TransInRCES}, $ \rho = \rces{\gamma} $ is an RCES and $ \transEq{\gamma}{\rho} $.
\end{proof}

On the other hand, there is no GES that is transition equivalent to the RCES $ \rho_{\gamma} $ in Figure~\ref{fig:transgraphs}.
It models the case, where after $ a $ and $ b $ the event $ c $ becomes impossible, \ie it models disabling by a set instead of a single event.

\begin{lem}\label{lma:GESinRCESstrictly}
	There is no transition-equivalent GES to $\rho_{\gamma}$ of Figure~\ref{fig:transgraphs}.
\end{lem}

\begin{proof}
	Assume a GES $ \gamma = \left( E, \enabOp, \growingCausality \right) $ such that $ \transEq {\gamma}{\rho_{\gamma}} $. Then $ \configurations{\gamma} = \configurations{\rho_{\gamma}} $.	By Definition~\ref{def:GESNoConcurTransition}, and  because of the configurations $ \Set{ a }, \Set{ b }, \Set{ c } \in \configurations{\rho_{\gamma}} $, there are no initial causes for $ a $, $ b $, or $ c $, \ie $ {\enabOp} \cap \Set{ \enab{x}{y} \mid y \in \Set{ a, b, c } } = \emptyset $. Moreover, because of the configurations $ \Set{ a, c }, \Set{ b, c } \in \configurations{\rho_{\gamma}} $, neither $ a $ nor $ b $ can add a cause to $ c $. Thus, we have $ \left( \ic{c} \cup \ac{\Set{ a, b }}{c} \right) = \emptyset \subseteq \Set{ a, b } $.	But then, by Definition~\ref{def:GESNoConcurTransition}, $ \transG{\Set{ a, b }}{\Set{ a, b, c }} $. Since $ \neg \left( \transRC[{\rho_\gamma}]{\Set{ a, b }}{\Set{ a, b, c }} \right) $, this violates our assumption, \ie there is no GES that is transition equivalent to $ \rho_{\gamma} $.
\end{proof}

Hence, GESs are strictly less expressive than RCESs.

\begin{thm}\label{thm:GESinRCESstrictly}
	GESs are strictly less expressive than RCESs.
\end{thm}

\begin{proof}
	Follows from Lemma~\ref{lma:GESinRCES} and \ref{lma:GESinRCESstrictly}.
\end{proof}

\section{Dynamic Causality}
\label{sec:DCES}

We have investigated shrinking and growing causality separately. In this section, we combine them and examine the resulting expressive power. Note that we again omit the binary conflict relation from this definition, because it can be modeled by adders (\cf \S\ref{sec:GES}). Since we slightly changed the definitions for the shrinking and growing causality relations, there is a slight difference to the version in \cite{dynamicCausality15}.

\begin{defi}
	\label{def:DCES}
	A \emph{Dynamic Causality Event Structure (DCES)} is a quadruple $ \Delta = $\linebreak
	$ \left( E, \enabOp, \shrinkingCausality, \growingCausality \right) $, where $E$ is a set of events, $\enabOp \; \subseteq E^2$ the initial causality relation, $ \shrinkingCausality \subseteq E^3$ is the shrinking causality relation, and $\growingCausality \subseteq E^3$ is the growing causality relation such that for all $d,c,t,a\in E$:
	\begin{enumerate}
		\item \label{eq:DCESDroppingExistentCauses} $ \drops{d}{c}{t} \land \nexists a	\in E \logdot \addcause{a}{c}{t} \Longrightarrow c \rightarrow t $
		\item \label{eq:DCESnoSelfDropping} $ \drops{d}{c}{t} \Longrightarrow d\notin\{c,t\} $
		\item \label{eq:DCESAddingMissingCauses} $ \addcause{a}{c}{t} \land \nexists d	\in E \logdot \drops{d}{c}{t}	\Longrightarrow \lnot (c \rightarrow t) $
		\item \label{eq:DCESnoSelfAdding} $ \addcause{a}{c}{t} \Longrightarrow a\notin\{c,t\} $
		\item \label{eq:ModifierCannotAddAndDropTheSameTarget} $ \addcause{a}{c}{t} \Longrightarrow \neg (\drops{d}{c}{t}) $.
		\end{enumerate}
\end{defi}

The Conditions~\ref{eq:DCESDroppingExistentCauses}, \ref{eq:DCESnoSelfDropping}, \ref{eq:DCESAddingMissingCauses}, and \ref{eq:DCESnoSelfAdding} are a generalization of the Conditions in the Definitions~\ref{def:SES} and \ref{def:GES}, respectively. If there are droppers and adders for the same causal dependency we do not specify whether this dependency is contained in $\enabOp$, because the semantics depends on the order in which the droppers and adders occur. Condition~\ref{eq:ModifierCannotAddAndDropTheSameTarget} prevents that a modifier adds and drops the same cause for the same target.

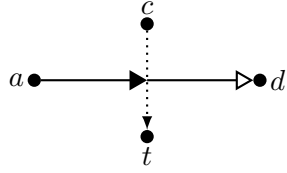
\begin{figure}[t]
	\centering
	\begin{tikzpicture}[]
		% figure a
		\event{c}{1.5}{1.5}{above}{$ c $};
		\event{t}{1.5}{0}{below}{$ t$};
		\event{a}{0}{0.75}{left}{$ a $};
		\event{d}{3}{0.75}{right}{$ d $};
		\draw[enablingAbsent] (c) edge (t);
		\draw[adding] (a) -- (1.5,0.75);
		\draw[dropping] (1.5,0.75) -- (d);
	\end{tikzpicture}
	\caption{An order sensitive DCES.}
	\label{fig:DCESorder}
\end{figure}

The order of occurrence of droppers and adders determines the causes of an event. For example assume \addcause{a}{c}{t} and \drops{d}{c}{t} (as depicted in Figure~\ref{fig:DCESorder}), then after $ad$, $t$ does not depend on $c$, whereas after $da$, $t$ depends on $c$. Thus, configurations like $\Set{a,d}$ are not expressive enough to represent the state of such a system.

Therefore, in a DCES, a state is a pair of a configuration $C$ and the current causal state, \ie, the relation $\caus{C}$ that contains the current causal dependencies.
Note that the initial state is the only configuration with an empty set of events; for other sets of events there can be different configuration.
The behavior of a DCES is defined by the transition relation between finite configurations.

\begin{defi}
	\label{def:DCTransition}
	Let $\Delta = \left( E, \caus{in}, \shrinkingCausality,
	\growingCausality\right) $ be a DCES and $X,Y\subseteq E$.\\ Then
	$\transDC{(X,\caus{X})}{(Y,\caus{Y})$ if:
	\begin{enumerate}
		\item \label{eq:ConfigContainement} $ X \subset Y $
		\item \label{eq:OnlyEnabledEvents} $ \forall e \in Y\setminus X \logdot \left\{e'\mid (e',e)\in{\caus{X}}\right\} \subseteq X $
		\item \label{eq:correctCausalityUpdate} $ {\caus{Y}} = {(\caus{X} \setminus \dropped{X,Y}) \cup \added{X, Y}}} $
		\item \label{eq:NCNoConcurrency} $ \forall (c,t) \in {\added{X, Y}} \cap {\dropped{X,Y}} \logdot (c \in X \vee t \in X) $
		\item \label{eq:NotToMuchConcurrency}$\forall t,a \in Y\setminus X \logdot \forall c \in E \logdot \addcause{a}{c}{t} \implies c \in X $.
	\end{enumerate}
\end{defi}

Condition~\ref{eq:ConfigContainement} ensures the accumulation of events. Condition~\ref{eq:OnlyEnabledEvents} ensures that only events that are enabled after $X$ can take place in $Y$. Condition~\ref{eq:correctCausalityUpdate} ensures the correct update of the current causality relation with respect to this transition. Condition~\ref{eq:NCNoConcurrency} prevents race conditions; it forbids the occurrence of an adder and a dropper of the same still relevant causal dependency within one transition. Condition~\ref{eq:NotToMuchConcurrency} ensures that DCESs generalizes the GESs, it prevents the concurrent occurrence of an adder and its target since they are not independent. The one exception is when the cause has already occurred; in that case, the adder does not change the target's predecessors.

\begin{defi}
	\label{def:DCESConf}
Let $\Delta = (E,\caus{in}, \shrinkingCausality,
	\growingCausality) $ and $\Delta' = (E',\caus{in}', \shrinkingCausality',
	\growingCausality') $  be two DCESs. We call them transition sequence equivalent if they allow for the same transition sequences with respect to the events and write $\eqts{\Delta}{\Delta'}$, \ie for any sequence $\emptyset=X_0, X_1, \dots, X_n$ with $X_i\subset X_{i+1}$ and $X_n\in E\cap E'$, we have current causality relations $\caus{X_0}=\caus{in},\caus{X_1},\dots,\caus{X_n}$ such that $\transDC{(X_i,\caus{X_i})}{(X_{i+1},\caus{X_{i+1}})}$ in $\Delta$, iff we have current causality relations $\caus{X_0}'=\caus{in}',\caus{X_1}'\dots,\caus{X_n}'$ such that $\transDC{(X_i,\caus{X_i}')}{(X_{i+1},\caus{X_{i+1}'})}$ in $\Delta'$.
\end{defi}

As visualized by the DCES in Figure~\ref{fig:DCESorder}, the order of events may be relevant for the behavior of a DCES. Therefore, we use state transition equivalence to compare DCESs. When we compare a DCES to GESs, SESs, RCESs, and EBES, then the state will be unique with respect to the configuration (\ie the current causality relation only depends on the configuration and not on the order), thus we can use the transition equivalence $ \transEq{}{} $ from Definition~\ref{def:transEq} for those comparisons.

\subsection{DCESs versus RCESs}

DCESs and RCESs are incomparable.

On the one hand, RCESs can express the disabling of an event $ c $ after a conjunction of events $ a $ and $ b $ as visualized by $ \rho_{\gamma} $ of Figure~\ref{fig:transgraphs}, whereas DCESs can model only disabling after single events.

\begin{lem}
	\label{lma:RCESNotinDCES}
	There is no transition-equivalent DCES to $ \rho_{\gamma} $ of Figure~\ref{fig:transgraphs}.
\end{lem}

\begin{proof}
	Assume $ \Delta = \left( E, \enabOp,\shrinkingCausality, \growingCausality \right) $ such that $ \transEq{\Delta}{\rho_{\gamma}} $. Then $ \configurations{\Delta} = \configurations{\rho_{\gamma}} $.
	By Definition~\ref{def:DCTransition} and  because of the configurations $ \Set{ a }, \Set{ b }, \Set{ c } \in \configurations{\rho_{\gamma}} $, there are no initial causes for $ a $, $ b $, or $ c $, \ie $ \enabOp \cap \Set{ \enab{x}{y} \mid y \in \Set{ a, b, c } } = \emptyset $. Note that the relation $\shrinkingCausality$ cannot disable events.
	Finally, because of the configurations $ \Set{ a, c }, \Set{ b, c } \in \configurations{\rho_{\gamma}} $, neither $ a $ nor $ b $ can add a cause  to $ c $.
	Thus we have the state $(\Set{a,b}, \caus{\Set{a,b}})$, where $\Set{e'\mid (e',c)\in\caus{\Set{a,b}}}=\emptyset$.
	But then, by Definition~\ref{def:DCTransition}, $ \transDC{(\Set{ a, b },\caus{\Set{a,b}})}{(\Set{ a, b, c },\caus{\Set{a,b,c}})} $ for some current causality relations $\caus{\Set{a,b}}$ and $\caus{\Set{a,b,c}}$.
	Since $ \neg \left( \transRC{\Set{ a, b }}{\Set{ a, b, c }} \right) $, this violates our assumption, \ie there is no DCES that is transition equivalent to $ \rho_{\gamma} $.
\end{proof}

On the other hand, RCESs cannot distinguish the order of occurrences of events as it is done for the order of $ a $ and $ d $ in the DCES in Figure~\ref{fig:DCESorder}. Hence, both structures are incomparable.

\begin{thm}
	\label{thm:DCESvsRCES}
	DCESs and RCESs are incomparable.
\end{thm}

\begin{proof}
	The Theorem follows from Lemma~\ref{lma:RCESNotinDCES} and the order insensitivity of the RCESs: There is no RCES with the same behavior as the DCES in Figure~\ref{fig:DCESorder} since the transition behavior of RCESs is configuration based (Definition~\ref{def:RCES}) and therefore it simply cannot behave differently \wrt the order of events in the configuration $ \Set{a,d} $.
\end{proof}

\subsection{DCESs versus GESs and SESs}

By construction, DCESs are at least as expressive as GESs and SESs. To embed a GES (or SES) into a DCES, it suffices to choose $ \shrinkingCausality = \emptyset $ (or first $ \growingCausality = \emptyset $ and then add those tuples that encode the binary conflict).

\begin{defi}
	\label{def:InclusionIntoDCES}
	Let $ \sigma = \left( E, \confOp, \enabOp, \shrinkingCausality \right) $ be a SES. \\
        The embedding of $\sigma$  into a DCES is given by $ \emb{\sigma} = (E\uplus{\{e_{\text{imp}}\}},{\enabOp}\cup{\{(e_{\text{imp}},e_{\text{imp}})\}},\shrinkingCausality,\growingCausality) $, where $e_{\text{imp}}$ is a fresh (and impossible) event and
${\growingCausality}=\{(a,e_{\text{imp}},b),(b,e_{\text{imp}},a)\mid (a,b)\in\confOp)\}$.
	
	Similarly, let $ \gamma = (E,\enabOp,\growingCausality) $ be a GES. \\Then, its embedding into a DCES is given by $\emb{\gamma}=(E,\enabOp,\emptyset,\growingCausality)$.
\end{defi}

In the case of a SES, since $e_{\text{imp}}$ is fresh, there are no droppers for the newly added dependencies and thus $\emb{\sigma}$ is a DCESs in which the order of modifiers does not matter (\cf Definition~\ref{def:SingleStateDC}). Similarly, $\emb{\gamma}$ is a DCESs in which the order of modifiers does not matter. In the following,
we show that SESs (respectively GESs) and their embeddings are transition equivalent.
Note that, in the proofs, we refer to the technical concept of Single State Dynamic Causality ESs (SSDCs), which are only defined in the Appendix.

\begin{lem}
	\label{lma:SESinDCES}
	\label{lma:GESinDCES}
	Let $\mu$ be a GES or SES, then we have $\transEq{\emb{\mu} }{\mu}$.
\end{lem}

\begin{proof}
	Let $\mu$ be a GES.  We show $ \transGa{(X,\caus{X})}{(Y,\caus{Y})} \iff \transDC{(X,\caus{X})}{(Y,\caus{Y})} $. Note that Condition~\ref{eq:NotToMuchConcurrency} in Definition~\ref{def:DCTransition} always holds for $ \emb{\mu} $, since it is a SSDC. The other conditions in Definition~\ref{def:DCTransition} are exactly as in Definition~\ref{def:GEStransalt} (since $\dropped{X,Y}$ is always empty).
	
	Let now $\mu$ be a SES. We prove $\transSa{(X,\caus{X_s})}{(Y,\caus{Y_s})}\iff \transDC{(X,\caus{X}')}{(Y,\caus{Y}')}$ where ${\caus{X}'}={(\incaus'\setminus\dropped{\emptyset,X})\cup\added{\emptyset,X}}$, ${\caus{Y}'}={(\incaus'\setminus\dropped{\emptyset,Y})\cup\added{\emptyset,Y}}$, and ${\incaus'}={{\incaus}\cup{\{(e_{\text{imp}},e_{\text{imp}})\}}}$, by Lemma~\ref{lma:SingleCausalState} and Definition~\ref{def:InclusionIntoDCES}.
	
	We consider $\Longrightarrow$ first. Definition~\ref{def:SEStransalt} ensures Condition~\ref{eq:ConfigContainement}. For Condition~\ref{eq:OnlyEnabledEvents}, consider ${\caus{X}'}={(\incaus'\setminus\dropped{\emptyset,X})\cup\added{\emptyset,X}}$ since $e_{\text{imp}}$ is not in the set of events of $\mu$, we only have to argue about the added part. But since $\added{\emptyset,X}=\{(e_{\text{imp}},a)\mid\exists b\in X\logdot \conf{a}{b}\}$, we only add dependencies for events in conflict to events in $X$. 
        With Definition~\ref{def:SEStransalt}, then Condition~\ref{eq:OnlyEnabledEvents} holds.
        For Condition~\ref{eq:correctCausalityUpdate}, we have to show $\caus{Y}'=(\caus{X}'\setminus\dropped{X,Y})\cup\added{X,Y}$, which unfolds to $ \caus{Y}' \; = \left( \left( \left( \incaus' \setminus \dropped{\emptyset,X} \right) \cup \added{\emptyset,X} \right) \setminus \dropped{X,Y} \right) \cup \added{X,Y} $; since $\emb{\mu}$ is a SSDC, we can reorder and simplify and obtain $\caus{Y}'=(\incaus'\setminus\dropped{\emptyset,Y})\cup\added{\emptyset,Y}$, wherefore this condition holds. Condition~\ref{eq:NCNoConcurrency} holds because $\emb{\mu}$ is a SSDC. Condition~\ref{eq:NotToMuchConcurrency} holds, because all $\adds{a}{e_{\text{imp}}}{b}$ result from conflicts $\conf{a}{b}$ in $\mu$ and thus, $Y$ cannot contain both.
	
	We consider $\Longleftarrow$ by proving the conditions of Definition~\ref{def:SEStransalt}: $X\subset Y$ follows from Condition~\ref{eq:ConfigContainement}.
	Assume that $Y$ is not conflict-free. Then, there are $a,b\in Y$ with $\conf{a}{b}$. Because $a,b \in Y$, there is a transition sequence in $\emb{\mu}$ leading to $a$ and $b$. But, by Definition~\ref{def:InclusionIntoDCES} $\adds{a}{e_{\text{imp}}}{b}$ (and $\adds{b}{e_{\text{imp}}}{a}$) and because of Condition~\ref{eq:NotToMuchConcurrency}, they cannot occur in the same transition.
	Since $e_{\text{imp}}$ is impossible and since there are no droppers for the dependencies between $e_{\text{imp}}$ and $a$ (or $b$), the first occurrence of $a$ or $b$ will add a dependency from $e_{\text{imp}}$ to the respective other event.
	Without loss of generality, assume $a$ happens first.
	So, $a$ will add the dependency $e_{\text{imp}} \to b$ to the causality relation (by Condition~\ref{eq:correctCausalityUpdate}) and, as argued above, this dependency cannot be dropped later.
	Therefore, $b$ can no longer occur (by Condition~\ref{eq:OnlyEnabledEvents}). Thus, $Y$ is conflict-free.
	The condition $ \forall e \in Y \setminus X \,. \left\{e'\mid (e',e)\in{\caus{X}}\right\} \subseteq X $ is a slightly more general version of Condition~\ref{eq:OnlyEnabledEvents}. And the last condition holds because of ${\caus{Y}}={\incaus\setminus\dropped{\emptyset,Y}}$ and by Lemma~\ref{lma:SEScaus}.
\end{proof}

We conclude that the extension from GESs and SESs to DCESs adds expressive power.

\begin{thm}
	\label{thm:DCESvsGSES}
	DCESs are strictly more expressive than GESs and SESs.
\end{thm}

\begin{proof}
        With Theorem~\ref{thm:DCESvsRCES}, DCESs are incomparable to RCESs.
        With Theorem~\ref{thm:GESinRCESstrictly} and \ref{thm:SESinRCESstrictly}, RCESs are strictly more expressive than GESs and SESs.
        With Lemma~\ref{lma:GESinDCES}, DCES is at least as expressive as SESs and GESs.

        Now, if DCESs were \emph{not strictly} more expressive than SESs and GESs, then DCESs would become comparable to RCESs, which they are not.

\end{proof}

\subsection{DCESs versus EBESs}

To show that DCESs are strictly more expressive than EBESs, we use the disabling of GESs and the disjunctive causality of SESs. More precisely, EBESs cannot model the disjunctive causality without a conflict. As counterexample, we use the embedding of the SES $\sigma_\xi$ in Figure~\ref{fig:counterExamples}.

\begin{lem}
\label{lma:DCESninEBES}
	There is no EBES with the same configurations as the DCES $\emb{\sigma_\xi}$, where $\sigma_\xi$ is the SES given in Figure~\ref{fig:counterExamples}.
\end{lem}
	
\begin{proof}
	We consider the embedding $\emb{\sigma_\xi}$ of the SES $\sigma_\xi$ in Figure~\ref{fig:counterExamples}, which models disjunctive causality.
	According to the Definitions~\ref{def:SEStrans}, because $\neg(a\#c)$ and $\ic{a}=\ic{c}=\emptyset$, we have $\transS{\emptyset}{\{a,c\}}$ and so $\{a,c\}\in\configurations{\sigma_\xi}$.
	Furthermore, there is no transition $\transS{\emptyset}{\{b\}}$, because $\ic{b}=\{a\}$.
	Yet, there are transitions $\transS{\{a\}}{\{a,b\}}$ and $\transS{\{c\}}{\{c,b\}}$, because $\ic{b}\setminus\dc{\{a\}}{b}\subseteq\{a\}$ ($\ic{b}\setminus\dc{\{c\}}{b}\subseteq\{c\}$ respectively).
	The transitions are translated to the embedding according to Lemma~\ref{lma:SESinDCES} and Definition~\ref{def:DCESConf}; the same holds for the configurations. 
	
	If we assume that there is some EBES $\xi$ with the configurations $ \emptyset $, $ \{a\} $, $ \{c\} $, $ \{a,c\} $, $ \{a,b\} $, $ \{b,c\} $ and $\{a,b,c\}$, then because of Definition~\ref{def:EBES} and since there is no configuration $\{b\}$, there must be a non-empty bundle $\buEn{X}{b}$ and because of the configurations $\{a,b\},\{b,c\}$, this bundle $X$ must contain $a$ and $c$. Then, the stability condition of Definition~\ref{def:EBES} implies $\disa{a}{c}$ and $\disa{c}{a}$, so $a$ and $c$ are in mutual conflict, which contradicts the assumption that we have $\{a,c\}\in\configurations{\xi}$. Thus, there is no EBES with the same configurations as $\emb{\sigma_\xi}.$
\end{proof}

On the other hand, we formally define an encoding of an EBES into a DCES in Definition~\ref{def:EBES2DCES} in the
Appendix, where disabling uses self-loops of target events, while droppers use auxiliary impossible events (to ensure that they do not intervene with the disabling).
Besides, we construct posets for the configurations of the encoding, and compare them with those of the original EBES.
In this way, we prove that DCESs are at least as expressive as EBESs. But since EBESs cannot model the disjunctive causality---without a conflict---of SESs, DCESs are strictly more expressive than EBESs.

\begin{thm}
	\label{thm:EBESslessExpThanDCESs}
	DCESs are strictly more expressive than EBESs.
\end{thm}

\begin{proof}
	Follows from the Lemmata~\ref{lma:DCESninEBES} and \ref{lma:EBESintoEBDC}.
\end{proof}

\subsection{HDESs}\label{sec:HDES}
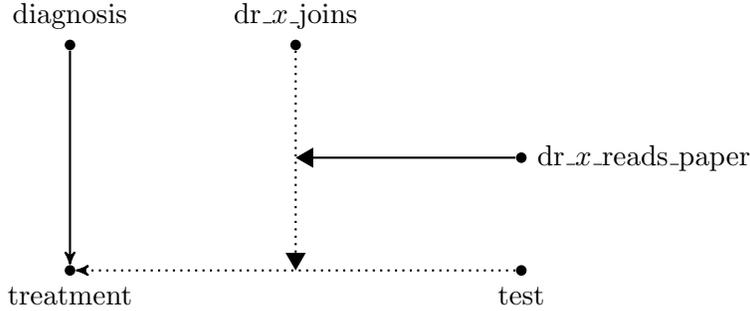
\begin{figure}
\centering
\begin{tikzpicture}[node distance=3cm,>=stealth']
			\node (diag) [circle,fill,minimum size=4pt,inner sep=0pt, label=above:diagnosis]{};
			\node (treat) [below of=diag,circle,fill,minimum size=4pt,inner sep=0pt, label=below:treatment]{};
			\node (docx) [right of=diag,circle,fill,minimum size=4pt,inner sep=0pt, label=above:dr\_{$x$}\_joins]{};
			\node (aux)[right of=docx,,inner sep=0pt ]{};
			\node (paper) [ below of=aux, node distance=1.5cm,circle,fill,minimum size=4pt,inner sep=0pt, label=right:dr\_{$x$}\_reads\_paper]{};
	
			\node (test) [below of=paper, node distance=1.5cm,circle,fill,minimum size=4pt,inner sep=0pt, label=below:test]{};
			\draw [->, thick] (diag) -- (treat);
			\draw [->,dotted,thick] (test) to node {} (treat);
			\draw[adding]	(paper) -- (3,-1.5) ;		
			\draw[adding, dotted]	(docx) -- (3,-3);
\end{tikzpicture}
\caption{An HDES example: Dotted arrows denote initially absent dependencies or initially absent dynamic rules.}
\label{fig:HDES}
\end{figure}

\emph{Higher order Dynamic Cau\-sa\-li\-ty ESs} (HDESs) are a recent generalization of DCESs \cite{HDES}. Figure~\ref{fig:HDES} presents a small example of a second-order addition: After a diagnosis, a treatment becomes possible. Now, a doctor might join the treatment team and nothing changes. However, if the doctor reads|before joining|a medical paper, in which an additional test before the treatment is suggested, the joining of the doctor leads to an additional precondition, namely the test for the treatment. This behavior is modeled with a second-order rule, $\addrule{\text{dr\_$x$\_reads\_paper}}{\adds{\text{dr\_$x$\_joins}}{\text{test}}{\text{treatment}}}$, and the zeroth-order rule  $\text{diagnosis}\rightarrow\text{treatment}$.

In general, a HDES consists of a set of events $E$ and a set of rules of arbitrary order of dynamics (also arbitrary combination of adding and deletion is allowed).
As a second generalization, it is possible to have sets of events as modifiers (in contrast to the singletons in the DCES approach), this allows for conjunctive adding and dropping where in DCESs only disjunctive modifications were possible. Accordingly, a HDES consists of a set of events and set of dynamic rules (causality is considered as zeroth-order dynamics), whereby this rule set is updated in each transition.

The strict inclusion \wrt the expressive power of RCESs into HDESs is proven in \cite{HDES}. The proof uses configuration structures that were introduced in~\cite{GG90}  and later relaxed in~\cite{ConfStr}. In~\cite{HDES}, a definition in-between these two is used.

\begin{defi}[\cite{HDES}]\label{def:ConfStr}
A \emph{configuration structure} is a pair $\mathcal{C}=(E,C)$ with $E$ a set and $C\subseteq\Powerset{E}$ a collection of subsets. We call the elements of $E$ \emph{events} and the elements of $C$ \emph{configurations}.
 For $x,y$ in $C$ we write $x\rightarrow_C y$ if $x\subseteq y$ and
\[\forall Z\logdot x\subseteq Z\subseteq y\Rightarrow Z\in C.\]
The relation $\rightarrow_C$ is called the \emph{step transition relation}.
\end{defi}

In \cite{HDES}, the authors define a map $ \HDES{\cdot} $ from any configuration structure (\wrt Definition~\ref{def:ConfStr}) into a HDES and prove that both have the same transition behavior.

\begin{thm}[\cite{HDES}]	\label{the:HDES}
Let $\mathcal{C}$ be a configuration structure.
 Then the HDES $\Delta=\HDES{\mathcal{C}}$ is transition equivalent to $\mathcal{C}$.
\end{thm}

Hence, DCESs and RCESs can both be embedded in HDESs. Since DCESs and RCESs are incomparable, then both kinds of structures are strictly less expressive than HDESs.

\section{Conclusions}
\label{sec:conclusion}

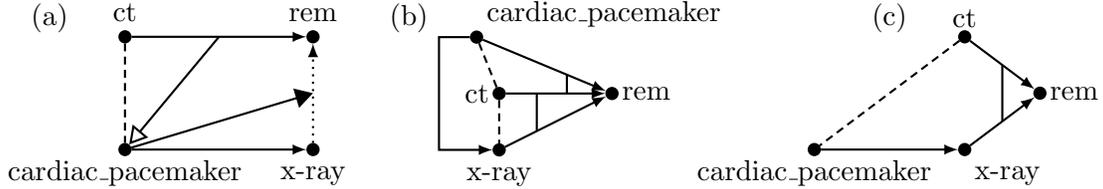
\begin{figure}
	\centering
	\begin{tikzpicture}
	\event{ct}{0}{1.5}{above}{\vphantom{Iq}ct};
	\event{cp}{0}{0}{below}{cardiac\_pacemaker};
	\event{r}{2.5}{1.5}{above}{\vphantom{Iq}rem};
	\event{x}{2.5}{0}{below}{\vphantom{Iq}x-ray};
	\node (a) at (-1, 1.75) {(a)};
	\draw[conflictPES] (ct) edge (cp);
	\draw[enablingAbsent] (x) -- (r);
	\draw[enablingPES] (ct) edge (r);
	\draw[enablingPES] (cp) edge (x);
	\draw[adding] (cp) -- (2.5, 0.75);
	\draw[dropping]	(1.25, 1.5) -- (cp);
	\end{tikzpicture}\hspace{0.2cm}
	\begin{tikzpicture}
	\event{ct}{1.5}{0}{left}{ct};
	\event{cp}{1.2}{0.75}{above right}{cardiac\_pacemaker};
	\event{r}{3}{0}{right}{rem};
	\event{x}{1.5}{-0.75}{below}{\vphantom{Iq}x-ray};
	\node (b) at (.3, 1) {(b)};
	\draw[enablingPES] (ct) edge (r);
	\draw[enablingPES] (cp) edge (r);
	\draw[enablingPES] (x) edge (r);
	\draw[-latex, thick] (cp) -- (.7,0.75) -- (.7,-0.75) -- (x);
	\draw[thick] (2.4, 0) -- (2.4, 0.25);
	\draw[thick] (2, 0) -- (2, -0.5);
	\draw[conflictPES] (ct) edge (cp);
	\draw[conflictPES] (ct) edge (x);
	\end{tikzpicture}\hspace{-0.7cm}
	\begin{tikzpicture}
	\event{ct}{2}{1.5}{above}{ct};
	\event{cp}{0}{0}{below}{cardiac\_pacemaker};
	\event{r}{3}{0.75}{right}{rem};
	\event{x}{2}{0}{below right}{\vphantom{Iq}x-ray};
	\node (c) at (1, 1.75) {(c)};
	\draw[enablingPES] (ct) edge (r);
	\draw[enablingPES] (cp) edge (x);
	\draw[enablingPES] (x) edge (r);
	\draw[conflictPES] (ct) edge (cp);
	\draw[thick] (2.5, 0.375) -- (2.5, 1.125);
	\end{tikzpicture}
	\vspace*{-1em}
	\caption{A DCES, a BES, and a DES modeling the medical example.}
	\label{fig:medical examples}
\end{figure}

We study the idea that causality may change during system runs of event structures. For this, we enhance a simple type of ESs---the rPESs---by means of additional relations capturing the changes in events' dependencies, driven by the occurrence of other events.

First, in \S\ref{sec:SES}, we limit our concern to the case where dependencies can only be dropped. We call the new resulting event structure Shrinking Causality ES (SES). We show that the exhibited dynamic causality can be expressed through a completely static perspective, by proving equivalence between SESs and DESs. By such a proof, we do not only show the expressive power of our new ES, but also the big enhancement in expressive power (\wrt rPESs) gained by adding only this one relation.

Later on, in \S\ref{sec:GES}, we study the complementary style where dependencies can be added, resulting in Growing Causality ES (GES). We show that this variant of dynamic causality can model both permanent and temporary disabling. Besides, it can be used to resolve conflicts and, furthermore, to force conflicts. Unlike the SESs, the GESs are not directly comparable to other types of ESs from the literature, except for PESs (and rPESs). GESs cannot model disjunction in the enabling relation but they are able to express conditions for causal relationships.

Finally, in \S\ref{sec:DCES}, we combine both approaches of dynamicity with a new type of event structures, which we call Dynamic Causality ES (DCES).
Therein, dependencies can be both added and dropped. For this new type of ESs, the following two---possibly surprising---facts can be observed:
\begin{enumerate*}[label=(1)]
	\item There are types of ESs that are incomparable to both SESs and GESs, but that are comparable to (here: strictly less expressive than) DCESs, i.e. the combination of SESs and GESs; one such type is EBESs.
	\item Though SESs and GESs are strictly less expressive than RCESs, their combination---the newly defined DCESs---is incomparable to RCESs, and  incomparable to or even strictly more expressive as any other type of configuration-based ESs.
\end{enumerate*}

To highlight the pragmatic advantages of dynamic-causality ESs over other classes of ESs, we refer to the example in the Introduction. Reichert et al.\ in \cite{Reichert:ForwardBackwardJumps} emphasize that the model of such processes should distinguish between the regular execution path and the exceptional one. Accordingly, they define two labels, \emph{REGULAR} and \emph{EXCEPTIONAL}, to be assigned to tasks.
Figure~\ref{fig:medical examples}(a) shows a DCES model of our example, where \textit{rem} represents the remainder of the treatment process, and \textit{ct} represents the computer tomography.
The initial causality in a DCES \eg \textit{ct} $\rightarrow$ \textit{rem} corresponds to the regular path of a process, while the changes carried by modifiers \eg \emph{cardiac\_pacemaker} correspond to the exceptional one.
Other static-causality ESs like a BES and a DES can model the fact that either the computer tomography XOR the X-ray is needed, as shown in Figure~\ref{fig:medical examples}(b) and \ref{fig:medical examples}(c). The same can be done by an equivalent RCES. However, we argue that none of these models can distinguish between regular and exceptional paths.

Thus, our main contributions are:
\begin{enumerate*}
	\item We provide a formal model that allows us to express dynamicity in causality. Using this model, we enhance the rPESs yielding SESs, GESs and DCESs.
	\item We show the equivalence of SESs and DESs.
	\item We show the incomparability of GESs to many other types of ESs.
	\item We show that DCESs are strictly more expressive than EBESs and thus strictly more expressive than many other existing types of ESs.
	\item We show that DCESs are incomparable to RCESs.
	\item The new model succinctly supports modern work-flow management systems.
\end{enumerate*}

In \cite{VaraccaExtrusion}, Crafa et al.\ defined an Event Structure semantics for the $\pi$-calculus based on Prime ESs. Since the latter do not allow for disjunctive causality required for their purpose, and in order to avoid duplications of events, they extended Prime ESs with a set of bound names, and altered the configuration definition to allow for such disjunction. With SESs---that can express disjunctive causality---this problem could possibly be addressed more naturally without copying events. Here, higher-order dynamicity (\cf \cite{HDES}) might help to deal with the instantiation of variables caused by communications involving bound names.

Up to now, we limit the execution of DCESs such that an interleaving between adders and droppers of the same causal dependency is forced. As future work, we intend to study the case where modifiers of the same dependency can occur concurrently---read: at the very same instant of time---in DCESs. 
Similarly, we want to investigate the situation of concurrent occurrence of an adder and its target in GESs.

More recently, \cite{HDES} studies the adding and dropping by sets of events and higher-order dynamics, \ie events that may change the role of events to adders, droppers or back to normal events.
We are currently developing a web tool (\href{http://hdes.mtv.tu-berlin.de}{http://hdes.mtv.tu-berlin.de}) for the modeling and simulation of DCESs and HDESs, containing the example DCESs and HDESs presented in this paper.

The set of possible changes in our newly defined ESs must still be declared statically. In \cite{Arbach:Thesis}, the author investigated the idea of \emph{evolutionary} ESs by supporting ad-hoc changes, such that new dependencies as well as events can be added to a structure.

\bibliographystyle{plain}
\bibliography{dynamicCausality}

\appendix

\section{Shrinking Causality}

In SESs both notions of configurations, traced-based and transition-based, coincide; and in different situations, the more suitable one can be used.

\begin{lem}\label{lem:SESconf}
	Let $ \sigma $ be a SES. Then $ \configTraces{\sigma} = \configurations{\sigma} $.
\end{lem}

\begin{proof}
	Let $ \sigma = \left( E, \confOp, \enabOp, \shrinkingCausality \right) $.
	By Definition~\ref{def:SEStraceDef}, $ C \in \configTraces{\sigma} $ implies that there is some $ t = e_1 \cdots e_n $ such that $ \overline{t} \subseteq E $, $ \forall 1 \leq i, j \leq n \logdot \neg \left( \conf{e_i}{e_j} \right) $, $ \forall 1 \leq i \leq n \logdot \left( \ic{e_i} \setminus \dc{\overline{t_{i - 1}}}{e_i} \right) \subseteq \overline{t_{i - 1}} $, and $ C = \overline{t} $.
	Hence, by Definition~\ref{def:SEStrans}, $ \transS{\overline{t_i}}{\overline{t_{i + 1}}} $ for all $ 1 \leq i \leq n $ and $ \transS{\emptyset}{\Set{ e_1 }} $.
	Thus, by Definition~\ref{def:SESconf}, $ C \in \configurations{\sigma} $.
	
	By Definition~\ref{def:SESconf}, $ C \in \configurations{\sigma} $ implies that there are $ X_1, \ldots, X_n \subseteq E $ such that $ \transS{\transS{\transS{\emptyset}{X_1}}{\ldots}}{X_n} $ and $ X_n = C $.
	Then, by Definition~\ref{def:SEStrans}, we have:
	\begin{align}
		& \emptyset \subseteq X_1 \subseteq X_2 \subseteq \ldots \subseteq X_n \subseteq E & \tag{C1} \label{eq:C1}\\
		& \forall e, e' \in X_n \logdot \neg \left( \conf{e}{e'} \right) & \tag{C2} \label{eq:C2}\\
		& \forall e \in X_1 \logdot \left( \ic{e} \setminus \dc{\emptyset}{e} \right) \subseteq \emptyset & \tag{C3} \label{eq:C3}\\
		& \forall 1 \leq i < n \logdot \forall e \in X_{i + 1} \setminus X_i \logdot \left( \ic{e} \setminus \dc{X_i}{e} \right) \subseteq X_i & \tag{C4} \label{eq:C4}
	\end{align}
	Let $ X_1 = \Set{ e_{1, 1}, \ldots, e_{1, m_1} } $ and $ X_i \setminus X_{i - 1} = \Set{ e_{i, 1}, \ldots, e_{i, m_i} } $ for all $ 1 < i \leq n $.
	Then by Definition~\ref{def:SEStraceDef}, $ t = e_{1, 1} \cdots e_{1, m_1} \cdots e_{n, 1} \cdots e_{n, m_n} = e_1' \cdots e_k' $ is a trace such that $ \overline{t} \subseteq E $ (because of \eqref{eq:C1}), $ \neg \left( \conf{e_i'}{e_j'} \right) $ for all $ 1 \leq i, j \leq k $ (because of \eqref{eq:C1} and \eqref{eq:C2}), for all $ 1 \leq i \leq k $ and all $ 1 \leq j \leq m_i $ we have $ \left( \ic{e_{i, j}} \setminus \dc{\overline{t_{i - 1}}}{e_{i, j}} \right) \subseteq \overline{t_{i - 1}} $ (because of \eqref{eq:C3} and \eqref{eq:C4}), and $ \overline{t} = C $ (because $ X_n = C $).
	Thus $ C \in \configTraces{\sigma} $.
\end{proof}

Moreover the following technical Lemma relates transitions and the extension of traces by causally independent events.

\begin{lem}
	\label{lem:SEStransTraces}
	Let $ \sigma = \left( E, \confOp, \enabOp, \shrinkingCausality \right) $ be a SES and $ X, Y \in \configurations{\sigma} $.\\
	Then $ \transS{X}{Y} $ iff there are $ t_1 = e_1 \cdots e_n, t_2 = e_1 \cdots e_n e_{n + 1} \cdots e_{n + m} \in \traces{\sigma} $ such that $ X = \overline{t_1} $, $ Y = \overline{t_2} $, and $ \forall e, e' \in Y \setminus X \logdot \left( \ic{e} \setminus \dc{X}{e} \right) \subseteq X $.
\end{lem}

\begin{proof}
	By Definition~\ref{def:SEStraceDef} and Lemma~\ref{lem:SESconf}, $ X \in \configurations{\sigma} $ implies that there is a trace $ t_1 = e_1 \cdots e_n \in \traces{\sigma} $ such that $ X = \overline{t_1} $.

	If $ \transS{X}{Y} $, then by Definition~\ref{def:SEStrans}, we have $ X \subseteq Y $, $ \forall e, e' \in Y \logdot \neg ( \conf{e}{e'} ) $, and\linebreak
	$ \forall e \in Y \setminus X \logdot ( \ic{e} \setminus \dc{X}{e} ) \subseteq X $. Then by Definition~\ref{def:SEStraceDef}, $ t_2 = e_1 \cdots e_n e_{n + 1} \cdots e_{n + m} \in \traces{\sigma} $ and $ Y = \overline{t_2} $ for an arbitrary linearization $ e_{n + 1} \cdots e_{n + m} $ of the events in $ Y \setminus X $, \ie with $ \Set{ e_{n + 1}, \ldots, e_{n + m} } = Y \setminus X $ such that $ e_{n + i} \neq e_{n + j} $ whenever $ 1 \leq i, j, \leq m $ and $ i \neq j $.
	
	If there is a trace $ t_2 = e_1 \cdots e_n e_{n + 1} \cdots e_{n + m} \in \traces{\sigma} $ such that $ Y = \overline{t_2} $ and $ \forall e, e' \in Y \setminus X \logdot \left( \ic{e} \setminus \dc{X}{e} \right) \subseteq X $ then $ X \subseteq Y $. Moreover, by Definition~\ref{def:SEStraceDef}, $ t_2 \in \traces{\sigma} $ implies $ \forall e, e' \in Y \logdot \neg \left( \conf{e}{e'} \right) $. Thus by Definition~\ref{def:SEStrans}, $ \transS{X}{Y} $.
\end{proof}

Note that the condition $ \forall e, e' \in Y \setminus X \logdot \left( \ic{e} \setminus \dc{X}{e} \right) \subseteq X $ states that the events in $ Y \setminus X $ are causally independent from each other.

In Definition~\ref{def:DESintoSES} we provide an encoding of a DESs into a SESs.
Of course it can be criticized that the encoding adds events (although they are fresh and impossible). But as the following example---with more bundles than events---shows it is not always possible to translate a DES into a SES without additional impossible events.

\begin{lem}
	\label{lem:DESninSES}
	There are DESs $ \delta = \left( E, \confOp, \buEnOp \right) $, as \eg $ \delta = \left( \Set{ a, b, c, d, e }, \emptyset, \buEnOp \right) $ with $ \buEnOp \; = \Set{ \buEn{\Set{ x, y }}{e} \mid x, y \in \Set{ a, b, c, d } \land x \neq y } $, that cannot be translated into a SES $ \sigma = \left( E, \confOp', \enabOp, \shrinkingCausality \right) $ such that $ \traces{\delta} = \traces{\sigma} $.
\end{lem}

\begin{proof}
	Assume a SES $ \sigma = \left( E, \confOp, \enabOp, \shrinkingCausality \right) $ such that $ E = \{ a, b, c, d,$ $ e \} $ and $ \traces{\sigma} = \traces{\delta} $.
	According to \S\ref{sec:DES}, $ \traces{\delta} $ contains all
	sequences of distinct events of $ E $ such that $ e $ is not the first, second, or third event, \ie for $ e $ to occur in a trace it has to be preceded by at least three of the other events.
	Since by Definition~\ref{def:SEStraceDef} conflicts cannot be dropped, $ \traces{\sigma} = \traces{\delta} $ implies $ \confOp' = \emptyset $.
	Moreover, since $ e $ has to be preceded by at least three other events that can occur in any order, $ \enabOp $ has to contain at least three initial causes for $ e $. Without loss of generality let $ \enab{a}{e} $, $ \enab{b}{e} $, and $ \enab{c}{e} $.
	Because of the traces $ abd, acd \in \traces{\delta} $, we need the droppers \drops{d}{b}{e} and \drops{d}{c}{e}. Then $ ad \in \traces{\sigma} $ but $ ad \notin \traces{\delta} $.
	In fact if we fix $ E = \Set{ a, b, c, d, e } $ there are only finitely many different SESs $ \sigma = \left( E, \confOp', \enabOp, \shrinkingCausality \right) $ and for none of them $ \traces{\delta} = \traces{\sigma} $ holds.
\end{proof}

Note that the above lemma implies that no encoding of the above DES can result into a SES with the same events such that the DES and its encoding have same configurations or posets.

\section{Alternative Partial Order Semantics in DES and SES}
\label{app:partialOrderSemantics}

To show that DES and SES are not only behavioral equivalent ES models but are also very closely related at the structural level we consider the remaining four intensional partial order semantics for DES of \cite{Langerak97causalambiguity}.

Liberal causality is the least restrictive notion of causality in \cite{Langerak97causalambiguity}. Here each set of events from bundles pointing to an event $ e $ that satisfies all bundles pointing to $ e $ is a cause.

\begin{defi}[Liberal Causality]
	Let $ \delta = \left( E, \confOp, \buEnOp \right) $ be a DES, $ e_1 \cdots e_n $ one of its traces, $ 1 \leq i \leq n $, and $ \buEn{X_1}{e_i}, \ldots, \buEn{X_m}{e_i} $ all bundles pointing to $ e_i $.
	A set $ U $ is a cause of $ e_i $ in $ e_1 \cdots e_n $ if
	\begin{itemize}[noitemsep]
		\item $ \forall e \in U \logdot \exists 1 \leq j < i \logdot e = e_j $,
		\item $ U \subseteq \left( X_1 \cup \ldots \cup X_m \right) $, and
		\item $ \forall 1 \leq k \leq m \logdot X_k \cap U \neq \emptyset $.
	\end{itemize}
	Let $ \posetsLib{t} $ be the set of posets obtained this way for a trace $ t $.
\end{defi}

Bundle satisfaction causality is based on the idea that for an event $ e $ in a trace each bundle pointing to $ e $ is satisfied by exactly one event in a cause of $ e $.

\begin{defi}[Bundle Satisfaction Causality]
	Let $ \delta = \left( E, \confOp, \buEnOp \right) $ be a DES, $ e_1 \cdots e_n $ one of its traces, $ 1 \leq i \leq n $, and $ \buEn{X_1}{e_i}, \ldots, \buEn{X_m}{e_i} $ all bundles pointing to $ e_i $.
	A set $ U $ is a cause of $ e_i $ in $ e_1 \cdots e_n $ if
	\begin{itemize}[noitemsep]
		\item $ \forall e \in U \logdot \exists 1 \leq j < i \logdot e = e_j $ and
		\item there is a surjective mapping $ f : \Set{ X_k } \to U $ such that $ f\!\left( X_k \right) \in X_k $ for all $ 1 \leq k \leq m $.
	\end{itemize}
	Let $ \posetsBsat{t} $ be the set of posets obtained this way for a trace $ t $.
\end{defi}

Minimal causality requires that there is no subset which is also a cause.

\begin{defi}[Minimal Causality]
	Let $ \delta = \left( E, \confOp, \buEnOp \right) $ be a DES and let $ e_1 \cdots e_n $ be  one of its traces, $ 1 \leq i \leq n $, and $ \buEn{X_1}{e_i}, \ldots, \buEn{X_m}{e_i} $ all bundles pointing to $ e_i $.
	A set $ U $ is a cause of $ e_i $ in $ e_1 \cdots e_n $ if
	\begin{itemize}[noitemsep]
		\item $ \forall e \in U \logdot \exists 1 \leq j < i \logdot e = e_j $,
		\item $ \forall 1 \leq k \leq m \logdot X_k \cap U \neq \emptyset $, and
		\item there is no proper subset of $ U $ satisfying the previous two conditions.	
	\end{itemize}
	Let $ \posetsMin{t} $ be the set of posets obtained this way for a trace $ t $.
\end{defi}

Late causality contains the latest causes of an event that form a minimal set.

\begin{defi}[Late Causality]
	Let $ \delta = \left( E, \confOp, \buEnOp \right) $ be a DES, $ e_1 \cdots e_n $ one of its traces, $ 1 \leq i \leq n $, and $ \buEn{X_1}{e_i}, \ldots, \buEn{X_m}{e_i} $ all bundles pointing to $ e_i $.
	A set $ U $ is a cause of $ e_i $ in $ e_1 \cdots e_n $ if
	\begin{itemize}[noitemsep]
		\item $ \forall e \in U \logdot \exists 1 \leq j < i \logdot e = e_j $,
		\item $ \forall 1 \leq k \leq m \logdot X_k \cap U \neq \emptyset $,
		\item there is no proper subset of $ U $ satisfying the previous two conditions, and
		\item $ U $ is the latest set satisfying the previous three conditions.	
	\end{itemize}
	Let $ \posetsLat{t} $ be the set of posets obtained this way for a trace $ t $.
\end{defi}

As derived in \cite{Langerak97causalambiguity}, it holds that
\begin{align*}
	\posetsLat{t}, \posetsEar{t} \subseteq \posetsMin{t} \subseteq \posetsBsat{t}
	\subseteq \posetsLib{t}
\end{align*}
for all traces $ t $.
Moreover a behavioral partial order semantics is defined and it is shown that two DESs have the same posets \wrt the behavioral partial order semantics iff they have the same posets \wrt the early partial order semantics iff they have the same traces.

Bundle satisfaction causality is---as the name suggests---closely related to the existence of bundles. In SESs there are no bundles. Of course, as shown by the encoding $ \des{\cdot} $ in Definition~\ref{def:SESintoDES}, we can transform the initial and dropped causes of an event into a bundle. And of course if we do so an SES $ \sigma $ and its encoding $ \des{\sigma} $ have exactly the same families of posets. But because bundles are no native concept of SESs, we cannot directly map the definition of posets \wrt bundle satisfaction to SESs.

To adapt the definitions of posets in the other three cases we have to replace the condition $ U \subseteq \left( X_1 \cup \ldots \cup X_m \right) $ by $ U \subseteq \left( \Set{ e \mid \enab{e}{e_i} \lor \exists e' \in E \logdot \drops{e}{e'}{e_i} } \right) $ and replace the condition $ \forall 1 \leq k \leq m \logdot X_k \cap U \neq \emptyset $ by $ \left( \ic{e_i} \setminus \dc{U}{e_i} \right) \subseteq U $ (as in Definition~\ref{def:SESposets}). The remaining conditions remain the same with respect to traces as defined in Definition~\ref{def:SEStraceDef}.
Let $ \posetsLib{t} $, $ \posetsMin{t} $, and $ \posetsLat{t} $ denote the sets of posets obtained this way for a trace $ t \in \traces{\sigma} $ of a SES $ \sigma $ \wrt liberal, minimal, and late causality. Moreover, let $ \posetsX{x}{\delta} = \bigcup_{t \in \traces{\delta}}{\posetsX{x}{t}} $ and $ \posetsX{x}{\sigma} = \bigcup_{t \in \traces{\sigma}}{\posetsX{x}{t}} $ for all $ x \in \Set{ \operatorname{lib}, \operatorname{bsat}, \operatorname{min}, \operatorname{late} } $.

Since again the definitions of posets in DESs and SESs are very similar the encodings $ \des{\cdot} $ and $ \ses{\cdot} $ preserve families of posets. The proof is very similar to the proofs of the Theorems~\ref{thm:SESintoDES} and \ref{thm:DESintoSES}.

\begin{thm}
	For each SES $ \sigma $ there is a DES $ \delta $, namely $ \delta = \des{\sigma} $, and for each DES $ \delta $ there is a SES $ \sigma $, namely $ \sigma = \ses{\delta} $, such that $ \posetsX{x}{\sigma} = \posetsX{x}{\delta} $ for all $ x \in \Set{ \operatorname{lib}, \operatorname{min}, \operatorname{late} } $.
\end{thm}

\begin{proof}
	The definitions of posets in DESs and SESs \wrt minimal and late causality differ in exactly the same condition and its replacement as the definitions of posets in DESs and SESs \wrt early causality. Thus the proof in these two cases is similar to the proofs of the Theorems~\ref{thm:SESintoDES} and \ref{thm:DESintoSES}.
	
	If $ \sigma = \left( E, \confOp, \enabOp, \shrinkingCausality \right) $ is a SES then, by Lemmas~\ref{lem:SEStoDES} and~\ref{lem:SEStoDES-Semantik}, $ \delta = \des{\sigma} = \left( E, \confOp, \buEnOp \right) $ is a DES such that $ \traces{\sigma} = \traces{\delta} $ and $ \configurations{\sigma} = \configurations{\delta} $.
	If $ \delta = \left( E, \confOp, \buEnOp \right) $ is a DES then, by Lemma~\ref{lem:DEStoSES}, $ \sigma = \ses{\delta} = \left( E, \confOp, \enabOp, \shrinkingCausality \right) $ is a DES such that $ \traces{\delta} = \traces{\sigma} $ and $ \configurations{\delta} = \configurations{\sigma} $.
	In both cases let $ t = e_1 \cdots e_n \in \traces{\sigma} $, $ 1 \leq i \leq n $, and $ \buEn{X_1}{e_i}, \ldots, \buEn{X_m}{e_i} $ be all bundles pointing to $ e_i $.
	
	In the case of liberal causality, for $ U $ to be a cause for $ e_i $ the definition of posets in SESs requires $ U \subseteq \left( \Set{ e \mid \enab{e}{e_i} \lor \exists e' \in E \logdot \drops{e}{e'}{e_i} } \right) $ and $ ( \ic{e_i} \setminus \dc{U}{e_i} ) \subseteq U $.
	The second condition holds iff $ \forall 1 \leq k \leq m \logdot X_k \cap U \neq \emptyset $ as shown in the proofs of the Theorems~\ref{thm:SESintoDES} and \ref{thm:DESintoSES}.
	By the Definitions~\ref{def:SESintoDES} and \ref{def:DESintoSES}, the first conditions holds iff $ U \subseteq \left( X_1 \cup \ldots \cup X_m \right) $.
	So, by the definitions of posets in DESs and SESs \wrt liberal causality, $ \posetsLib{\sigma} = \posetsLib{\delta} $.
\end{proof}

\section{Growing Causality}

As in SESs, both notions of configurations of GESs, traced-based and
transition-based coincide; in different situations the more suitable one can be used.

\begin{lem}
	\label{lma:GESConfigEquivalence}
	Let $ \gamma $ be a GES. Then $ \configTraces{\gamma} = \configurations{\gamma} $.
\end{lem}

\begin{proof}
	Let $ \gamma = \left( E, \enabOp, \growingCausality \right) $.
	
	By Definition~\ref{def:GEStraceDef}, $ C \in \configTraces{\gamma} $ implies that there is some $ t = e_1 \cdots e_n $ such that $ \overline{t} \subseteq E $, $ \forall 1 \leq i \leq n \logdot \forall 1 \leq i \leq n \logdot \left( \ic{e_i} \cup \ac{\overline{t_{i - 1}}}{e_i} \right) \subseteq \overline{t_{i - 1}} $, and $ C = \overline{t} $. Hence by Definition~\ref{def:GESNoConcurTransition}, $ \transG{\overline{t_i}}{\overline{t_{i + 1}}} $ for all $ 1 \leq i \leq n $ and $ \transG{\emptyset}{\Set{ e_1 }} $.	Thus by Definition~\ref{def:GESConfigurations}, $ C \in \configurations{\gamma} $.
	
	By Definition~\ref{def:GESConfigurations}, $ C \in \configurations{\gamma} $ implies that there are $ X_1, \ldots, X_n \subseteq E $ such that $ \transG{\transG{\transG{\emptyset}{X_1}}{\ldots}}{X_n} $ and $ X_n = C $.
	Then by Definition~\ref{def:GESNoConcurTransition}, we have:
	\begin{align}
		& \emptyset \subseteq X_1 \subseteq X_2 \subseteq \ldots \subseteq X_n \subseteq E & \tag{D1} \label{eq:D1}\\
		& \forall e \in X_1 \logdot \left( \ic{e} \cup \ac{\emptyset}{e} \right) \subseteq \emptyset & \tag{D2} \label{eq:D2}\\
		& \begin{array}{l} \forall 1 \leq i < n \logdot \forall e \in X_{i + 1} \setminus X_i \logdot\\ \hspace{3em} \left( \ic{e} \cup \ac{X_i}{e} \right) \subseteq X_i \end{array} & \tag{D3} \label{eq:D3}\\
		& \begin{array}{l} \forall 1 \leq i < n \logdot \forall t, m \in X_{i + 1} \setminus X_i \logdot\forall c\in E \logdot\\ \hspace{3em} \addcause{m}{c}{t} \implies c\in  X_i  \end{array} & \tag{D4} \label{eq:D4}
	\end{align}
	Let $ X_1 = \Set{ e_{1, 1}, \ldots, e_{1, m_1} } $ and $ X_i \setminus X_{i - 1} = \Set{ e_{i, 1}, \ldots, e_{i, m_i} } $ for all $ 1 < i \leq n $.
	Then by Definition~\ref{def:GEStraceDef}, $ t = e_{1, 1} \cdots e_{1, m_1} \cdots e_{n, 1} \cdots e_{n, m_n} = e_1' \cdots e_k' $ is a trace such that $ \overline{t} \subseteq E $ (because of \eqref{eq:D1}),  for all $ 1 \leq i \leq k $ and all $ 1 \leq j \leq m_i $ we have $ \left( \ic{e_{i, j}} \cup \ac{\overline{t_{i - 1}}}{e_{i, j}} \right) \subseteq \overline{t_{i - 1}} $ (because of \eqref{eq:D2}, \eqref{eq:D3}, and, by \eqref{eq:D4}, $ \ac{\overline{t_{i - 1}} \cup X_i}{e_{i, j}} = \ac{\overline{t_{i - 1}}}{e_{i, j}} $), and $ \overline{t} = C $ (because $ X_n = C $).
	Thus $ C \in \configTraces{\gamma} $.
\end{proof}

\section{Dynamic Causality}

To compare DCESs to other ESs we define the Single State Dynamic Causality ESs (SSDCs) as a subclass of DCESs in that no modifier can add and drop the same dependency.

\begin{defi}
	\label{def:SingleStateDC}
	Let SSDC be a subclass of DCESs such that $\varrho$ is a SSDC iff
\[\dropped{\emptyset,E}\cap\added{\emptyset, E}=\emptyset.\]
\end{defi}

Since there are no adders and droppers for the same causal dependency, the order of modifiers does not matter and thus there are no two different states sharing the same configuration, \ie each configuration represents a state. Thus it is enough for SSDC to consider transition equivalence with respect to configurations, \ie $\transEqOp$.

\begin{lem}
	\label{lma:SingleCausalState}
	Let $\varrho =(E,\incaus,\shrinkingCausality,\growingCausality)$ be a SSDC. Then for the current causality relation $\caus{X}$ of any
state $(X, \caus{X}) \in \reachables{\varrho}$ it holds \[\caus{X} = (\incaus\setminus\dropped{\emptyset,X})\cup\added{\emptyset, X}.\]
\end{lem}

\begin{proof}
	If $ X = \emptyset $ then, by the Definitions~\ref{def:SEStrans} and \ref{def:GEStrans}, $ \dropped{\emptyset,X} = \emptyset $ and $ \added{\emptyset, X} = \emptyset $. Thus $ \caus{\emptyset} \; = \; \incaus $.
	
	Assume $ \transDC{(U,\caus{U})}{(X,\caus{X})} $. By induction, we have $ \caus{U} \; = (\incaus \setminus \dropped{\emptyset,U}) \cup \added{\emptyset, U} $. By Condition~\ref{eq:correctCausalityUpdate} in Definition~\ref{def:DCTransition}, we have $ \caus{X} \; = (\caus{U} \setminus \dropped{U,X}) \cup \added{U,X} $. By combining these two equations we obtain $ \caus{X} \; = ((\incaus \setminus \dropped{\emptyset,U}) \cup \added{\emptyset, U}) \setminus \dropped{U,X}) \cup \added{U,X} $. But, since $\varrho$ is a SSDC, this can be reordered to $ \caus{X} = (\incaus \setminus (\dropped{\emptyset,U} \cup \dropped{U,X} )) \cup (\added{\emptyset, U} \cup \added{U,X}) $ and simplified to $ \caus{X} = (\incaus \setminus \dropped{\emptyset,X}) \cup \added{\emptyset,X} $.
\end{proof}

In SSDC Condition~\ref{eq:correctCausalityUpdate} holds whenever $ X \subset Y $.

\begin{lem}
	\label{lma:SSDCstateProp}
	Let $\varrho=(E,\incaus,\shrinkingCausality,\growingCausality) $ be a SSDC and let $(X,\caus{X})$ and $(Y,\caus{Y})$ be two states of $\varrho$ with $X\subset Y$ then Condition~\ref{eq:correctCausalityUpdate} of Definition~\ref{def:DCTransition} holds for those two states.
\end{lem}

\begin{proof}
	Condition~\ref{eq:correctCausalityUpdate} of Definition~\ref{def:DCTransition} states that $ \caus{Y} = (\caus{X} \setminus \dropped{X,Y}) \cup \added{X, Y} $ but since $ \caus{X} = (\incaus \setminus \dropped{\emptyset,X}) \cup \added{\emptyset, X} $ by Lemma~\ref{lma:SingleCausalState}, we have to show that $ \caus{Y} = ((\incaus \setminus \dropped{\emptyset,X}) \cup \added{\emptyset, X}) \setminus \dropped{X,Y}) \cup \added{X, Y} $. Because $\varrho$ is a SSDC, this could be reordered and simplified to $ \caus{Y} = (\incaus \setminus \dropped{\emptyset,Y}) \cup \added{\emptyset, Y} $ that holds by Lemma~\ref{lma:SingleCausalState}.
\end{proof}

\begin{lem}
	\label{lma:SSDCConfInTrans}
	Let $\varrho=(E,\incaus,\shrinkingCausality,\growingCausality) $ be a SSDC and $\transDC{(X,\caus{X})}{(Y,\caus{Y})}$ a transition in $\varrho$. Then for all $X',Y'$ with $X\subseteq X'\subset Y'\subseteq Y$ there is a transition $\transDC{(X',\caus{X'})}{(Y',\caus{Y'})}$ in $\varrho$, where $ \caus{X'} \; = (\incaus \setminus \dropped{\emptyset,X'}) \cup \added{\emptyset, X'} $ and $ \caus{Y'}\; = (\incaus \setminus \dropped{\emptyset,Y'}) \cup \added{\emptyset, Y'} $.
\end{lem}

\begin{proof}
	Condition~\ref{eq:ConfigContainement} of Definition~\ref{def:DCTransition} holds by assumption. Condition~\ref{eq:OnlyEnabledEvents} holds since $(Y'\setminus X') \subseteq (Y\setminus X) $ and $ X \subseteq X' $. For Condition~\ref{eq:correctCausalityUpdate} we show $ \caus{Y'} \; = (\caus{X'} \setminus \dropped{X',Y'}) \cup \added{X', Y'} $ by assumption it follows $ \caus{Y'} \; = ((\incaus \setminus \dropped{\emptyset,X'}) \cup \added{\emptyset, X'}) \setminus \dropped{X',Y'}) \cup \added{X', Y'} $ and since $\varrho$ is a SSDC we can reorder and simplify to $ \caus{Y'} \; = (\incaus \setminus \dropped{\emptyset,Y'}) \cup \added{\emptyset, Y'} $that holds by assumption. Condition~\ref{eq:NCNoConcurrency} holds because $ (\added{X',Y'} \cap \dropped{X',Y'}) = \emptyset $ since $ \varrho $ is a SSDC. Finally Condition~\ref{eq:NotToMuchConcurrency} holds since $ (Y' \setminus X') \subseteq (Y \setminus X) $ and $ X \subseteq X' $.
\end{proof}

To compare DCESs with EBESs, we define a another sub-class of DCESs.

\begin{defi}
	\label{def:EBDC}
	Let EBDC denotes a subclass of SSDC with the additional requirements:
	\begin{enumerate}
		\item \label{eq:EBDCOnlyDisabling} $\forall c,a,t \in E \logdot		\addcause{a}{c}{t} \Longrightarrow c = t$
		\item \label{eq:EBDCOnlyDisablingII} $\forall c,d,t \in E \logdot		\drops{d}{c}{t} \Longrightarrow c \neq t$
		\item \label{eq:NoCausalAmbiguity} $\forall c,d_1,\ldots , d_n, t \in E \logdot \drops{d_1}{c}{t}\wedge\dots\wedge\drops{d_n}{c}{t} \Longrightarrow\\ \forall a,b \in \{c,d_1,\ldots , d_n\} \logdot \left(a \neq b \Longrightarrow		\adds{a}{b}{b}\in\growingCausality\wedge\adds{b}{a}{a}\in\growingCausality\right)$
	\end{enumerate} 
	%Note that the binary conflict $\conf{a}{b}$ in Condition~\ref{eq:NoCausalAmbiguity} is only syntactic sugar and is encoded with a fresh impossible event $i$ (wiht $\enab{i}{i}$) and the mutual disablings $\adds{i}{a}{b}$ and $\adds{i}{b}{a}$, which do not fit in Condition~\ref{eq:EBDCOnlyDisabling}.
\end{defi}

The first condition translates disabling into $ \growingCausality $ and the second ensures that disabled events cannot be enabled again. The third condition reflects causal unambiguity by $ \shrinkingCausality $ such that either the initial cause or one of its droppers can happen. 

We adapt the notion of precedence.

\begin{defi}
	\label{def:EBDCLposets}
	Let $\vartheta$ be a EBDC and $X \in \configurations{\vartheta}$ we define the precedence relation $<_X \subseteq X\times X$ as $e <_X e'\Longleftrightarrow (e \rightarrow e') \lor (\addcause{e'}{e}{e}) \lor (\exists c\in E\logdot \drops{e}{c}{e'})$. Let $\leq_X$ be the reflexive and transitive closure of $<_X$.
\end{defi}

The relation $<_X$ indeed represents a precedence relation, and its reflexive transitive closure is a partial order.

\begin{lem}
	\label{lma:EBESPrecedence}
	Let $\vartheta=(E,\enabOp,\shrinkingCausality,\growingCausality)$ be a EBDC, $X \in \configurations{\vartheta}$, and let $e,e' \in X \logdot e<_X e'$. Let also $ \transDC{\transDC{(X_0,\caus{X_0})}{\ldots}}{(X_n, \caus{X_n})}$ with $ X_0 = \emptyset $ and $ X_n = X $ be the transition sequence of $X$ then $\exists X_i \in \{X_0,\ldots,X_n\} \logdot e \in X_i \land e' \notin X_i$.
\end{lem}

\begin{proof}
	Let $(X_f, \caus{X_f})$ be the first occurrence of $e$ in the sequence $ \transDC{\transDC{(X_0,\caus{X_0})}{\ldots}}{(X_n, \caus{X_n})}$, so according to Condition~\ref{eq:ConfigContainement} of Definition~\ref{def:DCTransition} it is enough to prove that $e' \notin X_f$.
	
	First, assume that $e \rightarrow e'$, then $(e,e') \in \incaus$. Then, by Definition~\ref{def:DCTransition}, to obtain $(e,e') \notin \caus{X_{f-1}}$ we need a dropper $d \in X_{f-1}$ for $e \rightarrow e'$ (according to Condition~\ref{eq:correctCausalityUpdate}).
	But that is impossible, since $e$ and $d$ will be mutually disabling each other, because of Condition~\ref{eq:NoCausalAmbiguity} of Definition~\ref{def:EBDC}.
	So $(e,e') \in \caus{X_{f-1}}$ and thus $e'\notin X_f$ because of Condition~\ref{eq:OnlyEnabledEvents} of Definition~\ref{def:DCTransition}.
	
	Second, assume that $\addcause{e'}{e}{e}$.
	If $e' \in X_{f}$ then, by Condition~\ref{eq:NotToMuchConcurrency} of Definition~\ref{def:DCTransition}, it follows $e'\in X_{f-1}$.
	Let $f'$ be minimal with $e'\in X_{f'}$ by Condition~\ref{eq:correctCausalityUpdate} we have $(e,e)\in\caus{X_{f'}}$, and---since $\vartheta$ is a EBDC and therefore a SSDC---there is no dropper for $e \rightarrow e$ and we have $(e,e)\in \caus{X_{f-1}}$. Thus, by Condition~\ref{eq:OnlyEnabledEvents}, $e\notin X_f$ contradicting our definition, and therefore the assumption $e' \in X_{f}$ must be wrong.
	
	Third, assume $\exists c\in E\logdot \drops{e}{c}{e'}$. Then since EBDC are a subclass of SSDC we have $\nexists a \in E \logdot \addcause{a}{c}{e'}$ according to Definition~\ref{def:SingleStateDC}. Then $c \rightarrow e'$ by Condition~\ref{eq:DCESDroppingExistentCauses} of Definition~\ref{def:DCES} that means $(c,e') \in \incaus$. Let us assume that $e' \in X_f$ then either $c$ or another dropper $d$ with $\drops{d}{c}{e'}$ occurred before $e'$ that is impossible because of the mutual disabling in Condition~\ref{eq:NoCausalAmbiguity} of Definition~\ref{def:EBDC}. So $e' \notin X_f$.
\end{proof}

\begin{lem}
	\label{lma:EBESPrecIsOrder}
	$\leq_X$ is a partial order over $X$. 
\end{lem}

\begin{proof}
	Let $e,e' \in X \logdot e <_X e'$ and let $(\emptyset = X_0, \caus{X_0}) \ldots(X_n = X, \caus{X_n})$ be the transition sequence of $X$.
	Let $X_h$ ($X_j$) be the configurations where $e$ (respectively $e'$) occurs first.
	By Lemma~\ref{lma:EBESPrecedence} then $h<j$.
	Since $\leq_X$ is the reflexive and transitive closure of $<_X$ we have $e \leq_X e' \Longrightarrow h \leq j$.
	For anti-symmetry, assume that $e' \leq_X e$. By Lemma~\ref{lma:EBESPrecedence}, we have $j \leq h$, but $h \leq j$. Hence $h = j$.
	The equality $h = j$ implies that $e = e'$ because otherwise $h < j$ and $j<h$ that is a contradiction.
\end{proof}

Let $ \posets{\vartheta} = \{(X, \leq_X) \mid X \in \configurations{\vartheta}\} $ denote the set of posets of the EBDC $\vartheta$. We show that the transitions of a EBDC $\vartheta$ can be extracted from its posets.

\begin{thm}
	\label{thm:EBDCTrFromPosets}
	Let $\vartheta=(E,\incaus,\shrinkingCausality,\growingCausality)$ be a EBDC and $(X,\caus{X}), (Y,\caus{Y}) \in \reachables{\vartheta}$ with $X\subset Y$. Then \[\left(\forall e,e'\in Y\logdot e\neq e'\land e'\leq_{Y} e\implies e'\in
X \right)\Longleftrightarrow\transDC{(X,\caus{X})}{(Y,\caus{Y})}.\]
\end{thm}

\begin{proof}
	We assume $\left(\forall e,e'\in Y\logdot e\neq e'\land e'\leq_{Y} e\implies e'\in X \right)$.
	Condition~\ref{eq:ConfigContainement} holds by assumption.
	For Condition~\ref{eq:OnlyEnabledEvents} we have to show that $\forall e\in Y\setminus X\logdot \Set{e'\mid (e',e)\in\caus{X}}\subseteq X$.
	By Lemma~\ref{lma:SingleCausalState}, we have  $\caus{X}=(\incaus\setminus\dropped{\emptyset,X})\cup\added{\emptyset,X}$ for all $(e',e)\in\incaus$.
	\begin{itemize}[noitemsep]
		\item The inclusion $\Set{e'\mid (e', e) \in (\incaus\setminus\dropped{\emptyset,X})} \subseteq X$ holds by assumption and because $e'\leq_Y e$ implies $e'\in X$.
		\item For all $(e',e')\in\added{\emptyset,X}$ of the form $\adds{e}{e'}{e'}$ we have $e'\leq_Y e$ and thus $e'\notin Y\setminus X$.
	\end{itemize}
	For Condition~\ref{eq:correctCausalityUpdate} we have to show $\caus{Y}=(\caus{X}\setminus\dropped{X,Y})\cup\added{X,Y}$.
	By Lemma~\ref{lma:SingleCausalState} we have $\caus{X}=(\incaus\setminus\dropped{\emptyset,X})\cup\added{\emptyset,X}$.
	We substitute $\caus{X}$ in $\caus{Y}=(\caus{X}\setminus\dropped{X,Y})\cup\added{X,Y}$ by $\caus{X}=(\incaus\setminus\dropped{\emptyset,X})\cup\added{\emptyset,X}$.
	This reduces to $\caus{Y}=(\incaus\setminus\dropped{\emptyset,Y})\cup\added{\emptyset,Y}$ that holds again by Lemma~\ref{lma:SingleCausalState}.
	Condition~\ref{eq:NCNoConcurrency} holds trivially because $\vartheta$ is a SSDC and thus $\dropped{\emptyset,E}\cap\added{\emptyset, E}=\emptyset$.
	Consider now $\adds{e}{e'}{e'}$. If $e,e' \in Y$, we have $e'\leq_Y e$. Therefore by assumption $e' \in X$ that fulfills Condition~\ref{eq:NotToMuchConcurrency}.
	
	Let us now assume $\transDC{(X,\caus{X})}{(Y,\caus{Y})}$, and $e,e'\in Y$ with $e\neq e'$ and $e'\leq_{Y} e$ by Lemma~\ref{lma:EBESPrecedence} it follows $e'\in X$.
\end{proof}

The following defines an encoding of an EBES into an EBDC. Furthermore the encoding preserves posets. Figure~\ref{fig:exampleEBES2DCES} provides an example where conflicts with impossible events are dropped for simplicity.

\begin{defi}
	\label{def:EBES2DCES}
	Let $\xi = \left( E, \disaOpT,\buEnOpT, l \right)$ be an EBES, $ \Set{ X_i }_{i \in I} $ an enumeration of its bundles, and $ \Set{ x_i }_{i \in I}$ a set of fresh events, \ie $ \Set{ x_i }_{i \in I} \cap E = \emptyset $.
	
	Then $\dces{\xi} = (E', \incaus, \shrinkingCausality,\growingCausality)$ such that:
	\begin{enumerate}
		\item $ E' = E \cup \Set{ x_i }_{i \in I}$
		\item $\incaus=\Set{ \enab{x_i}{e} \mid \buEn{X_i}{e} } \cup \Set{ \enab{x_i}{x_i} \mid i \in I } $
		\item \label{eq:EBES2DCESShrinking} $\shrinkingCausality \Set{ \drops{d}{x_i}{e} \mid d \in X_i \land \buEn{X_i}{e} }$
		\item \label{eq:EBES2DCESGrowing} $\growingCausality = \{\adds{e'}{e}{e} \mid \disa{e}{e'} \}$.
	\end{enumerate}
\end{defi}

\begin{figure}[tb]
	\centering
	\begin{tikzpicture}
		% figure a
		\event{a1}{-1}{1}{above}{$ a $};
		\event{a2}{-0.2}{1}{above}{$ b $};
		\event{a3}{-0.2}{0}{below left}{$ c $};
		\event{a5}{.6}{0}{right}{$ d $};
		\draw[enablingPES] (a1) edge (a3);
		\draw[enablingPES] (a2) edge (a3);
		\draw[enablingPES] (a3) .. controls (-0.9, 0.3) .. (a1);
		\draw[thick]	(-0.5, 0.37) -- (-0.2, 0.5);
		\draw[conflictPES] (a1) edge (a2);
		\draw[conflictEBES] (a3) edge (a5);
		\node (a) at (-1.7, 1) {(a)};
		% figure b
		\event{b1}{3}{1}{above}{$ a $};
		\event{b2}{3.8}{1}{above}{$ b $};
		\event{b3}{4.3}{0.7}{below right}{$ e_1 $};
		\event{b4}{3.8}{0}{below left}{$ c $};
		\event{b6}{5}{0}{below}{$ d $};
		\event{b5}{3}{0}{left}{$ e_2 $};
		\draw[enablingPES] (b5) edge (b1);
		\draw[enablingPES] (b3) edge (b4);
		\draw[enablingPES] (b3) .. controls (4.8, 1.2) and (3.8, 1.2) .. (b3);
		\draw[enablingPES] (b4) .. controls (4.3, .3) and (4.3, -.5) .. (b4);
		\draw[enablingPES] (b5) .. controls (2.5, -.5) and (3.5, -.5) .. (b5);
		\draw[adding]		(b6) -- (4.2 , 0);
		\draw[conflictPES]  (b1) edge (b2);
		\draw[dropping]		(3.1, .4) -- (b4);
		\draw[dropping]		(4, .4) -- (b1);
		\draw[dropping]		(4, .4) -- (b2);
		\node (b) at (2.3, 1) {(b)};
	\end{tikzpicture}
	\caption{An EBES and its poset-equivalent DCES.}
	\label{fig:exampleEBES2DCES}
\end{figure}
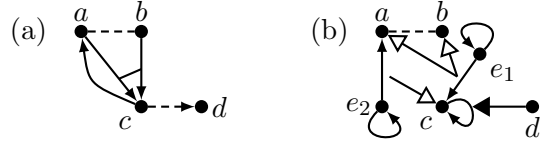

We show that the encoding yields an EBDC.

\begin{lem}
	\label{lma:EmdeddingIsEBDC}
	Let $\xi$ be an EBES. Then $\dces{\xi}$ is an EBDC.
\end{lem}

\begin{proof}
	First $\dces{\xi}$ is a DCES. The Definitions of $\incaus,\shrinkingCausality$, and $\growingCausality$ in Definition~\ref{def:EBES2DCES} ensure the Conditions~\ref{def:DCES}.\ref{eq:DCESDroppingExistentCauses} and \ref{def:DCES}.\ref{eq:DCESAddingMissingCauses}. According to the definition of $\shrinkingCausality$ in Condition~\ref{def:EBES2DCES}.\ref{eq:EBES2DCESShrinking} all dropped causes are the fresh events that cannot be added by $\growingCausality$ (Condition~\ref{def:EBES2DCES}.\ref{eq:EBES2DCESGrowing}). So Condition~\ref{def:DCES}.\ref{eq:ModifierCannotAddAndDropTheSameTarget} also holds.

	Second, $\dces{\xi} $ is a SSDC since all dropped events are fresh and because these fresh events are never added by $\growingCausality$, the conditions of Definition~\ref{def:SingleStateDC} are satisfied.
	
	Third, $\dces{\xi}$ is a EBDC. The Conditions~\ref{def:EBDC}.\ref{eq:EBDCOnlyDisabling} and \ref{eq:EBDCOnlyDisablingII} hold by definition. Bundle members in $\xi$ mutually disable each other. Then according to the Condition~\ref{def:EBES2DCES}.\ref{eq:EBES2DCESGrowing}, the Condition~\ref{def:EBDC}.\ref{eq:NoCausalAmbiguity} holds. Therefore $\dces{\xi} $ is a EBDC.
\end{proof}

Before comparing an EBES with its encoding according to posets, we show that they have the same configurations.

\begin{lem}
	\label{thm:EBESandEBDCconfigEqui}
	Let $\xi = \left(E, \disaOpT, \buEnOpT\right)$ be an EBES. Then $\configurations{\xi} = \configurations{\dces{\xi}}$.
\end{lem}

\begin{proof}
	First, $\forall X \subseteq E \logdot X\in \configurations{\xi}\Longrightarrow X\in \configurations{\dces{\xi}}$.
	
	According to \S\ref{sec:EBES}, $X\in \configurations{\xi}$ means there is a trace $t = e_1,\ldots,e_n$ in $\xi$ such that $X = \bar{t}$. Let us prove that $t$ corresponds to a transition sequence in $\dces{\xi}$ leading to $X$. \Ie let us prove that there exists a transition sequence $\transDC{(X_0, \caus{X_0})} {\ldots} \transDC {}{(X_n,\caus{X_n})}$ with $X_0=\emptyset,\caus{X_0}=\incaus$ and $X_n=X$ such that $X_i = X_{i-1} \cup \Set{e_i}$ for $1 \leq i \leq n$, and $\caus{X_i}$ is defined according to Lemma~\ref{lma:SingleCausalState}. This means we have to prove that $\transDC{(X_{i-1}, \caus{X_{i-1}})} {(X_{i}, \caus{X_{i}})}$ for $1 \leq i \leq
n$.
	
	Since $X_i = X_{i-1} \cup \Set{e_i}$, we have $X_{i-1}\subset X_i$ and so Condition~\ref{eq:ConfigContainement} in Definition~\ref{def:DCTransition} holds.
	Next, let us prove that $\forall e\in X_i\setminus X_{i-1} \logdot  \{e'\mid (e',e)\in\caus{X_{i-1}}\}\subset X_{i-1}$, \ie $\{e'\mid (e',e)\in(\incaus\setminus\dropped{\emptyset,X_{i-1}})\cup\added{\emptyset,X_{i-1}}\} \subseteq X_{i-1}$.
	By Lemma~\ref{lma:SingleCausalState} and the Definition of $\dces{\xi}$, $\incaus$ contains only fresh (and impossible) events as causes and the elements $x \in X_i$ of bundles ${X_i}\buEnOpT{e}$ are droppers of these fresh events.
	But since each of these bundles is satisfied, each of these fresh events in $\incaus$ is dropped. Furthermore, there cannot be added causality in $\dces{\xi}$ for $e$, except disabling of $e$, but this is not possible since it occurs in a configuration. Therefore $\{e'\mid (e',e)\in(\incaus\setminus\dropped{\emptyset,X_{i-1}})\cup\added{\emptyset,X_{i-1}}\} \subseteq X_{i-1}$ for all $e \in X_i$ and all $1\leq i \leq n$.
	Condition~\ref{eq:correctCausalityUpdate} follows from Lemma~\ref{lma:SSDCstateProp}.
	Condition~\ref{eq:NCNoConcurrency} of Definition~\ref{def:DCTransition} holds by Definition~\ref{def:SingleStateDC}.
	Since in the transition $\transDC {(X_{i-1}, \caus{X_{i-1}})} {(X_i, \caus{X_i})}$ only one event---namely $e_i$---occurs, the last Condition~\ref{eq:NotToMuchConcurrency} is satisfied.
	
	On the other hand let $X\in\configurations{\dces{\xi}}$ and $\transDC{(X_0, \caus{X_0})} {\ldots} \transDC {}{(X_n,\caus{X_n})}$ such that $X_0=\emptyset, \caus{X_0}=\incaus$.
	By Lemma~\ref{lma:SSDCConfInTrans}, we can assume $e_{i+1}:=X_{i+1}\setminus X_i$.
	We have to show that $t = e_1,\ldots,e_n$ is a trace in $\xi$.
	By Definition~\ref{def:EBESconf}, this means $ \forall 1 \leq i, j \leq n \logdot \disa{e_i}{e_j} \implies i < j $ and $ \forall 1 \leq i \leq n \logdot \forall Y \subseteq E \logdot \buEn{Y}{e_i} \implies \overline{t_{i - 1}} \cap Y \neq \emptyset $.
	Assume $\disa{e_i}{e_j}$.
	By Definition~\ref{def:EBES2DCES}, we have $\adds{e_j}{e_i}{e_i}$.
	By Definition~\ref{def:EBDCLposets} and Lemma~\ref{lma:EBESPrecedence}, then $e_i<_X e_j$ and $\exists i'<n\logdot e_i\in X_{i'} \wedge e_j\notin X_{i'}$.
	But, since $e_j\in X$, the first Condition holds.
	Next assume $\buEn{Y}{e_i}$.
	By Definition~\ref{def:EBES2DCES}, there is a fresh and impossible cause $x_i$ for $e_i$ in the initial causality ($(x_i,e_i),(x_i,x_i)\in\incaus$), there is no dropper for $(x_i,x_i)$, and the elements of $Y$ are exactly the droppers for $(x_i,e_i)$ ($y\in Y\iff\drops{y}{x_i}{e_i}$).
	Since $e_i\in X_i$ and by Condition~\ref{eq:OnlyEnabledEvents} of Definition~\ref{def:DCTransition}, it follows $(x_i,e_i)\notin \caus{X_{i-1}}$ (since $x_i$ never becomes enabled).
	By Condition~\ref{eq:correctCausalityUpdate}, there must be a $j<i$ such that $\drops{e_j}{x_i}{e_i}$ (since $(x_i,e_i)\in\incaus$), thus we have $e_j\in\overline{t_{i - 1}} \cap Y$ and we are done.
\end{proof}

Finally we show that the encoding preserves posets.

\begin{lem}
	\label{lma:EBESintoEBDC}
	For each EBES $\xi$ there is a DCES, namely $\dces{\xi}$, such that $\posets{\xi} = \posets{\dces{\xi}}$.
\end{lem}

\begin{proof}
	Let $p = (X,\leq_X)$, then $X\in \configurations{\xi}$ by the definition of posets of EBESs. Then according to Lemma~\ref{thm:EBESandEBDCconfigEqui}: $X \in\configurations{\dces{\xi}}$. On the other hand, let $\leq_X'$ be the partial order defined for $X$ in $\dces{\xi}$ as in Definition~\ref{def:EBDCLposets}. This means that we should prove that $\leq_X = \leq_X'$. Since $\leq_X, \leq_X'$ are the reflexive and transitive closures of $\prec_X, <_X$ respectively, we have to prove that $\prec_X = <_X$. In other words we have to prove $\forall e,e'\in X \logdot e \prec_X e' \Leftrightarrow e'<_X e$.
	
	Let us start with $e \prec_X e'\Longrightarrow e<_X e'$. According to \S\ref{sec:EBES} $e \prec_X e'$ means $\exists Y\subseteq E \logdot e \in \buEn{Y}{e}' \lor \disa{e}{e'}$. If $\exists Y\subseteq E \logdot e \in \buEnOpT{Y}{e}'$ then $\exists c\in E' \logdot \drops{e}{c}{e'}$, where $E'$ is the set of events of $\dces{\xi}$, by the definition of $\dces{\xi}$ (Definition~\ref{def:EBES2DCES}). This means $e<_X e'$ according to the definition of $<_X$ (Definition~\ref{def:EBDCLposets}). If $\disa{e}{e'}$, then $ \neg \disa{e'}{e}$ because otherwise $e$ and $e'$ are in conflict. This means $\addcause{e'}{e}{e}$ according to Definition~\ref{def:EBES2DCES} that means $e<_X e'$ according to Definition~\ref{def:EBDCLposets}.
	
	Let us consider the other direction: $e <_X e'\Longrightarrow e\prec_X e'$. $e<_X e'$ means $\exists c\in E' \logdot \drops{e}{c}{e'} \lor \addcause{e'}{e}{e}$, where $E'$ is the set of events of $\dces{\xi}$, according to the definition of $<_X$ in Definition~\ref{def:EBDCLposets}. The third option, where $\enab{e}{e'}$, is rejected, because all initial causes in $\dces{\xi}$ are the fresh impossible events. If $\exists c\in E' \logdot \drops{e}{c}{e'}$ then $\exists Y \subseteq E \logdot \disa{e\in Y}{e'}$ according to the definition of $\shrinkingCausality$ in $\dces{\xi}$. This means $e \prec_X e'$ by the definition of $\prec_X$ in \S\ref{sec:EBES}. If on the other hand $\addcause{e'}{e}{e}$, then $\disa{e}{e'}$ according to the definition of $\dces{\xi}$ that means $e \prec_X$ in \S\ref{sec:EBES}.
	
	So we have $\prec_X = <_X$ that implies $\leq_X =\leq'_X$. 
\end{proof}

\end{document}